\documentclass[11pt]{article}
\usepackage{paper-ak}
\usepackage{wrapfig}

\newcommand{\FcOne}{\Fc([-1, 1], \R)}
\newcommand{\FcN}{\Fc(\N, \R)}
\DeclareMathOperator*{\epsMV}{\epsilon_{\textsf{MV}}}

\newcommand{\amv}{\mathrm{AMV}}

\newcommand{\bhat}{\hat{b}}
\newcommand{\Tbar}{\bar{T}}

\title{Sublinear Time Spectral Density Estimation}

\author{
	Vladimir Braverman \\ Johns Hopkins University\\ \texttt{vova@cs.jhu.edu}
		\and
	Aditya Krishnan \\ Johns Hopkins University \\ \texttt{akrish23@jhu.edu}
		\and 
	Christopher Musco\\ New York University\\ \texttt{cmusco@nyu.edu}
}

\date{\today}

  \usepackage{nth}
  \usepackage{intcalc}

\begin{document}
\pagenumbering{arabic}

\maketitle 

\begin{abstract}

	We present a new sublinear time algorithm for approximating the spectral density (eigenvalue distribution) of an $n\times n$ normalized graph adjacency or Laplacian matrix. The algorithm recovers the spectrum up to $\epsilon$ accuracy in the Wasserstein-1 distance in  $O(n\cdot \poly(1/\epsilon))$ time given sample access to the graph. This result compliments recent work, which obtains a solution with runtime independent of $n$, but exponential in $1/\epsilon$ \cite{Cohen-SteinerKongSohler:2018}. We conjecture that the trade-off between dimension dependence and accuracy is inherent.

	Our method is simple and works well experimentally. It is based on a Chebyshev polynomial moment matching method that employees randomized estimators for the matrix trace. We prove that, for any Hermitian $A$, this moment matching method returns an $\epsilon$ approximation to the spectral density using just $O({1}/{\epsilon})$ matrix-vector products with $A$. By leveraging stability properties of the Chebyshev polynomial three-term recurrence, we then prove that the method is amenable to the use of coarse approximate matrix-vector products. Our sublinear time algorithm follows from combining this result with a novel sampling algorithm for approximating matrix-vector products with a normalized graph adjacency matrix.

	Of independent interest, we show a similar result for the widely used \emph{kernel polynomial method} (KPM), proving that this practical algorithm nearly matches the theoretical guarantees of our moment matching method. Our analysis uses tools from Jackson's seminal work on approximation with positive polynomial kernels \cite{Jackson:1912}.

\end{abstract}

\section{Introduction}
A ubiquitous task in computational science, engineering, and data science is to extract information about the eigenvalue spectrum of a matrix $A\in \R^{n\times n}$. A full eigendecomposition takes at least $O(n^\omega)$ time\footnote{Here $\omega< 2.373$ is the fast matrix multiplication exponent.}, which is prohibitively expensive for large matrices \cite{Parlett:1998,BanksVargasKulkarni:2019}. So, we are typically interested in extracting \emph{partial information} about the spectrum. This can be done using iterative methods like the power or Lanczos methods, which access $A$ via a small number of matrix-vector multiplications. Each multiplication takes at most $O(n^2)$ time to compute, and can be accelerated when $A$ is sparse or structured, leading to fast algorithms. 

However, the partial spectral information computed by most iterative methods is limited. Algorithms typically only obtain accurate approximations to  the outlying, or \emph{largest magnitude} eigenvalues of $A$, missing information about the \emph{interior} of $A$'s spectrum that may be critical in applications.
For example, in network science, clusters of interior eigenvalues can indicate graph structures like repeated motifs \cite{DongBensonBindel:2019}. In deep learning, information on how the spectrum of a weight matrix differs from its random initialization can give hints about model convergence and generalization \cite{PenningtonSchoenholzGanguli:2018,MahoneyMartin:2019}, and Hessian eigenvalues are useful in optimization \cite{GhorbaniKrishnanXiao:2019}.
Coarse information about interior eigenvalues is also used to initialize parallel GPU based methods for full eigendecomposition \cite{AurentzKalantzisSaad:2017,LiXiErlandson:2019}. 

To address these needs and many other applications, there has been substantial interest in methods for estimating the full \emph{spectral density} of a matrix $A$ \cite{weisse2006kernel}. Concretely, assume that $A$ is Hermitian with real eigenvalues $\lambda_1, \ldots, \lambda_n$. We view its spectrum as a probability density $s$:
\begin{align}
	\label{eq:spect_density}
	&\text{Spectral density:} & s(x) &= \frac{1}{n}\sum_{i=1}^n \delta(x - \lambda_i).
\end{align}
Here $\delta$ is the Dirac delta function. The goal is to find a probability density $q$ that approximates $s$ in some natural metric, like the Wasserstein distance. The density $q$ can either be continuous (represented in some closed form) or discrete (represented as a list of approximate eigenvalues $\tilde{\lambda}_1,\ldots, \tilde{\lambda}_n)$. See Figure \ref{fig:sde_illustration} for an illustration. Both sorts of approximation are useful in applications.

\begin{figure}[h]
	\centering
	\includegraphics[width=.32\linewidth]{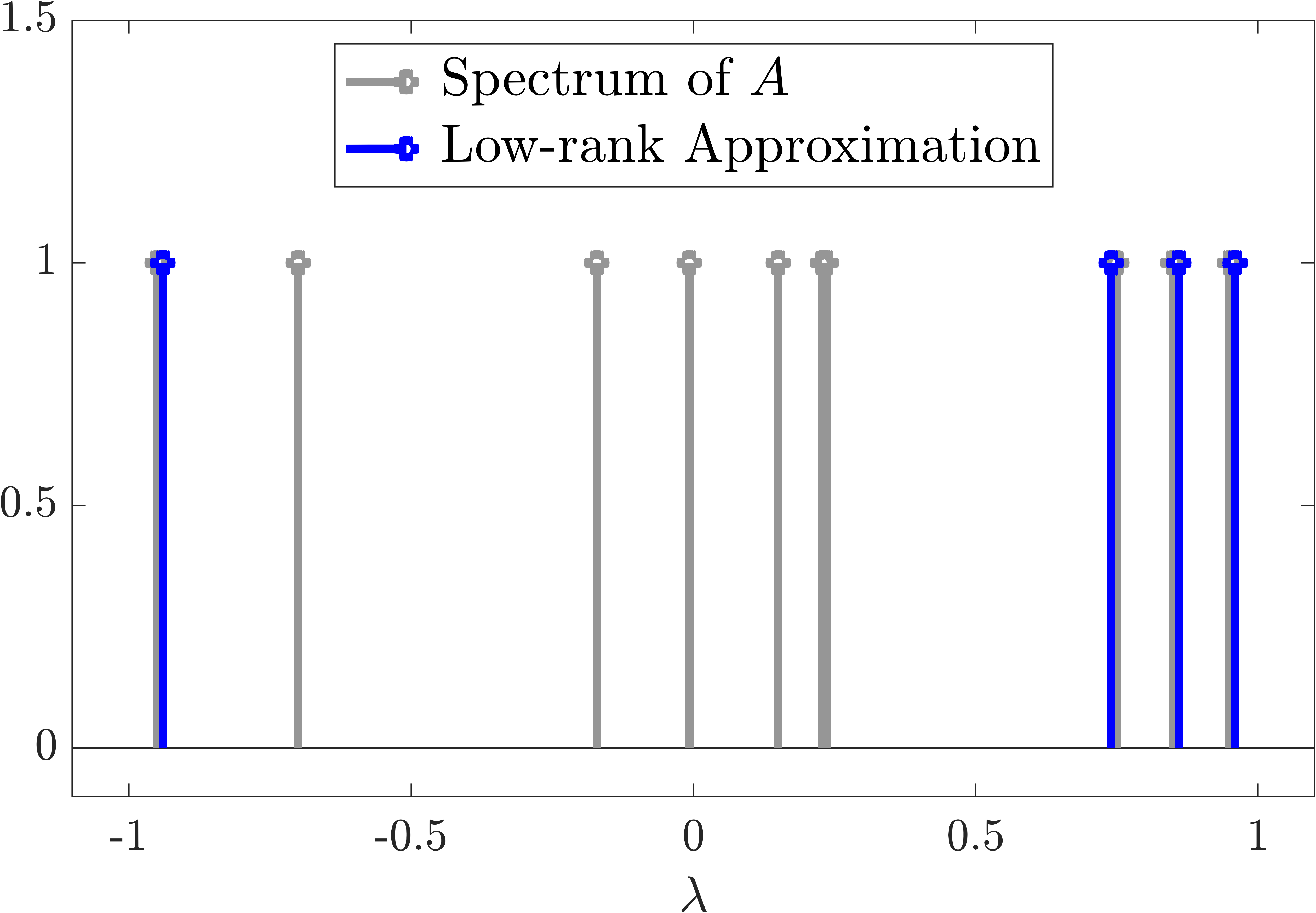}  \hfill
	\includegraphics[width=.32\linewidth]{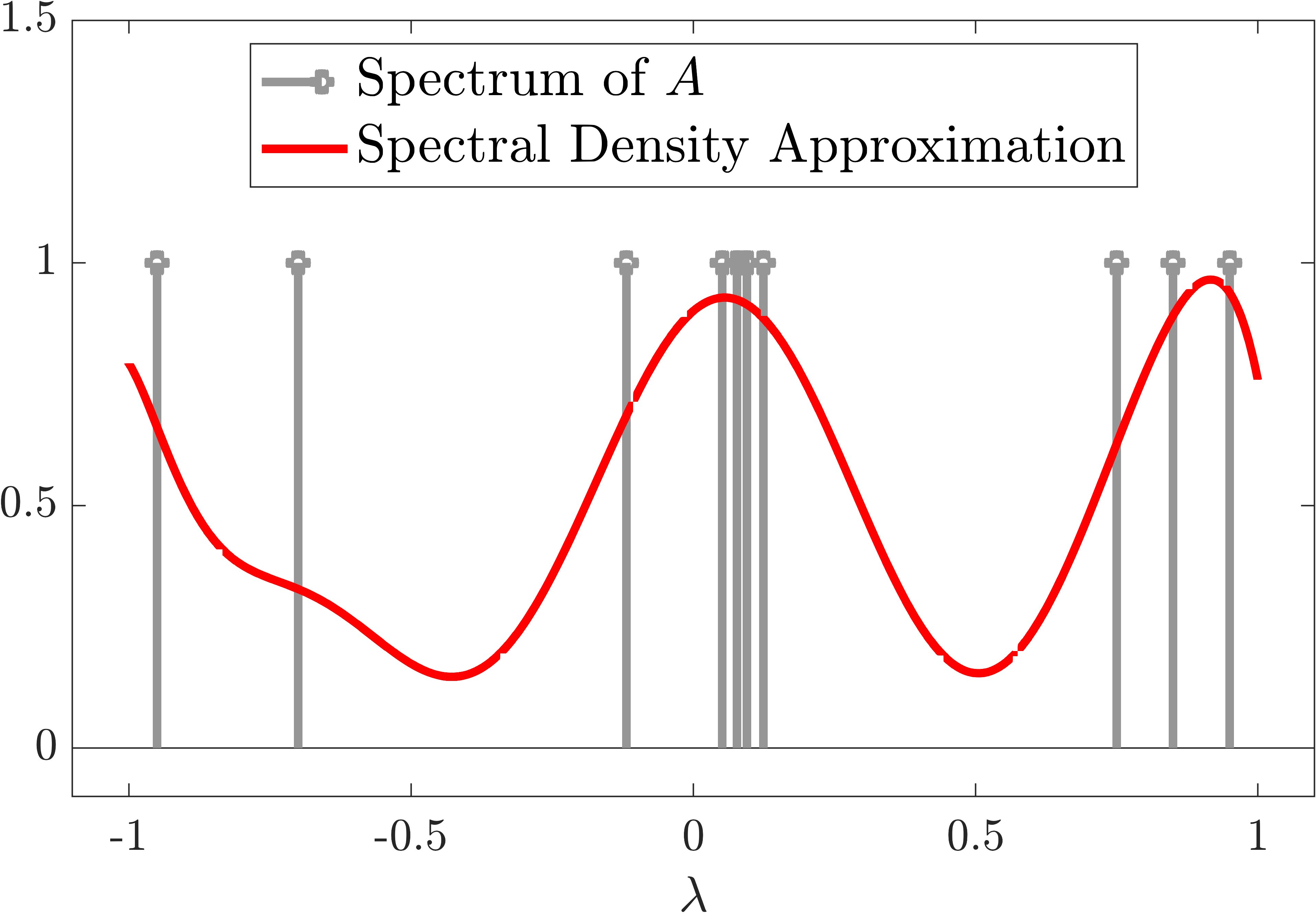} \hfill \includegraphics[width=.32\linewidth]{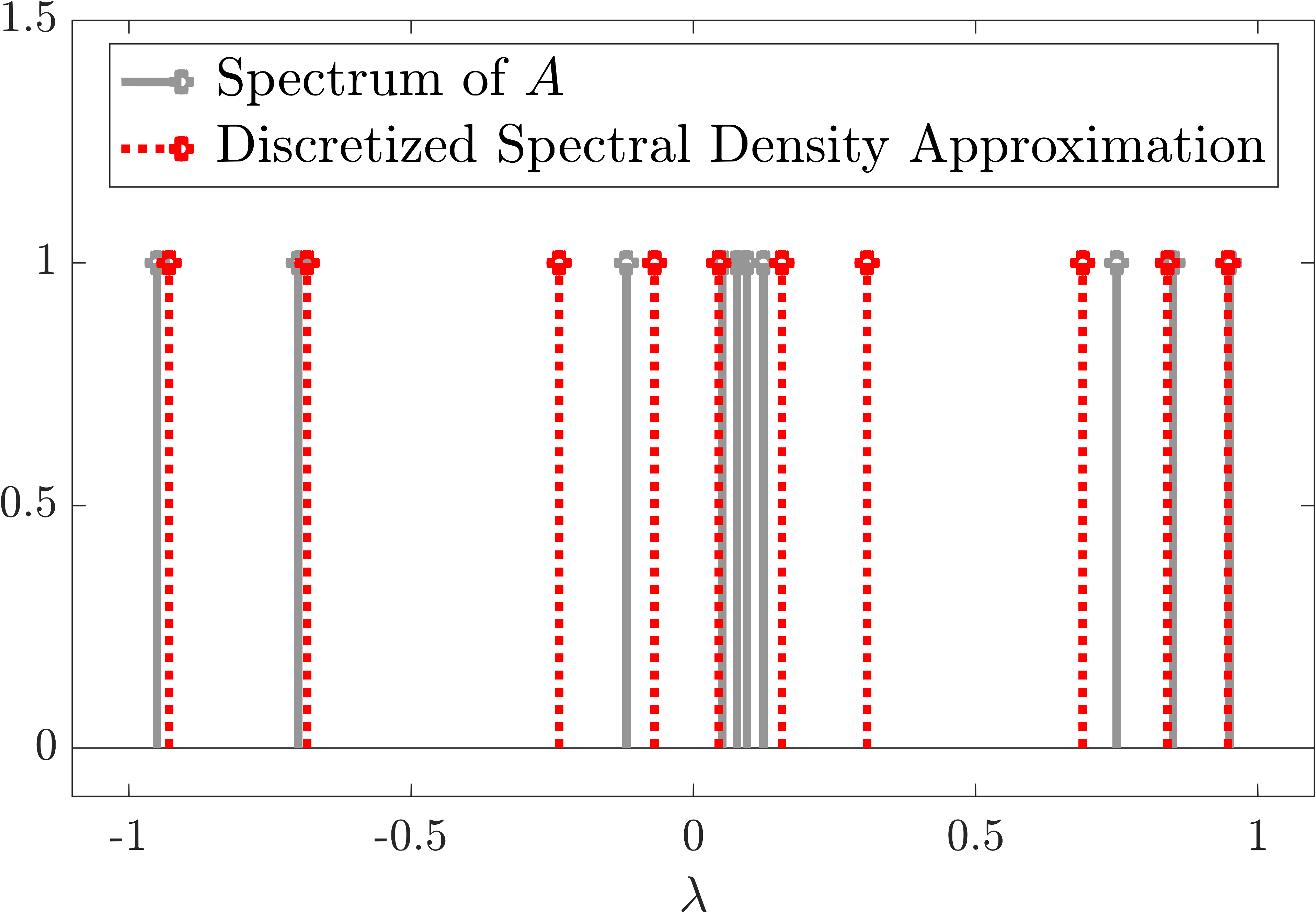} 
	\vspace{-.5em}
	\caption{Different approximations for the spectrum of a matrix $A$ with eigenvalues in $[-1,1]$. A typical approximation computed using an iterative eigenvalue algorithm mostly preserves information about the largest magnitude eigenvalues. In contrast, the spectral density estimates in the two right figures coarsely approximate the entire distribution of $A$'s eigenvalues, the first with a low-degree polynomial, and the second with a discrete distribution.
	}
	\label{fig:sde_illustration}
\end{figure} 

Methods for spectral density estimation that run in $o(n^\omega)$ time were first introduced for applications in condensed matter physics and quantum chemistry \cite{Skilling:1989,SilverRoder:1994,Wang:1994}. Many are based on the combination of two important tools: 1) moment matching, and 2) stochastic trace estimation. Specifically, if we had access to moments of the distribution $s$, i.e. $\frac{1}{n}\sum_{i=1}^n \lambda_i$, $\frac{1}{n}\sum_{i=1}^n \lambda_i^2$, $\frac{1}{n}\sum_{i=1}^n \lambda_i^3$, etc., then we could find a good approximation $q$ by finding a distribution that agrees with $s$ on these moments. Moreover, these \emph{spectral moments} can be computed via the matrix trace: note that $\tr(A) = \sum_{i=1}^n\lambda_i$, $\tr(A^2) = \sum_{i=1}^n\lambda_i^2$, $\tr(A^3) = \sum_{i=1}^n\lambda_i^3$, etc. While we cannot hope to compute $\tr(A^k)$ exactly in $o(n^\omega)$ time, thanks to {stochastic trace estimators} like Hutchinson's method, this trace can be approximated much more quickly \cite{Hutchinson:1990,AvronToledo:2011}. Such estimators are based on the observation that, for any matrix $B\in \R^{n\times n}$,  $\tr(B)$ can be well approximated by $\tr(G^TBG)$ where $G\in \R^{n\times m}$ contains random sub-Gaussian entries and $m \ll n$. 
%is independent of $n$. Specifically, $m = O(1/\epsilon^2)$ for an accuracy parameter $\epsilon$. 
For any $k$ degree polynomial $g$, $G^T g(A)G$ can be computed with just $O(km)$ matrix-vector multiplications, so we can quickly approximate any low-degree moment of $A$'s spectral density.

While this high-level approach and related techniques have been applied successfully to estimating the spectra of a wide variety of matrices \cite{weisse2006kernel,LinSaadYang:2016}, theoretical guarantees have only appeared relatively recently. Perhaps surprisingly, it can be shown that many common methods provably run in \emph{linear time} for any Hermitian matrix $A$. For instance, in work concurrent to ours, Chen, Trogdan, and Ubaru \cite{ChenTrogdonUbaru:2021} show that for any $n\times n$ Hermitian matrix $A$ with spectral density $s$, the popular Stochastic Lanczos Quadrature (SLQ) method provably computes an approximate spectral density $q$ satisfying: 
\begin{align}
	\label{eq:gaurantee}
	W_1(s, q) \leq \epsilon
\end{align}
using just $\poly(1/\epsilon)$ matrix-vector multiplications with $A$. Above $W_1$ denotes the Wasserstein-1 distance, aka the ``earth-movers distance''.\footnote{We assume $\|A\|_2 \leq 1$ for simplicity of stating errror guarantees, noting that Wasserstein distance is not scale invariant. This assumption is without loss of generality since $\|A\|_2$ can always be scaled after computing the top eigenvector up to constant fact accuracy, which takes just $O(\log n)$ matrix-vector multiplications \cite{MuscoMusco:2015}.} 
We defer a formal definition of $W_1$ to Section \ref{sec:Preliminaries}. The measure is convenient because, unlike many other measures of statistical distance, it allows a discrete distribution like the spectral density to be meaningfully compared to a possibly continuous approximation.  For discrete approximations, the Wasserstein distance is related to a simple $\ell_1$ metric. If we let $\Lambda = [\lambda_1, \ldots, \lambda_n]$ be a vector of $A$'s eigenvalues and $\tilde{\Lambda} = [\tilde{\lambda}_1, \ldots, \tilde{\lambda}_n]$ be a vector of approximate eigenvalues, then $\|\Lambda - \tilde{\Lambda}\|_1 \leq n\epsilon$ if and only if $W_1(s, q) \leq \epsilon$ for the discrete spectral density $q$ with eigenvalues in $\tilde{\Lambda}$. 

As a step towards our main sublinear time result, in this work we show that similar bounds to \cite{ChenTrogdonUbaru:2021} can also be proven for the popular kernel polynomial method (KPM) \cite{weisse2006kernel} and for a natural moment matching algorithm based on Chebyshev polynomials.

\subsection{Our contributions}
\label{sec:contrib}
With linear time spectral density estimation algorithms in hand for all Hermitian matrices, a natural question is if we can go faster for specific classes of matrices. In particular, there has been growing interest in SDE algorithms for graph structured matrices like adjacency matrices and Laplacians \cite{DongBensonBindel:2019}. A remarkable recent result by Cohen et al. \cite{Cohen-SteinerKongSohler:2018} shows that, for normalized graph adajeceny matrices, it is possible to achieve guarantee \eqref{eq:gaurantee} in $2^{O(1/\epsilon)}$ time, given appropriate query access to the target graph. Importantly, this runtime \emph{does not depend on $n$}. However, given the exponential dependence on $\epsilon$, the algorithm is impractical even for coarse spectral approximations.
% unable to compete with $O(n^2)$ time methods based on matrix-vector multiplication.

Our main contribution is a method that obtains a \emph{polynomial} dependence on $\epsilon$, at the cost of a linear dependence on the matrix dimension $n$. Since $A$ can have $n^2$ non-zero entries, the runtime is still \emph{sublinear} in the problem size, but with a much more acceptable dependence on accuracy. 
\begin{restatable}[Sublinear time spectral density estimation for graphs.]{theorem}{graphSDE}\label{thm:graphSDE}
Let $G = (V, E)$ be an unweighted, undirected $n$-vertex graph and let $A \in \R^{n \times n}$ be the normalized adjacency of $G$ with spectral density $s$. Let $\epsilon, \delta \in (0, 1)$ be fixed values. Assume that we can 1) uniformly sample a random vertex in constant time, 2) uniformly sample a random neighbor of any vertex $i\in V$ in constant time, and 3) for a vertex $i$ with degree $d_i$, read off all neighbors in $O(d_i)$ time.\footnote{A standard adjacency list representation of the graph would support these operations. As discussed in Section \ref{sec:graphSDE}, assumption (3) can be eliminated at the cost of an extra $\log n$ in the runtime as long as we know vertex degrees.} Then there is a randomized algorithm with expected running time $O(n \poly({\log(1/\delta)}/{\epsilon}))$ which outputs a density function $q : [-1, 1] \rightarrow \R^{+}$ such that $W_1(q, s) \leq \epsilon$ with probability at least $1 - \delta$.
\end{restatable}
 Note that the normalized graph Laplacian $L = I - A$ has the same eigenvalues as $A$ up to a shift and reflection, so Theorem \ref{thm:graphSDE} also yields a sublinear time result for normalized Laplacians, whose spectral densities are of interest in network science \cite{DongBensonBindel:2019}.

%Theorem \ref{thm:graphSDE} is proved in Section \ref{sec:graphSDE}. To establish the result, we present a more general algorithm for computing spectral density estimates given access to the target matrix $A$ only via \emph{approximate matrix-vector multiplications}. We title this robust spectral density estimation. 

\subsubsection{Robust spectral density estimation}
Theorem \ref{thm:graphSDE} is proven in Section \ref{sec:graphSDE}. 
A key component of the result is a sublinear time routine for computing coarse approximate matrix-vector products with any normalized graph adjacency matrix. To make use of such a routine, we need to develop an SDE algorithm that is \emph{robust} to the use of an approximate matrix-vector oracle. This is one of the main contributions of our work, as previous methods assume exact matrix-vector products. Formally, we assume access to the oracle:

\begin{definition}\label{def:epsAMV}
	An \emph{$\epsMV$-approximate matrix-vector multiplication oracle} for $A \in \R^{n \times n}$ and error parameter $\epsMV \in (0, 1)$ is an algorithm that, given any vector $y \in \R^n$, outputs a vector $z$ such that $\|z - Ay\|_2 \leq \epsMV \|A\|_2\|y\|_2$. We will denote a call to such an oracle for by $\amv(A,y,\epsMV)$.
\end{definition}

In Section \ref{sec:approxMatMultSDE} we prove the following for any Hermitian matrix $A$ (e.g., real symmetric) under the assumption that $\|A\|_2 \leq 1$, i.e., that $A$'s eigenvalues lie in $[-1,1]$:
\begin{restatable}[Robust spectral density estimation]{theorem}{approxMatMultSDE}\label{thm:approxMatMultSDE}
	Let $A \in \R^{n \times n}$ be a Hermitian matrix with spectral density $s$ and $\|A\|_2 \leq 1$. Let $C,C',C''$ be fixed positive constants. For any $\epsilon, \delta \in (0,1)$ and $\epsMV = C''\epsilon^{-3}\ln(1/\epsilon)$, there is an algorithm (Algorithm \ref{alg:MomentMatching_full}, with Algorithm \ref{alg:hutch_moments_approx} used as a subroutine to approximate moments) which makes
	 $T = C\ell/\epsilon$ calls to an $\epsMV$-approximate matrix-vector oracle for $A$, where $\ell = \max\left(1, \ \frac{C'}{n}\epsilon^{-2}\log^2(\frac{1}{\epsilon\delta})\log^2(\frac{1}{\epsilon})\right)$, and in $\poly(1/\epsilon)$ additional runtime, outputs a probability density function $q : [-1, 1] \rightarrow \R^{\geq 0}$ such that $W_1(s, q) \leq \epsilon$ with probability $1 - \delta$.
\end{restatable}

The requirement for the approximate matrix-vector oracle in Theorem \ref{thm:approxMatMultSDE} is relatively weak: we only need accuracy $\epsilon_{\textsf{MV}}$ that is polynomial in the final accuracy $\epsilon$. Importantly, there is no dependence on $1/n$, which allows for the theorem to be combined with coarse $\amv$ methods, including the one developed in Section \ref{sec:graphSDE} for normalized adjacency matrices. Based on random sampling, that method returns an $\epsilon$-approximate matrix-vector multiply in $O(n/\epsilon^2)$ time. This immediately yields our result for graphs given by Theorem \ref{thm:graphSDE}. We hope that Theorem \ref{thm:approxMatMultSDE} will find broader applications, since spectral density estimation is often applied to matrices where we only have inexact access to $A$. For example, $A$ might be a Hessian matrix that we can multiply by approximately using stochastic approximation \cite{Pearlmutter:1994,YaoGholamiKeutzer:2020}, or the inverse of some other matrix, which we can multiply by approximately using an iterative solver. 

We note that the result in Theorem \ref{thm:approxMatMultSDE} actually \emph{improves} as $n$ increases. Intuitively, when $A$ is larger, each matrix-vector product returns more information about the spectral density $s$, so we can estimate it more easily. We also remark that the density function $q$ returned by Algorithm \ref{alg:MomentMatching_full} is in the form of an $O(1/\epsilon^3)$ dimensional vector, with the $i$-th entry corresponding to probability mass placed on the $i$-th point of an evenly spaced grid on $[-1, 1]$. Alternatively, a simple rounding scheme that runs in $O(n + \poly({1}/{\epsilon}))$ time can extract from $q$ a vector of approximate eigenvalues $\tilde{\Lambda} = [\tilde{\lambda}_1, \ldots, \tilde{\lambda}_n]$ satisfying $\|\Lambda - \tilde{\Lambda}\|_1 \leq n\epsilon$, which, as discussed, is $\epsilon$ close to the spectral density $s$ in Wasserstein distance (see Theorem \ref{thm:approxEigs}).

Our approach for density estimation is based on a moment matching method that approximates \emph{Chebyshev polynomial} moments instead of the standard moments. I.e. we approximate $\tr(T_0(A))$, $\ldots$,  $\tr(T_N(A))$ where $T_0, \ldots, T_N$ are the Chebyshev polynomials of the first kind and then return a distribution whose Chebyshev moments closely match our approximations. By leveraging Jackson's theorem on polynomial approximation of Lipschitz functions  \cite{Jackson:1930}, we show how to bound the Wasserstein distance between two distributions in terms of the magnitude of the differences between their first $N = O(1/\epsilon)$ Chebyshev moments (see Lemma \ref{lemma:wass_err_moments}). Unlike results for standard moments \cite{KongValiant:2017}, the bound shows a near-linear relationship between Wasserstein distance and difference in the Chebyshev moments. Ultimately this allows us to obtain a polynomial dependence on $\epsilon$ in the number of approximate matrix-vector multiplications needed in Theorem \ref{thm:approxMatMultSDE}.

Along the way to proving that theorem, in Section \ref{sec:exactMatMultSDE} we first establish the follow result that is compatible with exact matrix-vector multiplications:

\begin{restatable}[Linear time spectral density estimation]{theorem}{exactMatMultSDE}\label{thm:exactMatMultSDE}
Let $A \in \R^{n \times n}$ be a Hermitian matrix with spectral density $s$ and $\|A\|_2 \leq 1$. Let $C,C'$ be fixed positive constants. For any $\epsilon, \delta \in (0,1)$, there is an algorithm (Algorithm \ref{alg:MomentMatching_full}, with Algorithm \ref{alg:hutch_moments} used as a subroutine to approximate moments) which computes $T = C\ell/\epsilon$ matrix-vector multiplications with $A$ where $\ell = \max\left(1, \ \frac{C'}{n}\epsilon^{-2}\log^2(\frac{1}{\epsilon\delta})\log^2(\frac{1}{\epsilon})\right)$, and in $\poly(1/\epsilon)$ additional runtime, outputs a probability density function $q : [-1, 1] \rightarrow \R^{\geq 0}$ such that $W_1(s, q) \leq \epsilon$ with probability $1 - \delta$.
\end{restatable}
As in Theorem \ref{thm:approxMatMultSDE}, the theorem improves as $n$ increases, requiring just $T = O(1/\epsilon)$ matrix vector multiplies when $n = \Omega(1/\epsilon^2)$. The runtime of Theorem \ref{thm:exactMatMultSDE} is dominated by the cost of the matrix-vector multiplications, which take $O(T\cdot n^2)$ time to compute for a dense matrix, and $O(T\cdot \nnz(A))$ time for a sparse matrix with $\nnz(A)$ non-zero entries, so the algorithm runs in linear time when $\epsilon,\delta$ are considered constant. 

Given Theorem \ref{thm:exactMatMultSDE}, we prove Theorem \ref{thm:approxMatMultSDE} by showing that the error introduced by approximate matrix-vector multiplications does not hinder our ability to estimate the Chebyshev polynomial moments. We do so by drawing on stability results for the three-term recurrence relation defining these polynomials \cite{Clenshaw:1955,MuscoMuscoSidford:2018}.
%, allowing us to use a coarse $\amv$ method that has accuracy $\epsMV$ that is polynomial in the final accuracy $\epsilon$. 

\begin{remark} The number of matrix-vector multiplies $N\ell = N \cdot \max(1, \frac{C'}{n}\epsilon^{-2}\log^2(\frac{1}{\epsilon \delta})\log^2(\frac{1}{\epsilon}))$ in Theorems \ref{thm:approxMatMultSDE} and  \ref{thm:exactMatMultSDE} can be improved by up to a $\log^2(1/\epsilon)$ factor in the regime when $n$ is small, specifically $n \leq C'\epsilon^{-2}\log^2(1/(\epsilon \delta))$. This is discussed further in Section \ref{sec:SDE_approx}.
\end{remark}

\subsubsection{Spectral density estimation via the kernel polynomial method}
In addition to the Chebyshev moment matching method used to give Theorem \ref{thm:exactMatMultSDE} and Theorem \ref{thm:approxMatMultSDE}, we prove that a version of the popular kernel polynomial method (KPM) can be used to obtain a spectral density estimate with similar running times, albeit with slightly worse dependence on the accuracy parameter $\epsilon$.\footnote{We believe that the extra $O(\epsilon^{-2})$ factor in the number of matrix-vector multiplications (or calls to an approximate matrix-vector oracle in the robust setting) may be an artifact of our analysis and can be further improved to match the approximate Chebyshev moment matching bounds.}  Along with the Stochastic Lanczos Quadrature method, the kernel polynomial method is one of two dominant spectrum estimation algorithms used in practice.

Given sufficiently accurate approximations to the Chebyshev polynomial moments, the KPM method outputs a density function $q$ in the form of a $O(1/\epsilon)$ degree polynomial multiplied by a simple closed form function. This is described in Algorithm \ref{alg:SDEpolynomial_full} in Section \ref{sec:SDE} and should be thought of as analagous to Algorithm \ref{alg:MomentMatching_full}. Specifically, we can obtain Theorem \ref{thm:exactMatMultSDE} and Theorem \ref{thm:approxMatMultSDE} with $\ell = \max(1, \ \frac{C'}{n}\epsilon^{-4}\log^2(\frac{1}{\epsilon\delta}))$ and $\epsMV = C''\epsilon^{-4}$ (in the robust setting), by using Algorithm \ref{alg:SDEpolynomial_full} instead of Algorithm \ref{alg:MomentMatching_full}.
Our proof in the KPM case is again based on Jackson's work on polynomial approximations for Lipschitz functions: we take advantage of the fact that Jackson constructs approximations that are both \emph{linear} and \emph{preserve positivity}  \cite{Jackson:1912}.

\subsection{Related work} 
As mentioned, most closely related to our sublinear time result on graphs is the result of Cohen et al. \cite{Cohen-SteinerKongSohler:2018}. They prove a result which matches the guarantee of Theorem \ref{thm:graphSDE}, but with runtime of $2^{O(1/\epsilon)}$ -- i.e., with \emph{no dependence} on $n$. In comparison, our result depends linearly on $n$, but only polynomially on $1/\epsilon$. An interesting open question is if a $\poly(1/\epsilon)$ time algorithm is possible but we conjecture that the trade-off between the dependence on $n$ and the accuracy $\epsilon$ is inherent. Our bound in Lemma \ref{lemma:wass_err_moments} on the Wasserstein-1 distance between two distributions can be seen as analagous to Proposition 1 from \cite{KongValiant:2017}, which is the basis of the result in \cite{Cohen-SteinerKongSohler:2018}. They bound the Wasserstein-1 distance between two distributions in terms of the differences in the standard moments of the distributions. The bound requires an exponentially small dependence on $1/\epsilon$, i.e. $2^{-O(1/\epsilon)}$, in the difference between the standard moments while the bound from Lemma \ref{lemma:wass_err_moments} only requires an $O(\epsilon/\ln(1/\epsilon))$ difference in the Chebyshev moments.

As discussed, algorithms for spectral density estimation have been studied since the early 90s \cite{Skilling:1989,SilverRoder:1994,Wang:1994} but only analyzed recently. In addition to the work of Chen, Trogdon, and Ubaru that was discussed \cite{ChenTrogdonUbaru:2021}, \cite{MuscoNetrapalliSidford:2018} provides an algorithm for computing an approximate histogram for the spectrum of matrix. That result can be shown to yield an $\epsilon$ error approximation to the spectral density in the Wasserstein-1 distance with roughly $O(1/\epsilon^5)$ matrix-vector multiplications. This compares to the improved $O(1/\epsilon)$ matrix-vector multiplications required by our  Theorem \ref{thm:exactMatMultSDE}.

\textbf{Matrix-vector query algorithms.} Our work fits into a broader line of work on proving upper and lower bounds on the \emph{matrix-vector query complexity} of linear algebraic problems, from top eigenvector, to matrix inversion, to rank estimation \cite{SunWoodruffYang:2019,SimchowitzEl-AlaouiRecht:2018,BravermanHazanSimchowitz:2020,MeyerMuscoMusco:2020,DharangutteMusco:2021}. The goal in this model is to minimize the total number of matrix-vector multiplications with $A$, recognizing that such multiplications either 1) dominate runtime cost or 2) are the only way to access $A$ when it is an implicit matrix. The matrix-vector query model generalizes both classical Krylov subspace methods, as well as randomized sketching methods \cite{Woodruff:2014}. Studying other basic linear algebra problem when matrix-vector multiplication queries are only assumed to be approximate (as in Definition \ref{def:epsAMV}) is an interesting future direction.

\subsection{Paper Roadmap} 
We describe notation and preliminaries on polynomial approximation in Section \ref{sec:Preliminaries}. We use these tools in Section \ref{sec:approxCheby_mm} to prove that a good approximation to the first $O(1/\epsilon)$ Chebyshev polynomial moments of the spectral density can be used to extract a good approximation in Wasserstein-1 distance. This result is the basis for our result on robust spectral density estimation stated in Theorem \ref{thm:exactMatMultSDE} and linear time spectral density estimation stated in Theorem \ref{thm:approxMatMultSDE}, which are proven in Section \ref{sec:SDE_approx}. Finally, we give a randomized algorithm to implement an approximate matrix-vector multplication oracle for adjacency matrices in Section \ref{sec:graphSDE} and prove our main result, Theorem \ref{thm:graphSDE}. In Section \ref{sec:KPM} we describe and analyze the kernel polynomial method, showing that it too can be used to obtain a spectral density estimate given approximations to the first $O(1/\epsilon)$ Chebyshev polynomial moments. In Section \ref{sec:experiments}, we empirically investigate the potential of combining approximate matrix-vector multiplications with our moment matching method, the kernel polynomial method, and the stochastic Lanczos quadrature method studied in \cite{ChenTrogdonUbaru:2021}. We show that all three can achieve accurate SDE estimates in sublinear time for a variety of graph Laplacians.

\section{Preliminaries}\label{sec:Preliminaries}

Throughout we assume that $A \in \R^{n \times n}$ is Hermitian with eigendecomposition $A = U \Lambda U^{*}$, where $UU^* = U^*U = I_{n\times n}$. We assume that $A$'s eigenvalues satisfy $-1 \leq \lambda_n \leq \cdots \leq \lambda_1 \leq 1$. In many applications $A$ is real symmetric.
We denote $A$'s spectral density by $s$, which is defined in \eqref{eq:spect_density}. Our goal is to approximate $s$ in the  Wasserstein-1 metric with another distribution $q$ supported on $[-1,1]$. Specifically, 
%To The Wasserstein-1 metric on probability measures $\gamma, \beta$ on $\R^d$ is defined,
as per the dual formulation given by the Kantorovich-Rubinstein theorem \cite{kantorovich1957functional}, for $s,q$ supported on $[-1,1]$ the metric is equal to:
\begin{align}
	\label{eq:wass_def}
W_1(s, q) = \sup_{\substack{f : \R \rightarrow \R\\ |f(x) - f(y)| \leq |x-y| ~\forall x,y} }\left\{\int_{-1}^1f(x)\left(s(x) - q(x)\right)dx\right\}.
\end{align}
In words, $s$ and $q$ are close in Wasserstein-1 distance if their difference has small inner product with all 1-Lipschitz functions $f$. Alternatively, $W_1(s, q)$ is equal to the cost of ``changing'' one distribution to another, where the cost of moving one unit of mass from $x$ to $y$ is $|x - y|$: this is the ``earthmover's'' formulation common in computer science. Note that \eqref{eq:wass_def} can be applied to arbitrary functions $s,q$, even if they are \emph{not distributions}, and we will occasionally do so.

\paragraph{Functions and inner products.} 
We introduce notation for functions used throughout the paper. 
Let $\FcOne$ denote the space of real-valued functions on $[-1,1]$. For  $g,h\in \FcOne$,  let $\iprod{g}{h}$ denote $\iprod{g}{h} \eqdef \int_{-1}^1 g(x)h(x)dx$. For $f\in \FcOne$, we define $\|f\|_2 \eqdef \sqrt{\iprod{f}{f}}$ and let $\|f\|_{\infty}$ denote the max-norm  $\|f\|_{\infty} = \max_{x\in [-1,1]} |f(x)|$. We let $\|f\|_{1}$ denote  $\|f\|_{1} = \int_{-1}^1|f(x)|dx$.

Let $\Fc(\Z, \R)$ be the space of real-valued functions on the integers, $\Z$. For $f,g\in \Fc(\Z, \R)$ let $(f*g)$ denote the discrete convolution: $(f*g)[n] = \sum_{m=-\infty}^\infty f[m]g[n-m]$.
Let $\Fc(\N, \R)$ be the space of real-valued functions on the natural numbers, $\N$. For functions in $\Fc(\Z, \R)$ or $\Fc(\N, \R)$ we typically used square brackets instead of parentheses.

For two functions $f,g$ let $h = fg$ (or $h=f\cdot g$) and $j = f/ g$ denote the pointwise product and quotient respectively. I.e. $h(x) = f(x)g(x)$ and $j(x) = f(x)/g(x)$ for all $x$.

\paragraph{Chebyshev polynomials.} Our approach is based on approximating Chebyshev polynomial moments of $A$'s spectral density, and we will use basic properties of these polynomials, the $k^\text{th}$ of which we denote $T_k$. The Chebyshev polynomial of the first kind can be defined via the recurrence:
\begin{align*}
	T_0(x) &= 1  \hspace{2em} T_1(x) = x \\ 
	T_k(x) &= 2x \cdot T_{k-1}(x) -  T_{k-2}(x) \hspace{2em}\text{for $k \geq 2$}.
\end{align*}
We will use the well known fact that the Chebyshev polynomials of the first kind are bounded between $[-1,1]$, i.e. $\max_{x\in [-1, 1]}|T_k(x)|\leq 1$. 

Let $w(x) \eqdef \frac{1}{\sqrt{1 - x^2}}$. It is well known that $\iprod{T_0}{w\cdot T_0} = \pi$, $\iprod{T_k}{w\cdot T_k} = \pi/2$ for $k > 0$, and 
\begin{align*}
	\iprod{T_i}{w\cdot T_j} &= 0  &&\text{for $i \neq j$}.
\end{align*} 
In other words, the Chebyshev polynomials are orthogonal on $[-1,1]$ under the weight function $w$. The first $k$ Chebyshev polynomials form an orthogonal basis for the degree $k$ polynomials under this weight function. We let $\bar{T}_k$ denote the \emph{normalized} Chebyshev polynomial $\bar{T}_k \eqdef T_{k}/\sqrt{\iprod{T_k}{w\cdot T_k}}$. 
\begin{definition}[Chebyshev Series]\label{def:cheby_series}
	The \emph{Chebyshev expansion or series} for a function $f \in \FcOne$ is given by \vspace{-1em}
	\begin{align*}
		\sum_{k = 0}^\infty \iprod{f}{w\cdot \bar{T}_k} \cdot \bar{T}_k.
	\end{align*}
	We call $\sum_{k = 0}^N \iprod{f}{w\cdot \bar{T}_k} \cdot \bar{T}_k$ the truncated Chebyshev expansion or series of degree $N$. 
\end{definition}

\paragraph{Other notation.} 
%Let $\N$ denote the natural numbers $1,2, \ldots$.
 Let $[n]$ denote $1,\ldots, n$. For a scalar function $f:\R\rightarrow \R$ and $n\times n$ matrix $A$ with eigendecomposition $U\Lambda U^{*}$ , we let $f(A)$ denote the matrix function $Uf(\Lambda)U^{*}$. Here $f(\Lambda)$ is understood to mean $f$  applied entrywise to the diagonal matrix $\Lambda$ containing $A$'s eigenvalues. Note that $\tr(f(A)) = \sum_{i=1}^n f(\lambda_i)$. When $f(x)$ is a degree $q$ polynomial, $c_0 + c_1x + \ldots, c_q x^q$, then we can check that $f(A)$ exactly equals $c_0I + c_1A + \ldots, c_q A^q$, where $I$ is then $n\times n$ identity matrix. So $f(A)y$ can be computed for any vector $y$ using $q$ matrix-vector multiplications with $A$. 
%Note that if $\max_{x\in [-1, 1]}|f(x)|\leq C$, then $\|f(A)\|_2\leq C$ for any $A$ with $\|A\|_2 \leq 1$.  

% \section{Chebyshev Moment Matching Method}
% \paragraph{Recap.} Recall that we denote the inner-product $\iprod{f}{g}$ between two functions $f, g \in \Fc([-1, 1], \R)$ to be $\int_{-1}^1 f(x)g(x)dx$ and the $\ell_2$-norm of a function $f$ to be $ \|f\|_2  := \sqrt{\iprod{f}{f}}$. We use $\|\cdot\|_2$ to denote the $\ell_2$-norm of both vectors and functions but make it clear from the context which definition is being used. We denote $\Delta_{[-1, 1]}$ to be the probability simplex on the interval $[-1, 1]$.  

\section{Approximate Chebyshev Moment Matching}\label{sec:approxCheby_mm}
In this section we show that the spectral density $s$ of a Hermitian matrix $A$ with eigenvalues in $[-1, 1]$ can be well approximated given access to \emph{approximations} of the first $N = O(1/\epsilon)$ normalized Chebyshev polynomial moments of $s$, i.e., to approximations of $\tr(\Tbar_1(A)), \dots, \tr(\Tbar_N(A))$. We state our result in Algorithm \ref{alg:MomentMatching_full} and analyze it in Section \ref{sec:approx_mm}. We show later, in Section \ref{sec:SDE_approx}, a method to approximate these moments using a stochastic trace estimator, implemented with either exact or approximate matrix vector multiplications with $A$.

Given approximations $\tilde{\tau}_1, \dots, \tilde{\tau}_N$ to the first $N$ normalized Chebyshev moments of $A$, a natural approach is to find a probability density $q : [-1, 1] \to \R^+$ such that the first $N$ normalized Chebyshev moments of $q$, i.e., $\iprod{\Tbar_1}{q}, \dots, \iprod{\Tbar_N}{q}$, closely approximate $\tilde{\tau}_1, \dots, \tilde{\tau}_N$. In order for this approximate moment matching approach to return a good spectral density estimate, it requires that: \emph{for any density function $q$, if the first $N$ Chebyshev moments of $q$ closely approximate those of $s$, then $q$ must be close to $s$ in Wasserstein distance.}
% Our goal in this section is to give a bound on the Wasserstein-1 distance between two distrubtions $s, z \in \Delta_{[-1, 1]}$ in terms of the magnitude of differences of their first $N$ (normalized) Chebyshev moments. Before we give this bound, we discuss two facts on the convergence of the truncated Chebsyhev series of a continuous function on $[-1, 1]$ and the rate of decay of its normalized Chebyshev moments. 
To that end, we prove the following lemma:
% The lemma bounds the Wasserstein distance between two distributions in terms of the magnitude of the differences of their first $N$ normalized Chebyshev moments. Specifically: 
\begin{lemma}\label{lemma:wass_err_moments}
Let $N \in 4\N^+$ be a degree parameter and $p, q$ be distributions on $[-1, 1]$. $$W_1(p, q) \leq \frac{36}{N} + 2\sum_{k = 1}^N \frac{|\iprod{\Tbar_k}{p} - \iprod{\Tbar_k}{q}|}{k}.$$
\end{lemma}
Lemma \ref{lemma:wass_err_moments} shows that if the first $N$ normalized Chebyshev moments of two distributions are identical, then the Wasserstein distance between the distributions is at most $ O(1/N)$. When the moments between the distributions differ, the contribution of the difference between the $k$-th moments to the Wasserstein distance is scaled by $O(1/k)$. In particular, the lemma shows that deviation in the lower moments between distributions contributes more to the Wasserstein distance.

To prove Lemma \ref{lemma:wass_err_moments}, we will use two well-known results on approximating Lipschitz functions by polynomials. The first is proven in \cite{Jackson:1930}. and concerns uniform approximation of Lipschitz continuous functions by a Chebyshev series:
\begin{fact}\label{fact:Jacksons}
Let $f \in \Fc([-1, 1], \R)$ be a Lipschitz continuous function with Lipschitz constant $\lambda > 0$. Then, for every $N \in 4\N^+$, there exists $N+1$ constants $\bhat_N[0] >  \dots > \bhat_N[N] \geq 0$ such that the polynomial $\bar{f}_N = \sum_{k =0}^N \frac{\bhat_{N}[k]}{\bhat_N[0]}\iprod{f}{w \cdot \Tbar_k} \Tbar_k$ has the property that $\max_{x \in [-1, 1]}|f(x) - \bar{f}_N(x)| \leq 18\lambda/N.$
\end{fact}
The coefficients of the polynomial in Fact \ref{fact:Jacksons} are not explicitly stated since we only require the existence of such a polynomial in order to prove Lemma \ref{lemma:wass_err_moments}. We defer the reader to Appendix \ref{sec:JacksonKernelWasserstein} for an explicit construction of the polynomial\footnote{The construction of the polynomial $\bar{f}_N$ in Fact \ref{fact:Jacksons} and its uniform approximation to $f$ forms the basis of our alternate approach, the Kernel Polynomial Method, which is discussed in-depth in Appendix \ref{sec:JacksonKernelWasserstein}.} and Appendix \ref{thm:jacksons_theorem_algebraic} for a proof of Fact \ref{fact:Jacksons}.  

Next, we state a well-known fact that the magnitude of the inner-product of a Lipschitz function $f$ with the $k$-th Chebyshev polynomial (for $k \geq 1$) under the Chebyshev weight function $w = 1/\sqrt{1 - x^2}$ is bounded by $O(1/k)$, i.e., $|\iprod{f}{w \cdot \Tbar_k}| \leq O(1/k)$. Our proof is given in Appendix \ref{appx:moment_magnitude_proof} and is a simple adaptation of the proof of Theorem 4.2 in \cite{trefethen2008gauss}. 
\begin{fact}\label{fact:moment_magnitude_bound}
Let $f \in \Fc([-1, 1], \R)$ be a Lipschitz continuous function with Lipschitz constant $\lambda > 0$. Then, for any $k \geq 1$, we have that $|\langle f, w  \cdot \Tbar_k \rangle| = |\int_ {-1}^1 f(x)\Tbar_k(x)w(x)dx | \leq {2\lambda}/{k}$.
\end{fact}
With Fact \ref{fact:Jacksons} and \ref{fact:moment_magnitude_bound} in place, we are now ready to prove Lemma \ref{lemma:wass_err_moments}

\begin{proof}[Proof of Lemma \ref{lemma:wass_err_moments}]
Recall that the dual formulation of the Wasserstein-1 distance due to Kantorovich-Rubinstein gives us that $W_1(p, q) = \sup_{\substack{f \in \text{lip}_1}}\int_{-1}^1 f(x)(p(x) - q(x))dx$ where $\text{lip}_1$ denotes the set of $1$-Lipschitz functions on $[-1, 1]$. Let $f \in \text{lip}_1$ be an arbitrary $1$-Lipschitz function and let $\{\bhat_N[k]\}_{k = 0}^N$ and $\bar{f}_N$ be the coefficients and polynomial respectively from Fact \ref{fact:Jacksons} for function $f$. We can then bound $W_1(p, q)$ using the triangle inquality as 
\begin{align*}
    W_1(p, q) \leq \underbrace{\int_{-1}^1  |f(x) - \bar{f}_N(x)|(p(x) - q(x)) dx}_{t_1} +  \underbrace{\int_{-1}^1 \bar{f}_N(p(x) - q(x))dx}_{t_2}. 
\end{align*}
Using the fact that $f$ is Lipschitz and the bound from Fact \ref{fact:Jacksons}, along with the fact that $p$ and $q$ are distributions, we have that $t_1 \leq 36/N$. 

It is left to bound $t_2$. We expand $t_2$ using the Chebyshev series expansion of $\bar{f}_N$ and note that $\iprod{g/w}{w \cdot \Tbar_k} = \iprod{g}{\Tbar_k}$ for any function $g \in \Fc([-1, 1], \R)$, giving us 
\begin{align*}
    t_2 &= \int_{-1}^1 \bar{f}_N(x) w(x) \cdot \frac{p(x) - q(x)}{w(x)} dx
    = \int_{-1}^1 \bar{f}_N(x) w(x) \cdot \sum_{k=0}^\infty \iprod{p-q}{\Tbar_k}\Tbar_k(x)dx\\
    &= \int_{-1}^1 \left( w(x)\sum_{k=0}^N \frac{\bhat_N[k]}{\bhat_N[0]}\iprod{f}{w \cdot \Tbar_k}\Tbar_k(x) \right)  \left(\sum_{k=0}^\infty \iprod{p-q}{\Tbar_k}\Tbar_k(x) \right) dx.
\end{align*}
By the orthogonality of the Chebyshev polynomials under the weight function $w$ and the fact that $\iprod{\Tbar_k}{\Tbar_k} = 1$ for all $k \in [N]$, we can bound the magnitude of $t_2$ as  
\begin{align*} 
|t_2| \leq \sum_{k=1}^N |\langle f, w \cdot \Tbar_{k}\rangle| \cdot |\iprod{\Tbar_k}{p} - \iprod{\Tbar_k}{q}|
\end{align*}
since we have that $0 \leq \bhat_N[k]/\bhat_N[0] \leq 1$ and $|\int_{-1}^1 \Tbar_k(p(x) - q(x))dx| = |\langle \Tbar_k, p \rangle - \langle\Tbar_k, q\rangle|$ for each $k \in [N]$. Additionally, since $p$ and $q$ are distributions we have that $\iprod{\Tbar_0}{s} = \iprod{\Tbar_0}{z} = 1/\sqrt{\pi}$. We then use the bound from Fact \ref{fact:moment_magnitude_bound} on $|\langle f, w \cdot \Tbar_{k}\rangle|$ for each $k \in [N]$. Putting this together gives us that $|t_2| \leq \sum_{k =1}^N 2|\langle \Tbar_k, p \rangle - \langle\Tbar_k, q\rangle|/k$.

Putting together the bound on $t_1$ and $t_2$ gives us the bound on $W_1(p, q)$.
\end{proof}

\subsection{Moment Matching Algorithm}\label{sec:approx_mm}
With Lemma \ref{lemma:wass_err_moments} in place, our next step is develop a method to find a distribution $q$ with Chebyshev moments closely matching a given set of target moments.
In order to search for a distribution, we consider an evenly-spaced grid of the interval $[-1, 1]$. Specifically, let $d \in \N^+$ be a discretization parameter and let $X_d = [-1, -1 + \frac{2}{d}, \dots, 1-\frac{2}{d}, 1]$ be a $(d+1)$-length evenly-spaced grid of the interval $[-1, 1]$. Our goal is to output a distribution supported on $X_d$ for an appropriately chosen value of $d$. Any such distribution can be described by a vector in $\R_{\geq 0}^d$ such that the $i$-th entry corresponds to the probability mass placed at point $-1 + 2i/d$ on the grid. Where it is clear from the context, we will denote the distribution and its probability mass vector interchangeably.  

% The pseudocode is stated in Algorithm \ref{alg:MomentMatching_full} and takes as inputs the degree parameter $N \in 4\N^+$ and approximations $\tilde{\tau}_1, \dots, \tilde{\tau}_N$ to the first $N$ normalized Chebyshev moments of $s$. 

% Let $d \in \mathbb{N}$ be a discretization parameter. We denote a $d$-length evenly spaced grid of the interval $[-1, 1]$ by $X_d = [x_0 = -1, \dots, x_{d} = 1]$. Additionally, we define $\Delta_d \subseteq \R^d$ be the set of discrete distributions on the grid $X_d$. Specifically we have that $z^\top \vec{1} = 1$ and $z \geq 0$ for all $z \in \Delta_d$. When it is clear from the context, we will abuse notation and denote a distribution $z \in \Delta_d$ to be a function on $[-1, 1]$ or a vector in $\R^d$. 
In order to compute the first $N$ normalized Chebyshev moments of functions on the grid $X_d$, we define two matrices $\Tc_N^d, \widehat{\Tc}_N^d \in \R^{N \times d}$ such that for $k \in [N]$ and $i \in [d]$,
\begin{align*} 
(\Tc_N^d)_{k,i} = \Tbar_k(-1 +{2i}/{d}) &&\text{and} &&&(\widehat{\Tc}^d_N)_{k,i} = \frac{\Tbar_k(-1 + {2i}{d})}{k}.
\end{align*} 
The matrix $\Tc^d_N$ corresponds to a ``discretization'' of the continuous operator that computes the first $N$ normalized Chebyshev moments of a continuous function on $[-1, 1]$. In particular, for a distribution $q$ supported on $X_d$, we have that $\iprod{q}{\Tbar_k} = \sum_{i = 0}^d q_i\Tbar_k(-1 + 2i/d) = (\Tc^d_N q)_k$. Notice that the matrix $\Tc_N^d$ does not contain the row for $\Tbar_0$; since we are working with distributions we know that $\Tbar_0(q) = 1/\sqrt{\pi} \cdot \int_{-1}^1 qdx = 1/\sqrt{\pi}$ for any distribution $q$ on $[-1, 1]$. The matrix $\widehat{\Tc}^d_N$ is the matrix $\Tc^d_N$ with the $k$-th row scaled by $1/k$. 
With this notation in place, we state the approximate moment matching algorithm in full in Algorithm \ref{alg:MomentMatching_full}.
\begin{algorithm}[!h]\caption{Approximate Chebyshev Moment Matching}\label{alg:MomentMatching_full}
	\begin{algorithmic}[1]
		\Require Symmetric $A\in \R^{n\times n}$, degree parameter $N \in 4\N^+$, algorithm $\mathcal{M}(A)$ that computes moment approximations $\tilde{\tau}_1, \ldots, \tilde{\tau}_N$ with the guarantee that $|\tilde{\tau}_k - \frac{1}{n}\tr(\bar{T}_k(A))| \leq (N\ln(eN))^{-1}$ for all $k$. 
		\Ensure A vector $q$ corresponding to a discrete density function on $[-1, 1]$.
		\State For $k=1, \ldots, N$ use $\mathcal{M}$ to compute $\tilde{\tau}_1, \ldots, \tilde{\tau}_N$ and set $z = [{\tilde{\tau}_1}/{1}, {\tilde{\tau}_2}/{2}, \ldots, {\tilde{\tau}_N}/{N}]$.
		\State Set $d = \lceil N^3/2 \rceil$ and compute matrix $\widehat{\Tc}^d_N \in \R^{N \times d}$. \Comment{$(\widehat{\Tc}^d_N)_{k,i} = {\Tbar_k(-1 + \frac{2i}{d})}/{k}$.}
		\State Minimize $\|\widehat{\Tc}^d_N q - z\|_1$ subject to $q^\top \vec{1} = 1$ and $q \geq 0$. \label{line:regression_mm}
		\State {Return} $q$. 
	\end{algorithmic}
\end{algorithm}

Note that the optimization problem in Line \ref{line:regression_mm} of Algorithm \ref{alg:MomentMatching_full} can easily be written as a linear program in $O(d + N)$ variables and constraints and hence can be solved efficiently in $\poly(N, d) = \poly(1/\epsilon)$ time\footnote{Additionally, note that the optimization problem has a convex objective and constraints -- in particular, the set of distributions supported on $X_d$ is a convex set. The objective function $\|\widehat{\Tc}_N^dq - z\|_1$ is not differentiable, but has subgradients. Hence, this program can be solved efficiently in $\poly(1/\epsilon)$ time using a projected subgradient method. This requires an oracle that projects onto the the probability simplex supported on the grid $X_d$ -- an algorithm that runs in  $O(d\log d)$ time has been given in multiple papers, see \cite{wang2013projection} for more details.}. Since this method is independent of the matrix dimension $n$, it is a lower order term in the running time stated in Theorems \ref{thm:exactMatMultSDE} and \ref{thm:approxMatMultSDE}, as we will discuss in Section \ref{sec:SDE_approx}.

We show that when $N = O(1/\epsilon)$, Algorithm \ref{alg:MomentMatching_full} returns a distribution satisfying $W(s, q) \leq \epsilon$. 
\begin{lemma}\label{lemma:mm_guarantee}
Let $\epsilon \in [0, 1]$ and let $N \geq {18}/{\epsilon}$. Then the distribution $q : [-1, 1] \rightarrow \R^+$ returned by Algorithm \ref{alg:MomentMatching_full} satisfies $W_1(q, s) \leq 3\epsilon.$
\end{lemma}
\begin{proof}
We start by giving some notation -- for a distribution $y : [-1, 1] \to \R^+$, we denote $\vec{\tau}_y \coloneqq [\iprod{\Tbar_1}{y}, \dots, \iprod{\Tbar_N}{y}]$ to be the vector of the first $N$ normalized Chebyshev moments of $y$. For an integer $k \in \N^+$, we denote $\vec{k}$ to be the vector in $\R^{k}$ given by $\vec{k} \eqdef [1, \dots, k]$ and for a vector $y \in \R^k$ write $y/\vec{k}$ to denote the vector $y/\vec{k} \eqdef [y_1/1, \dots, y_k/k]$. Notice then that we have $\vec{\tau}_q = \Tc^d_N q$ and $\vec{\tau}_q/\vec{N} = \widehat{\Tc}^d_N q$.

We start by bounding the scaled differences in the first $N$ normalized Chebyshev moments of $q$ and $s$ in order to use Lemma \ref{lemma:wass_err_moments} on $q$ and $s$. 
\begin{align}
    \|\vec{\tau}_q/\vec{N} - \vec{\tau}_s/\vec{N}\|_1 \leq \|\vec{\tau}_q/\vec{N} - z\|_1 + \|z - \vec{\tau}_s/\vec{N}\|_1 \leq \|\vec{\tau}_q/\vec{N} - z\|_1 + \frac{1}{N}. \label{eqn:regression_bounds_eq1}
\end{align}
% \begin{align*}
%     \sum_{k=1}^N \frac{|(\Tc_Nz)_k - (\vec{\tau}_s)_k|}{k} \leq \sum_{k=1}^N \frac{|(\Tc_Nz)_k - \tilde{\tau}_k|}{k} + \sum_{k=1}^N \frac{|\tilde{\tau}_k - (\vec{\tau}_s)_k|}{k} \leq \sum_{k=1}^N \frac{|(\Tc_Nz)_k - \tilde{\tau}_k|}{k} + \frac{1}{N}.
% \end{align*}
The first inequality follows by applying the triangle inequality and in the second inequality we used the fact that $\|z - {\vec{\tau}_s}/{\vec{N}}\|_1 = \sum_{k=1}^N |\tilde{\tau}_k - (\vec{\tau}_s)_k|/k \leq H_n \cdot (N\ln(eN))^{-1} \leq 1/N$.

Next we show that there exists a distribution $q'$ supported on $X_d$ such that $\|\vec{\tau}_{q'}/\vec{N} - z\| \leq {1}/{N}$. To this end, consider the following distribution $q^*$ on $X_d$: 
$$q^*(x) = \frac{1}{n} \sum_{i = 1}^n \delta(x - \argmin_{p \in X_d} |p - \lambda_i|).$$
In words, $q^*$ is the distribution corresponding to moving the mass from each $\lambda_i$ to its nearest point on the grid $X_d$. Notice that we have $W_1(s, q^*) \leq 1/d$ due to the earthmover distance interpretation of the Wasserstein-1 distance. 

Applying the triangle inequality and the guarantee from the moment approximations, we get that 
$\|\vec{\tau}_{q^*}/\vec{N} - z\|_1 \leq 1/N + \|\vec{\tau}_{q^*}/\vec{N} - \vec{\tau}_s/\vec{N}\|_1$. It is left then to bound $\|\vec{\tau}_{q^*}/\vec{N} - \vec{\tau}_s/\vec{N}\|_1$. To this end, we state the following well-known fact about the derivatives of Chebyshev polynomials. 
\begin{fact}\label{fact:Cheby_Lipschitz}
For $k \geq 1$, $\frac{d T_k(x)}{dx} = k U_{k-1}(x)$. 
\end{fact}
We then have using the definition of $q^*$ that, for any $1 \leq k \leq N$,
\begin{align*}
    \vert\langle \Tbar_k, s\rangle - \langle \Tbar_k, q^* \rangle \vert &= \left \vert \frac{1}{n}\sum_{i = 1}^n \Tbar_k(\lambda_i) - \Tbar_k(\argmin_{p \in X_d} |p - \lambda_i| ) \right \vert
    \leq \frac{1}{n}\sum_{i = 1}^n\left \vert \Tbar_k(\lambda_i) - \Tbar_k(\argmin_{p \in X_d} |p - \lambda_i| ) \right \vert \\ 
    &\leq \frac{\sqrt{2}}{n\sqrt{\pi}}\sum_{i = 1}^n \max_{x \in [-1, 1]} \left \vert \frac{d T_k(x)}{dx} \right \vert \cdot |\lambda_i - \argmin_{p \in X_d} |p - \lambda_i| |  
    \leq \frac{\sqrt{2}k^2}{d\sqrt{\pi}}
\end{align*}
where in the last inequality we used the fact that $\max_{x \in [-1, 1]}|U_{k-1}(x)| \leq k$. It follows then that 
\begin{align*}
\|\vec{\tau}_{q^*}/\vec{N} - \vec{\tau}_s/\vec{N}\|_1 = \sum_{k=1}^N \frac{|(\vec{\tau}_{q^*})_k - (\vec{\tau}_s)_k|}{k} \leq \frac{N(N+1)}{d\sqrt{2\pi}} \leq  \frac{1}{N}
\end{align*} 
by taking the sum over all $k$ and noting that $d \geq N^3/2$. Putting these bounds together gives us that $\|\vec{\tau}_{q^*}/\vec{N} - z\|_1 \leq 2/N$.

% Our goal then is to bound $\|\vec{\tau}_z/\vec{N} - q\|_1$ using the bound on $\|\vec{\tau}_{z^*}/\vec{N} - q\|_1$ and the guarantee that $\|\vec{\tau}_z/\vec{N} - q\|_2 \leq 2\min_{z' \in \Delta_d}\|\vec{\tau}_{z'}/\vec{N} - q\|_2$. Let $z' \in \Delta_d$ be a distribution, we then have the following inequalities: 
% \begin{align*}
% \|\vec{\tau}_{z'}/\vec{N} - q\|_1^2 &= \left( \sum_{k = 1}^N \frac{|(\tau_{z'})_k - \tilde{\tau}_k|}{\sqrt{k}} \cdot \frac{1}{\sqrt{k}} \right)^2 \leq \left( \sum_{k = 1}^N \frac{|(\tau_{z'})_k - \tilde{\tau}_k|^2}{k} \right) \left( \sum_{k=1}^N \frac{1}{k} \right) \\ 
% \end{align*}
% We have from the guarantee of step 2 that $\|\vec{\tau}_z/\vec{N} - q\|^2_2 \leq 4\|\vec{\tau}_{z^*}/\vec{N} - q\|^2_2$

Since $\|\vec{\tau}_{q}/\vec{N} - z\|_1 \leq \|\vec{\tau}_{q^*}/\vec{N} - z\|_1$ from Line \ref{line:regression_mm} of Algorithm \ref{alg:MomentMatching_full}, we plug this into \eqref{eqn:regression_bounds_eq1} to get that $\|\vec{\tau}_{q}/\vec{N} - \vec{\tau}_s/\vec{N}\|_1 \leq 3/N$. We can then use Lemma \ref{lemma:wass_err_moments} with distributions $q$ and $s$ along with the fact that $\|\vec{\tau}_{q}/\vec{N} - \vec{\tau}_s/\vec{N}\|_1 = \sum_{k=1}^N |(\vec{\tau}_s)_k - (\vec{\tau}_q)_k|/k \leq 3/N$ to give us the result since $N > 18/\epsilon$.
\end{proof}

\begin{remark}
Note that Algorithm \ref{alg:MomentMatching_full} can easily be adapted when the minimization problem in Line \ref{line:regression_mm} is solved approximately -- as is the case if projected subgradient descent methods are used. In particular, a constant factor approximation to the minimal loss increases the Wasserstein distance bound in Lemma \ref{lemma:mm_guarantee} by an $O(1)$ factor.
\end{remark}

\section{Efficient Chebyshev Moment Approximation}\label{sec:SDE_approx}
With Lemma \ref{lemma:mm_guarantee} in place, we are ready to prove our main results. To do so, we need to show how to efficiently approximate the first $N$ Chebyshev moments of a matrix $A$'s spectral density $s$, as required by Algorithm \ref{alg:MomentMatching_full}. Recall that the $k^\text{th}$ normalized Chebyshev moment of ${s}$ is equal to $\iprod{s}{\bar{T}_k} = \frac{1}{n}\tr(\bar{T}_k(A))$. We will prove that this trace can be approximated using Hutchinson's stochastic trace estimator, implemented with either exact or approximate matrix-vector multiplications with $A$.

This estimator requires repeatedly computing $\bar{T}_k(A)g$ for a random vector $g$, which is done using the standard three-term (forward) recurrence for the Chebyshev polynomials and requires a total of $k$ matrix-vector multiplications with $A$. We analyze the basic approach in Section \ref{sec:exactMatMultSDE}, which yields Theorem \ref{thm:exactMatMultSDE}.  Then in Section \ref{sec:approxMatMultSDE}, we argue that the approach is stable even when implemented with approximate matrix-vector multiplication, which yields Theorem \ref{thm:approxMatMultSDE}.

\subsection{Exact Matrix-Vector Multiplications}\label{sec:exactMatMultSDE}

Hutchinson's estimator is a widely used estimator to efficiently compute accurate estimates of $\tr(R)$ for any square matrix $R \in \R^{n \times n}$. Each instance of the estimator computes the quadratic form $g^\top Rg$ for a random vector $g \in \{-1, 1\}^n$ whose entries are Rademacher random variables. This an unbiased estimator for $\tr(R)$ with variance $\leq 2\|R\|_F^2$, and its error has been analyzed in several earlier results \cite{AvronToledo:2011,Roosta-KhorasaniAscher:2015}. We apply a standard high-probability bound from \cite{MeyerMuscoMusco:2020,RudelsonVershynin:2013}:

\begin{lemma}[Lemma 2, \cite{MeyerMuscoMusco:2020}]\label{lemma:hutchinsonsEst}\hspace{-.5em}\footnote{In \cite{MeyerMuscoMusco:2020} the lemma is stated with an assumption that $\ell > O(1/\delta)$. However, it is easy to see that the same claim holds without this assumption, albeit with a quadratically worse $\log(1/\delta)$ dependence. The proof follows from same application of the Hanson-Wright inequality used in that work.}
Let $R \in \R^{n \times n}$, $\delta \in (0, 1/2]$, $l \in \N$. Let $g^{(1)}, \dots, g^{(\ell)}\in \{-1, 1\}^{n \times n }$ be $\ell$ random vectors with i.i.d $\{-1,+1\}$ random entries. For a fixed constant $C$, with probability at least $1 - \delta$, 
\begin{align*}
\abs{\tr(R) - \frac{1}{\ell}\sum_{i = 1}^l (g^{(i)})^\top R g^{(i)}} \leq \frac{C\log(1/\delta)}{\sqrt{\ell}}\|R\|_F.
\end{align*}
\end{lemma}

For a polynomial $p \in \FcOne$ with degree $k$, applying Hutchinson's estimator to $R = p(A)$ requires computing $p(A)g$, which can always be done with $k$ matrix-vector multiplies with $A$. If $p(x)$ admits a recursive construction, like the Chebyshev polynomials, then this recurrence  can be used. Specifically, for the Chebyshev polynomials, we have:
\begin{align}
	\label{eq:mat_recur}
T_0(A)g &= g \hspace{2em} T_1(A)g = Ag\nonumber \\
T_k(A)g &= 2A \cdot T_{k-1}(A)g -  T_{k-2}(A)g \hspace{2em} \text{for $k \geq 2$}.
\end{align}

A moment estimation algorithm based on Hutchinson's estimator is stated as Algorithm \ref{alg:hutch_moments}.

\begin{algorithm}\caption{Hutchinson Moment Estimator}\label{alg:hutch_moments}
\begin{algorithmic}[1]
	\Require Symmetric $A\in \R^{n\times n}$ with $\|A\|_2 \leq 1$, degree $N\in 4\N^+$, number of repetitions $\ell \in \N^+$.
	
	\Ensure Approximation $\tilde{\tau}_k$ to moment $\frac{1}{n}\tr(\bar{T}_k(A))$ for all $k\in 1,\ldots, N$.	
			\State Draw $g^{(1)}, \dots, g^{(l)} \sim \text{Uniform}(\{-1, 1\}^n)$.
			\State For $k=1, \ldots, N$, $\tilde{\tau}_k \gets \frac{\sqrt{2/\pi}}{\ell n}\sum_{i = 1}^l (g^{(i)})^\top T_k(A)g^{(i)}$. \Comment{Computed using recurrence in \eqref{eq:mat_recur}}
	\State Return $\tilde{\tau}_1, \ldots, \tilde{\tau}_N$. 
\end{algorithmic}
\end{algorithm}
\textbf{Remark.} In total, Algorithm \ref{alg:hutch_moments} requires $N\cdot \ell$ matrix multiplications with $A$ since for each $i$ $T_1(A)g^{(i)}, \ldots, T_N(A)g^{(i)}$ can but computed using the same $N$ steps of the \eqref{eq:mat_recur} recurrence. It requires $O(n\ell N)$ additional runtime to compute and sum all inner products of the form $(g^{(i)})^TT_k(A)g^{(i)}$.

Our main bound on the accuracy of Algorithm \ref{alg:hutch_moments} follows:
%
%The last two steps of Algorithm \ref{alg:exactSDE} are identical to Algorithm \ref{alg:SDEpolynomial}. Accordingly, we can analyze its performance via Lemma \ref{lem:degreeAndErrorBoundsforWass}, as long as it approximates each Chebyshev moment $\tau_k$ up to additive error $N^{-3}$, as required in Algorithm \ref{alg:SDEpolynomial},
%then the polynomial $q$ returned will be $\epsilon$ close in Wasserstein-1 distance to $A$'s spectral density $s$. Specifically, we prove the following:
\begin{lemma}
	\label{lem:exact_mult_lem}
	If Algorithm \ref{alg:hutch_moments} is run with $\ell = \max\left(1, C\cdot \log^2({N}/{\delta})/(n\Delta^2)\right)$, where $C$ is a fixed positive constant, then with probability $1-\delta$ the approximate moments returned satisfy $|\tilde{\tau}_k - \frac{1}{n}\tr(\bar{T}_k(A))| \leq \Delta$ for all $k = 1,\ldots, N$.
\end{lemma}
\begin{proof}
Fix $k \in \{1,\ldots, N\}$. Note that $\frac{1}{n}\tr(\bar{T}_k(A)) = \frac{\sqrt{2/\pi}}{n}\tr({T}_k(A))$. Let $C$ be the constant from Lemma \ref{lemma:hutchinsonsEst}. If $\ell = \max\left(1, C^2\cdot \log^2({N}/{\delta})/(n\Delta^2)\right)$, then by that lemma we have that with probability at least $1- \delta/N$:
\begin{align*}
\abs{\tilde{\tau_k} - \frac{\sqrt{2/\pi}}{n}\tr({T}_k(A))} 
\leq  \frac{1}{n}\frac{C\log(N/\delta)}{\sqrt{\ell}}\|T_k(A)\|_F \leq \frac{C\sqrt{2/\pi}}{\sqrt{n}}\sqrt{\frac{\log(N/\delta)}{\ell}} \leq \Delta.
\end{align*}
The second to  last inequality follows from the fact that $\|T_k(A)\|_2 \leq 1$ and thus $\|T_k(A)\|_F \leq \sqrt{n}$.
Applying a union bound over all $k \in 1,\ldots, N$ gives the claim. 
\end{proof}

Theorem \ref{thm:exactMatMultSDE} immediately follows as a corollary of Lemma \ref{lem:exact_mult_lem} and Lemma \ref{lemma:mm_guarantee}. 
\begin{proof}[Proof of Theorem \ref{thm:exactMatMultSDE}]
	We implement Algorithm \ref{alg:MomentMatching_full} with Algorithm \ref{alg:hutch_moments} used as a subroutine to approximate the Chebyshev polynomial moments, which requires setting $\Delta = \frac{1}{N\ln(eN)}$. By Lemma \ref{lem:exact_mult_lem}, we conclude that we need to set $\ell = \max\left(1, CN^2 \log^2({N}/{\delta})\log^2(eN)/n\right)$. Then, by Lemma \ref{lemma:mm_guarantee}, setting $N = O(1/\epsilon)$ ensures that Algorithm \ref{alg:MomentMatching_full} returns a distribution $q$ which is $\epsilon$ close to $A$'s spectral density $s$ in Wasserstein distance.
\end{proof}

\subsection{Approximate Matrix-Vector Multiplications}\label{sec:approxMatMultSDE}
Algorithm \ref{alg:hutch_moments} assumes access to an oracle for computing exact matrix-vector multiplies with $A$. In this section, we show that the method continues to work well even when each term in Hutchinson's estimator, $g^\top T_k(A) g$, is computed using an approximate matrix-vector multiplication oracle for $A$ (see Definition \ref{def:epsAMV}). As discussed in Section \ref{sec:contrib}, the robustness of the estimator allows the approximate moment matching method to be applied in many settings where $A$ can only be access implicitly. It also forms the basis of our sublinear time algorithm for computing the spectral density of a normalized graph adjacency or Laplacian matrix, which are presented in the Section \ref{sec:graphSDE}.

To show that approximate matrix-vector multiplications suffice, we leverage well understood stability properties of the three-term forward recurrence for Chebyshev polynomials of the first kind \cite{Clenshaw:1955,MuscoMuscoSidford:2018}. These properties allows us to analyze the cumulative error when  $T_k(A)g$ is computed via this recurrence. Specifically, we analyze the following algorithm:

\begin{algorithm}\caption{Hutchinson Moment Estimator w/ Approximate Multiplications}\label{alg:hutch_moments_approx}
	\begin{algorithmic}[1]
		\Require Symmetric $A\in \R^{n\times n}$ with $\|A\|_2 \leq 1$, degree $N\in 4\N^+$, number of repetitions $\ell \in \N^+$, $\epsMV$-approximate matrix vector multiplication oracle $\amv$ for $A$ (see Definition \ref{def:epsAMV}).
		
		\Ensure Approximation $\tilde{\tau}_k$ to moment $ \frac{1}{n}\tr(\bar{T}_k(A))$ for all $k\in 1,\ldots, N$.	
		\For{$i = 1,\ldots,\ell$ iterations}
		\State Draw $g \sim \text{Uniform}(\{-1, 1\}^n)$.
		\State $\tilde{v}_0 \gets g$, $\tilde{v}_1 \gets \amv(A, g, \epsMV)$. 
		\State $\tilde{\tau}_{1,i}\gets g^T\tilde{v}_1$
%		\Comment{$\epsMV$-approximate matrix vector mult. oracle for $A$}
		\For{$k = 2$ to $N$}
		\State $\tilde{v}_{k} \gets 2\cdot \amv(A, \tilde{v}_{k-1}, \epsMV)  - \tilde{v}_{k-2}$.
		\State $\tilde{\tau}_{k,i}\gets g^T\tilde{v}_k$
		\EndFor
		\EndFor
		\State For $k = 1, \ldots, N$, $\tilde{\tau}_k \gets \frac{1}{\ell}\sum_{i=1}^\ell \tilde{\tau}_{k,i}$.
		\State Return $\tilde{\tau}_1, \ldots, \tilde{\tau}_N$. 
	\end{algorithmic}
\end{algorithm}

Algorithm \ref{alg:hutch_moments_approx} assumes access to an approximate matrix-vector multiplication oracle for $A$ with error $\epsMV$ (recall Definition \ref{def:epsAMV}). Since $\|A\|_2 \leq 1$, for any vector $y$, we have that:
\begin{align}
	\|\amv(A, y, \epsMV) - Ay\|_2 \leq \epsMV \|y\|_2.
\end{align}
The algorithm uses this oracle to apply the recurrence from \eqref{eq:mat_recur}, approximately computing each $T_k(A)g$ for $k = 1, \ldots, N$, which in turn allows us to approximately compute $g^\top T_k(A)g$. Note that when $\epsMV = 0$, Algorithm \ref{alg:hutch_moments_approx} is exactly equivalent to Algorithm \ref{alg:hutch_moments}.

\textbf{Notation.} Analyzing this approach requires accounting for error accumulates across iterations. To do so, we introduce some basic notation. Let $v_k$ denote the true value of $T_k(A)g$, and let $\tilde{v}_k$ denote our computed approximation. We initialize the recurrence with $\tilde{v}_{-1} = \vec{0}$ and $\tilde{v}_{0} = v_0 = g$. For $k = 0, \ldots, N-1$, let $w_k = \amv(A,\tilde{v}_k,\epsMV)$ and note that $\|w_k - A\tilde{v}_k\|_2 \leq \epsMV \|\tilde{v}_k\|_2$.
In iteration $k$ of the recurrence, we compute $\tilde{v}_{k+1}$ by applying the recurrence:
\begin{align*}
	\tilde{v}_{k+1} \eqdef 2 w_k - \tilde{v}_{k-1}.
\end{align*}
For each $i \in 0, \ldots, N$ we denote:
\begin{itemize}
\item $\delta_k \eqdef v_k - \tilde{v}_k$, with $\delta_0 = \vec{0}$. This is the \emph{accumulated error} up to iteration $k$. 

\item $\xi_{k+1} \eqdef A\tilde{v}_k - w_k$, with $\xi_0 = 0$. $2\xi_{k+1}$ is the \emph{new error} introduced in iteration $k$ due to approximate matrix-vector multiplication.
\end{itemize}

As in Clenshaw's classic work \cite{Clenshaw:1955}, it can be shown that $\delta_{k}$ \emph{itself evolves according to a simple recurrence}, which ultimately lets us show that it can be expressed as a summation involving Chebyshev polynomials of the \emph{second} kind, which are easily bounded. Specifically, we have:
\begin{fact}\label{fact:deltaRecurrence}
	$\delta_1 = \xi_1$ and for $2 \leq k \leq N$, $\delta_k = 2A\delta_{k-1} - \delta_{k-2} + 2\xi_{k}$.
\end{fact}
\begin{proof}
	The claim for $\delta_1$ is direct since $v_0 = \tilde{v}_0$: we have $\delta_1 = {v}_1 - \tilde{v}_1 = Av_0 - w_0$. For $2 \leq k \leq N$, we prove the claim by writing the difference $\delta_k = v_k - \tilde{v}_k = v_k - 2(A\tilde{v}_{k-1} + \xi_k) + \tilde{v}_{k-2}$. We can then replace $v_k = 2Av_{k-1} - {v}_{k-2}$ and substitute in $(v_{k-1} - \tilde{v}_{k-1}) = \delta_{k-1}$ and $(v_{k-2} - \tilde{v}_{k-2}) = \delta_{k-2}$.
\end{proof}
The Chebyshev polynomials of the {second} kind are defined via the following recurrence:
\begin{definition}[Chebyshev Polynomials of the Second Kind]\label{def:chebyshevSecond}
For $k \in \N^{\geq 0}$ the $k$-th Chebyshev polynomial of the second kind $U_k(x)$ is given by 
\begin{align*}
U_0(x) &= 1 \hspace{2em} U_1(x) = 2x \\ 
U_k(x) &= 2 x \cdot U_{k-1}(x) - U_{k-2}(x)  \hspace{2em} \text{for $k \geq 2$}.
\end{align*}
We also define $U_{-1}(x) = 0$, which is consistent with the recurrence.
\end{definition}

Using these polynomials, we can characterize the accumulated error $\delta_k$ in terms of the error introduced in each of the prior iterations.
\begin{lemma}\label{lemma:chebyshevAccumulatedError}
For $k= 1,\ldots, N$, we have 
\begin{align}\label{eqn:deltaExpression}
\delta_k = U_{k-1}(A)\xi_1  +  2\sum_{i = 2}^k U_{k-i}(A)\xi_i.
\end{align}
\end{lemma}
\begin{proof}
We prove the lemma by induction on $j \leq k$. For $j = 0$, the lemma is trivial since $\delta_0 = 0$ by definition and $U_{-1}(A) = 0$. For $j = 1$, $\delta_1 = \xi_1 = U_0(A)\xi_1$. By Fact \ref{fact:deltaRecurrence}, for $2 \leq j < k$, we have:
\begin{align}
	\label{eq:delta_split}
\delta_j &= 2\xi_j + \underbrace{2A\delta_{j-1} - \delta_{j-2}}_{z_1}. 
\end{align} 
We can apply the inductive hypothesis on $z_1$ and recombine terms using Definition \ref{def:chebyshevSecond} to get:
\begin{align*}
z_1 &= 2A \cdot \brackets{U_{j-2}(A)\xi_1 +  2\sum_{i = 2}^{j-1} U_{j-1-i}(A)\xi_{i}} -U_{j-3}(A)\xi_1 - 2\sum_{i = 2}^{j-2} U_{j-2-i}(A)\xi_{i} \\ 
&= U_{j-1}(A)\xi_1 + {U_1(A) \cdot 2\xi_{j-1} + \sum_{i = 2}^{j-2} \brackets{2A U_{j-1-i}(A) - U_{j-2-i}(A)} \cdot 2\xi_{i}} \\ 
&= U_{j-1}(A)\xi_1 + \sum_{i = 2}^{j-1}U_{j-i}(A) \cdot 2\xi_{i}
\end{align*}
Noting that plugging into \eqref{eq:delta_split} and noting that $2\xi_j = 2U_0(A)\xi_j$ completes the proof. 
\end{proof}
Our goal is to use Lemma \ref{lemma:chebyshevAccumulatedError} to establish that $\delta_k$ is small because each $\xi_i$ is small. It is well known that the Chebyshev polynomials of the second kind satisfy the following bounds for any $k \in \N$: 
\begin{align}
	\label{eq:chebSecond_ub}
	|U_k(x)| &\leq k+1 && \text{for} & x&\in[-1,1]. 
\end{align}
This is the upper bound we need to proceed. Specifically, we will show that each estimator using Algorithm \ref{alg:hutch_moments_approx}, $g^\top \tilde{v}_k$, well approximates Hutchinson's estimator $g^\top T_k(A)g = g^\top v_k$.

\begin{claim}\label{claim:estimatorErrorHutchinsons}
For quantities $v_k, \tilde{v}_k$ and $0\leq \epsMV \leq 1/2k^2$, we have 
\begin{align*}
	\left|g^\top T_k(A)g - g^\top \tilde{v}_k\right| \leq 2\epsMV \cdot (k+1)^2 \|g\|_2^2.
\end{align*}
\end{claim}
\begin{proof}
By the definition of $\delta_k$, we have $|g^\top T_k(A)g - g^\top \tilde{v}_k| = |g^\top \delta_k|$. By Cauchy-Schwarz we can bound $|g^\top \delta_k| \leq \|g\|_2\|\delta_k\|_2$. We are left to bound $\|\delta_k\|_2$.
Applying  Lemma \ref{lemma:chebyshevAccumulatedError} and triangle inequality, we have
\begin{align*}
\|\delta_k\|_2 &\leq \|U_{k-1}(A)\|_2\|\xi_1\|_2 + \sum_{i =2}^k 2 \|U_{k-i}(A)\|_2\|\xi_i\|_2
\end{align*}
Then applying \eqref{eq:chebSecond_ub} and the fact that $\|A\|_2\leq 1$, we have $\|U_{k-i}(A)\|_2 \leq (k-i+1).$ Hence,
\begin{align*}
\|\delta_k\|_2 &\leq k\|\xi_1\|_2 + \sum_{i=2}^k 2(k-i+1)\|\xi_i\|_2 \leq \sum_{i=1}^k 2(k-i+1)\|\xi_i\|_2.
\end{align*}
Using that $\xi_i \leq \epsMV\|\tilde{v}_{i-1}\|_2$, and that $\|T_i(A)\|_2 \leq 1$ for all $i$ and thus $\|v_i\|_2 \leq \|g\|_2$, we have:
\begin{align*}
\|\delta_k\|_2&\leq \sum_{i=1}^k 2(k-i+1)\epsMV\|\tilde{v}_{i-1}\|_2 \leq 2\epsMV\sum_{i=1}^k (k-i+1)(\|v_{i-1}\|_2 + \|\delta_{i-1}\|_2) \\
&\leq \epsMV k(k+1)\left(\|g\|_2 + \max_{i < k}\|\delta_{i}\|_2\right).
\end{align*}
Inducting on $\delta_j$ for $j \leq k$ gives us $\|\delta_k\|_2\leq 2\epsMV(k+1)^2\|g\|_2$, which completes the proof.
\end{proof}

\begin{lemma}
	\label{lem:approx_mult_lem}
	If Algorithm \ref{alg:hutch_moments_approx} is run with $\ell = \max\left(1, C\cdot \log^2({N}/{\delta})/(n\Delta^2)\right)$ and $\epsMV = \Delta/4N^2$, where $C$ is a fixed positive constant, then with probability $1-\delta$ the approximate moments returned satisfy $|\tilde{\tau}_k - \frac{1}{n}\tr(\bar{T}_k(A))| \leq \Delta$ for all $k = 1,\ldots, N$.
\end{lemma}
\begin{proof}
Fix $k \in \{1,\ldots, N\}$. Let $g^{(1)}, \dots, g^{(\ell)}$ be the random vectors drawn in the outer for-loop of Algorithm \ref{alg:hutch_moments_approx}. Let $\{\tilde{v}_k^{(i)}\}_{i \in [\ell]}$ be the $\ell$ vectors computed by the inner for-loop and let $\{\delta_k^{(i)} \eqdef \tilde{v}_k^{(i)} - T_k(A)g^{(i)} \}_{i \in [\ell]}$ be the $\ell$ error vectors. Recalling that $\frac{1}{n}\tr(\bar{T}_k(A)) = \frac{\sqrt{2/\pi}}{n}\tr(T_k(A))$, we have:
\begin{align*}
\abs{\tilde{\tau}_k - \frac{\sqrt{2/\pi}}{n}\tr(T_k(A))} 
&\leq {\frac{\sqrt{2/\pi}}{n\ell}\sum_{i = 1}^\ell \left|(g^{(i)})^\top\delta_k^{(i)}\right|} + \abs{\frac{\sqrt{2/\pi}}{n\ell}\sum_{i = 1}^\ell (g^{(i)})^\top T_k(A)g^{(i)}- \frac{1}{n}\tr(T_k(A)) }
\end{align*}
Applying Claim \ref{claim:estimatorErrorHutchinsons} and Lemma \ref{lemma:hutchinsonsEst}, with probability at least $1 - \delta/N$, we thus have
\begin{align*}
|\tilde{\tau}_k - \frac{1}{n}\tr(\bar{T}_k(A))| \leq 2(k+1)^2 \epsMV\cdot \frac{\sqrt{2/\pi}}{n\ell}\sum_{i = 1}^\ell\|g^{(i)}\|^2_2 + \Delta/2 \leq \Delta/2 + \Delta/2.
\end{align*}
The last inequality follows from the fact that $\|g^{(i)}\|_2^2 = n$ for all $i \in [\ell]$, and the choice of $\epsMV = \Delta/4N^2$. Applying a union bound over all $k = 1,\ldots, N$ gives the claim. 
\end{proof}

Theorem \ref{thm:approxMatMultSDE} immediately follows. 
\begin{proof}[Proof of Theorem \ref{thm:approxMatMultSDE}]
	We implement Algorithm \ref{alg:MomentMatching_full} with Algorithm \ref{alg:hutch_moments_approx} used as a subroutine to approximate the Chebyshev polynomial moments, which requires setting $\Delta = \frac{1}{N\ln(eN)}$. By Lemma \ref{lem:approx_mult_lem}, we conclude that we need to set $\ell = \max\left(1, CN^2 \log^2({N}/{\delta})\log^2(eN)/n\right)$ and $\epsMV = 1/(4N^3\ln(eN))$. Then, by Lemma \ref{lemma:mm_guarantee}, setting $N = O(1/\epsilon)$ ensures that Algorithm \ref{alg:MomentMatching_full} returns a distribution $q$ which is $\epsilon$ close to $A$'s spectral density $s$ in Wasserstein distance.
\end{proof}

\paragraph{Improving the number of matrix-vector multiplications.} We currently require the error bound in Algorithm \ref{alg:MomentMatching_full} for estimating the Chebyshev moments to be the same for each of the $N$ moments, i.e., parameter $\Delta = (N\ln(eN))^{-1}$. We note that the number of matrix-vector multiplications can be improved slightly in Theorems \ref{thm:approxMatMultSDE} and  \ref{thm:exactMatMultSDE}, potentially by a factor of $\log^2(1/\epsilon)$ for small $n$. This can be achieved by requiring a different error bound for estimating each moment. 
Specifically, we modify the requirement in Algorithm \ref{alg:MomentMatching_full} for the estimate $\tilde{\tau}_k$ of the $k$-th normalized Chebyshev moment $\frac{1}{n}\tr(\Tbar_k(A))$ to have error $|\tilde{\tau}_k - \frac{1}{n}\tr(\Tbar_k(A))| \leq (k/N^5)^{1/4}$. Plugging this into Lemma \ref{lem:exact_mult_lem}, we require at most $\sum_{k=1}^N \max(1, CN^{2.5}\log^2(N/\delta)/(n\sqrt{k}))$ matrix-vector multiplications to estimate the $N$ moments, where $C$ is a fixed constant. For comparison to the bounds in Theorems \ref{thm:exactMatMultSDE} and \ref{thm:approxMatMultSDE}, the above bound decreases linearly in $n$ until $n \geq CN^{2}\log^2(N/\delta)$ and for very large $n$ is bounded by $O(1/N)$. In the regime where $n$ is small, e.g., when $n \leq CN^2\log^2(N/\delta)$, the bounds from the theorems give $O(N^3\log^2(N/\delta)\log^2(eN)/n)$ matrix-vector multiplications, whereas the above bound simplifies to at most $O(N^3\log^2(N/\delta)/n)$ multiplications, saving a $O(\log^2(N)) = O(\log^2(1/\epsilon))$ factor. Lemma \ref{lem:approx_mult_lem} can be adapted identically to give the same bound in the approximate matrix-vector multiplication case. 
To give intuition for the Wasserstein error of the resulting density, if the density estimate $q$ output by Algorithm \ref{alg:MomentMatching_full} satisfied the requirement that $|\iprod{q}{\Tbar_k} - \frac{1}{n}\tr(\Tbar_k(A))| \leq (k/N^5)^{1/4}$ for $k \in 1, \dots, N$, then we have by Lemma \ref{lemma:wass_err_moments} that $W_1(s, q) \leq 36/N + ({2}/{N^{5/4}}) \cdot \sum_{k =1}^N k^{-3/4} \leq 36/N + 8/N = O(1/N)$. This intuition can be used to adapt the proof of Lemma \ref{lemma:mm_guarantee} to show that Algorithm \ref{alg:MomentMatching_full} with moment guarantees as mentioned output a density $q$ such that $W_1(s, q) \leq O(1/N)$.

\section{Sublinear Time Methods for Graphs}\label{sec:graphSDE}

With the proof of Theorem \ref{thm:approxMatMultSDE} in place, we are now ready to state our sublinear time result for adjacency matrices of graphs. The significance of Theorem \ref{thm:approxMatMultSDE} is that it allows for the approximate Chebyshev moment matching method in  Algorithm \ref{alg:MomentMatching_full} to be combined with any randomized algorithm for approximating matrix-vector multiplications with $A$. In this section we prove Theorem \ref{thm:graphSDE} by showing that for the normalized adjacency matrix of any undirected, un-weighted graph, such an algorithm can actually be implemented in \emph{sublinear time}, leading to a sublinear time spectral density estimation (SDE) algorithm for computing graph spectra from these matrices.

\textbf{Computational Model.} Let $A \in \R^{n \times n}$ be the adjacency matrix for an unweighted, $n$-vertex graph $G = (V, E)$ and let $\bar{A} = D^{-1/2}AD^{-1/2}$ be the {symmetric} normalized adjacency matrix, where $D$ is an $n\times n$ diagonal matrix containing the degree of each vertex in $V$. For a node $i$, let $\Nc(i) = \{j:(j,i)\in E\}$ denote the set of $i$'s neighboring vertices. We assume a computational model where we can 1) uniformly sample a random vertex in constant time, 2) uniformly sample a random neighbor of any vertex $i$ in constant time, and 3) for a vertex $i$ with degree $d_i$, read off all neighbors of $i$ in $O(d_i)$ time. A standard adjacency list representation of the graph would allow us to perform these operations but weaker access models would also suffice.\footnote{E.g., random crawl access to a network \cite{KatzirLibertySomekh:2011}. We also note that, if desired, assumption 3) can be removed entirely with a small logarithmic runtime overhead, as long as we know the degree of $i$. Specifically, 3) can be implemented with $O(d_i\log n)$ calls to 2): we simply randomly sample neighbors until all $d_i$ are found.
%. First use the mark-and-recapture method \cite{KatzirLibertySomekh:2011} to estimate $d_i$ based on a random sample of $i$'s neighbors and then its neighbors can be collected by repeatedly sampling until all $d_i$ are found. 
A standard analysis of the coupon collector problem \cite[Section 3.6,][]{motwaniraghavan1995} shows that that the expected number of samples will be $O(d_i\log d_i) \leq O(d_i\log n)$.}

% While these are the only types of access to the graph we need to state our algorithm for Theorem \ref{thm:graphSDE} from Section \ref{sec:contrib}, we can also assume that we can compute all the neighbors of any vertex $i$ in $O(d_i \log n)$ time in expectation by simply making $O(d_i\log n)$ accesses of type 2) for vertex $i$. 
%Since our algorithm for Theorem \ref{thm:graphSDE} is allowed to fail with probability $\delta > 0$, we can absorb the failure probability of computing the neighbors of any vertex $i$ over all such ``neighborhood'' queries by using a sufficiently large constant $C > 0$ and taking a union bound.

Using this model for accessing the adjacency matrix, we show that, for any $\epsMV \in (0, 1)$ and failure probability $\delta \in (0, 1)$, an $\epsMV$-approximate matrix-vector multiplication oracle for $\bar{A}$ can be implemented in $O(n\epsMV^{-2}\log({1}/{\delta}))$ time. Via Theorem \ref{thm:approxMatMultSDE}, this immediately yields an algorithm for computing an SDE that is $\epsilon$ close in Wasserstein-1 distance to $\bar{A}$'s spectral density in roughly $\tilde{O}(n /\epsilon^7)$ time for sufficiently large $n$, and at most  $\tilde{O}(n /\epsilon^{9})$ time, for fixed $\delta$ where the $\tilde{O}(\cdot)$ hides factors of $\poly(\log(1/\epsilon))$. Our main result is stated as Theorem \ref{thm:graphSDE} in Section \ref{sec:contrib}.

The same algorithm can be used to approximate the spectral density of the normalized Laplacian of $G$ by a simple shift and scaling. Specifically, $\bar{A}$ can be obtained from the normalized Laplacian $\bar{L}$  via $\bar{A} = I - \bar{L}$, and the spectral density of $\bar{L}$, $s_{\bar{L}}(x)$ satisfies $s_{\bar{L}}(1-x) = s_{\bar{A}}(x)$, where $s_{\bar{A}}$ is the spectral density of $\bar{A}$. So if we obtain an $\epsilon$-approximate SDE $q$ for $\bar{A}$ by Theorem \ref{thm:graphSDE}, then the function $p$ satisfying $p(1-x) = q(x)$ is an $\epsilon$-approximate SDE for $s_{\bar{L}}$. We thus have:
\begin{restatable}{corollary}{laplacianSDE}\label{cor:laplacianSDE}
Given the the normalized adjacency matrix of $G$, there exists an algorithm that takes $O\left(n\poly\left(\frac{\log(1/\delta)}{\epsilon}\right)\right)$ expected time and outputs a density function $q$ that is $\epsilon$ close to the spectral density of the normalized Laplacian of $G$ with probability at least $1 - \delta$.
\end{restatable}

% \subsection{Approximate Matrix-Vector Multiplication for Adjacency Matrices -- FIXED}
% Consider the actual normalized adjacency matrix $A = D^{-1/2}\bar{A}D^{-1/2}$. The entries $ij$ entry in this matrix is either equal to $0$ or or $\frac{1}{\sqrt{d_id_j}}$, if $ij$ is an edge in the graph. To fix our proof, we basically need to show that: 
% 	\begin{align*}
% 		\sum_{i=1}^n d_i \|A^i\|_2^2 = O(n).
% 	\end{align*}
% To see that this is true, note that $\|A^i\|_2^2 = \sum_{j\in \mathcal{N}(i)}\frac{1}{d_id_j}$, so $d_i\|A^i\|_2^2 = \sum_{j\in \mathcal{N}(i)}\frac{1}{d_j}$, where $\mathcal{N}(i)$ is the set of neighbors of node $i$. So,
% \begin{align*}
% 	\sum_{i=1}^n d_i \|A^i\|_2^2 =\sum_{i=1}^n \sum_{j\in \mathcal{N}(i)}^n \frac{1}{d_j} = n.
% \end{align*}
% The last bound holds because every node $j$ appears in the neighborhood of exactly $d_j$ nodes.

% Btw, this requires us to sample each column by $\|A^i\|_2^2 = \sum_{j\in \mathcal{N}(i)}\frac{1}{d_id_j}$. I am not sure how to do that efficiently.

\subsection{Approximate Matrix-Vector Multiplication for Adjacency Matrices}
We implement an approximate matrix-vector multiplication oracle for $\bar{A}$ in Algorithm \ref{alg:amvGraphs}, which is inspired by a randomized matrix-multiplication method of \cite{drineas2006fast}. Throughout this section, let $\bar{A}^i$ denote the $i^\text{th}$ column of $\bar{A}$. 
Given a sampling budget $t \in \N$, the algorithm samples $t$ indices from ${1,\ldots, n}$ independently and with replacement -- i.e., the same index might be sample multiple times. For each index it samples, the algorithm decides to accept or reject the column corresponding to that index with some probability. To approximate $\bar{A}y$, the algorithm outputs the multiplication of the accepted columns, rescaled appropriately, with the corresponding elements of $y$. 
% \begin{algorithm}\caption{AMV Multiplication Oracle for Normalized Adjacency Matrices}\label{alg:amvGraphs}
% 	\begin{algorithmic}[1]
% 	\Require Normalized adjacency matrix $\bar{A} \in \R^{n \times n}$, degrees $[d_1, \dots, d_n]$, vector $y$, and parameter $t \in \N$.
% 	\Ensure A vector $z \in \R^n$ that approximates $\bar{A}y$.
% 	\State Initialize multiset $\mathcal{S} \gets \emptyset$
	 
% 		\For{$t$ iterations}
% 			\State Add random index $i \in 1, \ldots, n$ to $\Sc$, where $i$ is sampled with prob. $p_i = \frac{\|\bar{A}^i\|_2^2}{\|\bar{A}\|_F^2} = \frac{1/d_i}{\sum_{j = 1}^n 1/d_j}$
% 		\EndFor
	
% 	\State \textbf{return} $z \gets \sum_{i \in \mathcal{S}} \frac{1}{tp_i} \cdot y_i\bar{A}^i $
% 	\end{algorithmic}
% \end{algorithm}
\begin{algorithm}\caption{AMV Multiplication Oracle for Normalized Adjacency Matrices}\label{alg:amvGraphs}
	\begin{algorithmic}[1]
	\Require Normalized adjacency matrix $\bar{A} \in \R^{n \times n}$, degrees $[d_1, \dots, d_n]$, $y\in \R^n$, and parameter $t \in \N$.
	\Ensure A vector $z \in \R^n$ that approximates $\bar{A}y$.
	\State Initialize $z \gets \vec{0}$.
	 
		\For{$t$ iterations}
			\State Sample a node $j$ uniformly at random from $\{1, \ldots, n\}$.
			\State Sample a neighbor $i \in \Nc(j)$ uniformly at random.
			\State Sample $x$ uniformly at random from $[0,1]$.
			\If {$x\leq \frac{1}{d_i}$}
			\State $w\gets \frac{1}{p_i} \cdot y_i\bar{A}^i$ where $p_i = \frac{1}{nd_i}\sum_{j \in \Nc(i)} \frac{1}{d_j}$.
			\Else
			\State $w\gets \vec{0}$.
			\EndIf
			\State $z \gets z + w$.
			% \State Add random index $i \in 1, \ldots, n$ to $\Sc$, where $i$ is sampled with prob. $p_i = \frac{\|\bar{A}^i\|_2^2}{\|\bar{A}\|_F^2} = \frac{1/d_i}{\sum_{j = 1}^n 1/d_j}$
		\EndFor
	
	\State \textbf{return} $\frac{1}{t}z$
	\end{algorithmic}
\end{algorithm}

The following lemma bounds the expected squared error of Algorithm \ref{alg:amvGraphs}'s:
\begin{lemma}\label{lemma:graphAMMVariance}
%Let $\bar{A} \in \R^{n \times n}$ be a matrix, $y \in \R^n$ be a vector and $t \in \N^+$ be a fixed constant. 
Let $z \in \R^n$ be the output of Algorithm \ref{alg:amvGraphs} with sampling budget $t$. We have:
\begin{align*}
	\Exp{\|\bar{A}y - z\|_2^2} = \frac{n}{t}\|y\|_2^2 - \frac{1}{t}\|\bar{A}y\|_2^2
\end{align*}
\end{lemma}
\begin{proof}	
Let $b$ denote $b = \bar{A}y$. Consider a single iteration of the main loop in Algorithm \ref{alg:amvGraphs}, which generates a vector $w$ that is added to $z$.
Let $X_i$ be an indicator random variable that is $1$ if $w$ is set to a scaling of $\bar{A}^i$ on that iteration, and $0$ otherwise. $X_i = 1$ if and only if 1) a neighbor of $i$ is sampled at Line 3 of the algorithm, 2) $i$ is sampled at Line 4 of the algorithm, and 3) the uniform random variable $x$ satisfies $x < 1/d_i$. So, we see that $\Prob{X_i = 1}$ is exactly equal to $p_i = \frac{1}{nd_i}\sum_{j \in \Nc(i)} \frac{1}{d_j}$. It follows that, by the time we reach Line 11, $w$ is an unbiased estimator for $b$. I.e.,  $\Exp{w} = b$. Of course, this also implies that $\Exp{z} = b$. 

Our goal is to show that $\Exp{\|b - z\|^2} = \frac{n}{t}\|y\|_2^2 - \frac{1}{t}\|b\|_2^2$. Since the random vector $b - z$ has mean zero and is the average of $t$ i.i.d. copies of the mean zero random vector $b-w$, it suffices that show:
\begin{align}
	\label{eq:main_var_goal}
	\Exp{\|b - w\|_2^2} = n \|y\|_2^2 - \|b\|_2^2.
\end{align}
By linearity of expectation and the fact that  $\Exp{w} = b$, we have
\begin{align*}
\Exp{\|b - w\|_2^2} = \|b\|_2^2 + \Exp{\|w\|_2^2} - 2\langle\Exp{w},b\rangle = \Exp{\|w\|_2^2} -  \|b\|_2^2.
\end{align*} 
So to prove \eqref{eq:main_var_goal}, we need to show that $\Exp{\|w\|_2^2} = n \|y\|_2^2$. We expand $w$ in terms of the indicator random variables $X_1, \dots, X_n$. Notice that since we only sample one column in each iteration, the random variable $X_iX_j = 0$ for all $i \neq j$. Thus, we have:
\begin{align*}
\Exp{\|w\|_2^2} &= \sum_{k=1}^n\ExpBig{{\sum_{i, j \in [n]} \frac{X_iX_j}{p_ip_j}(\bar{A}^iy_i)_k (\bar{A}^jy_j)_k}} = \sum_{k=1}^n\ExpBig{{\sum_{i = 1}^n \frac{X_i^2}{p_i^2} \cdot (\bar{A}^i y_i)^2_k}} \\ 
&= \sum_{i=1}^n\frac{1}{p_i} \cdot \|\bar{A}^i y_i\|_2^2 =  \sum_{i=1}^n ny_i^2 = n\|y\|_2^2
\end{align*}
In the last equalities we used the fact that $\Exp{X_i^2} = p_i$ and that, for a normalized graph adjacency matrix, $\|\bar{A}^i\|_2^2 = \sum_{j \in \Nc(i)} \frac{1}{d_id_j} = np_i$. This proves \eqref{eq:main_var_goal}, from which we conclude the lemma.  
\end{proof}

% We state a lemma by \cite{drineas2006fast}, which bounds the variance Algorithm \ref{alg:amvGraphs}'s error.
% \begin{lemma}[Lemma 4, \cite{drineas2006fast} -- rephrased for our setting]\label{lemma:graphAMMVariance}
% Let $\bar{A} \in \R^{n \times n}$ be a matrix, $y \in \R^n$ be a vector and $t \in \N^+$ be a fixed constant. Let $z \in \R^n$ be the output of Algorithm \ref{alg:amvGraphs} after sampling $t$ indices with probabilities $[p_1, \dots, p_n]$, then $$\Exp{\|\bar{A}y - z\|_2^2} = \sum_{i = 1}^n \frac{y_i^2\|\bar{A}^i\|_2^2}{tp_i} - \frac{1}{t}\|\bar{A}y\|_2^2.$$
% \end{lemma}

Using Lemma \ref{lemma:graphAMMVariance}, we show that there is an $\epsMV$-approximate matrix-vector oracle for $\bar{A}$ based on Algorithm \ref{alg:amvGraphs} with success probability at least $1 - \delta$ that runs in $O(n\epsMV^{-2}\log^2(\frac{1}{\delta}))$ time.
\begin{proposition}\label{prop:laplacianAMM}
Let $\bar{A} \in \R^{n \times n}$ be the symmetric normalized adjacency matrix of an $n$-vertex graph $G$ and let $\epsMV, \delta \in (0,1)$ be fixed constants. There is an algorithm that, given a vector $y \in \R^n$, and access to $G$ as described above, takes $O(n\epsMV^{-2}\log(\frac{1}{\delta}))$ expected time and outputs a vector $z \in \R^n$ such that $\|z - \bar{A}y\|_2 \leq \epsMV\|y\|_2$ with probability at least $1 - \delta$.
\end{proposition}
\begin{proof}
%	Notice that $\|\bar{A}^i\|_2^2 = \frac{1}{d_i} \sum_{j \in \Nc(i)} \frac{1}{d_j}$. 
By Lemma \ref{lemma:graphAMMVariance}, we have that $\Exp{\|\bar{A}y - z\|_2^2} \leq \frac{n}{t}\|y\|_2^2$. Fix $t = 48n\epsMV^{-2}$. Then, by Lemma \ref{lemma:graphAMMVariance} and Markov's inequality, we have that when Algorithm \ref{alg:amvGraphs} is called on $\bar{A}$ with parameter $t$,
\begin{align}\label{eqn:AMMfailureprob}
\Prob{\|\bar{A}y - z\|_2 > \frac{\epsMV}{4}\|y\|_2} &\leq \frac{16n\|y\|_2^2}{t\epsMV^2\|y\|^2_2} \leq \frac{1}{4}.
\end{align}
% $$\|Ay - z\|_2^2 \leq \frac{8\log(\frac{1}{\delta})}{t}\|A\|_F^2\|y\|_2^2 \leq \gamma^2\|y\|_2^2$$ with probability at least $1 - \delta$, where $t$ is the number of columns sampled and $z$ is the output of the algorithm. 
In order improve our success probability from $3/4$ to $1 - \delta$, we use the standard trick of repeating the above process $r = c\log(\frac{1}{\delta})$ times for a constant $c$ to be fixed later. Let $z_1, \dots, z_r \in \R^n$ be the output of running Algorithm \ref{alg:amvGraphs} $r$ times with parameter $t$. We can return as our estimate for $\bar{A}y$  the first $z_i$ such that $\|z_i - z_j\|_2 \leq \frac{\epsMV}{2}\|y\|_2$ for at least $r/2+1$ vectors $z_j$ from $z_1, \dots, z_n$. 

To see why this works, note that a Chernoff bound can be used to claim that with probability $>1 - \delta$, at least $r/2+1$ vectors $z_j$ from $z_1, \dots, z_r$ have that $\|z_j - \bar{A}y\|_2 \leq \frac{\epsMV}{4}\|y\|_2$. 

By a triangle inequality we have that for all such $z_j$ and $z_k$, $$\|z_j - z_k\|_2 \leq \|z_j - \bar{A}y\|_2 + \|z_k - \bar{A}y\|_2 \leq \frac{\epsMV}{2}\|y\|_2.$$
Thus, the $z_i$ we picked must satisfy that $\|z_i - \bar{A}y\| \leq \frac{3\epsMV}{4}\|y\|_2$ by the triangle inequality.

All that remains is to bound the expected runtime of Algorithm \ref{alg:amvGraphs}, which we will run $r$ separate times. To do so, note that all index sampling can be done in just $O(t)$ time, since sampling a random vertex and a random neighbor of the vertex are assumed to be $O(1)$ time operations. The costly part of the algorithm is computing the sampled column $w$ at each iteration. In the case that $w = \vec{0}$, this cost is of course zero. However, when $w = \frac{1}{p_i}\bar{A}^i y_i$ for some $i$, computing the column and adding it to $z$ takes $O(d_i)$ time, which can be large in the worst case. Nevertheless, we show that it is small in expectation. This may seem a bit surprising: while nodes with high degree are more likely to be sampled by Line 4 in Algorithm \ref{alg:amvGraphs}, they are rejected with higher probability in Line 6. Formally, let $\nnz(w)$ denote the number of non-zero entries in $w$. We have:
\begin{align*}
\ExpBig{\nnz(w)} &= \sum_{i = 1}^n \nnz(\bar{A}^i) \cdot p_i = \sum_{i = 1}^n \sum_{j \in \Nc(i)} \frac{d_i}{n \cdot d_id_j} = \frac{1}{n}\sum_{i = 1}^n \sum_{j \in \Nc(i)} \frac{1}{d_j} = 1.
\end{align*}
The final equality follows from expanding the double sum: since node $j$ has exactly $d_j$ neighbors, $\frac{1}{d_j}$ appears exactly $d_j$ times in the sum. So $\sum_{i = 1}^n \sum_{j \in \Nc(i)} \frac{1}{d_j} = n$. 

We run Algorithm \ref{alg:amvGraphs} with $t = O(n/\epsMV^2)$ iterations, so it follows that the expected total sparsity of all $w$'s constructed equals $O(n/\epsMV^2)$, which dominates the expected running time of our method.
%We can get $1 - \delta$ success probability for the number of non-zeros in the sampled columns by having a large enough budget $b = 96n\epsMV^{-2}$, repeatedly running Algorithm \ref{alg:amvGraphs} until the number of non-zero entries in the sampled columns is within the budget and throwing out a run of the algorithm if it exceeds the budget. By Markov's inequality, with probability $1/2$, the number of non-zero entries will not exceed $b$. So with probability at least $1 - \delta$ we will have to re-run the algorithm only $O(\log(\frac{1}{\delta}))$ times. We can take a union bound on the failure condition in (\ref{eqn:AMMfailureprob}) and the number of non-zero entries being within the budget $b$.

%Overall, our approach need to make $O(\log^2(\frac{1}{\delta}))$ calls to Algorithm \ref{alg:amvGraphs}. The $r^2$ pairwise comparisons can be done in $nr^2$ time and the running time of computing the output vector in Algorithm \ref{alg:amvGraphs} is proportional to the number of non-zeros in the sampled columns, i.e. $O(n\epsMV^{-2})$. 
\end{proof}

\begin{proof}[Proof of Theorem \ref{thm:graphSDE}]
The accuracy and running time claim follows from combining the $\epsMV$-approximate vector multiplication oracle described in Proposition \ref{prop:laplacianAMM} with Algorithm \ref{alg:MomentMatching_full}, which is analyzed in Theorem \ref{thm:approxMatMultSDE}.
\end{proof}
As discussed in the introduction, Cohen et al. \cite{Cohen-SteinerKongSohler:2018} prove a result which matches the guarantee of Theorem \ref{thm:graphSDE}, but with runtime of $2^{O(1/\epsilon)}$ -- i.e., with \emph{no dependence} on $n$. In comparison, our result depends linearly on $n$, but only polynomially on $1/\epsilon$. In either case, the result is quite surprising, as the runtime is \emph{sublinear} in the input size: $A$ could have up to $O(n^2)$ non-zero entries.

\section{Experiments}
\label{sec:experiments}
\begin{figure}[!t]
\centering
\begin{subfigure}
  \centering
  \includegraphics[width=0.325\linewidth]{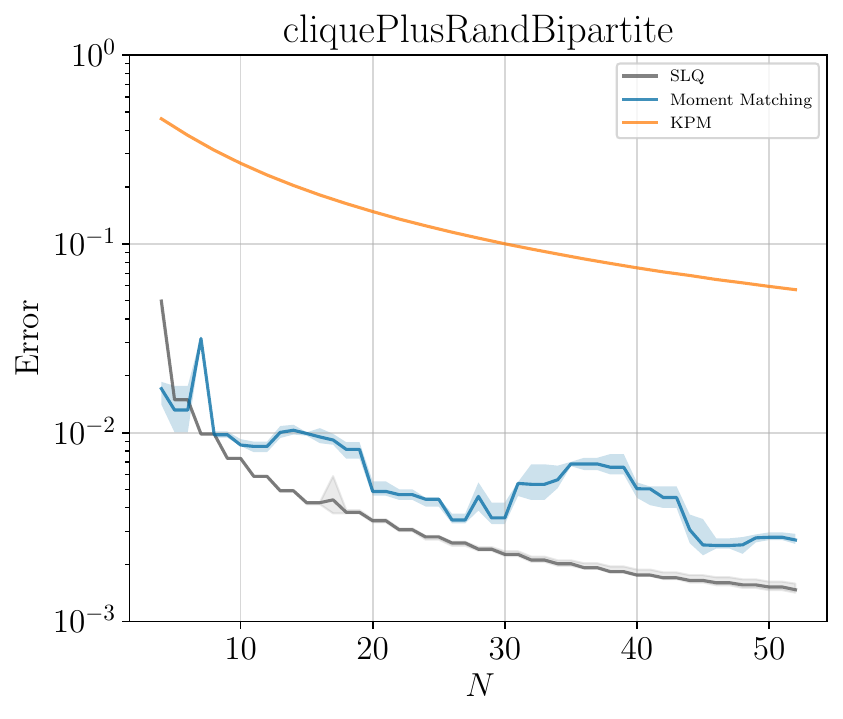}
  % \caption{Plot of relative error for 1000 queries, sorted by the magnitude of error for k-means++ and ScaNN respectively.}
  \label{fig:sdeRep-cliquePlusRandBipartite}
\end{subfigure}%
\begin{subfigure}
  \centering
  \includegraphics[width=0.312\linewidth]{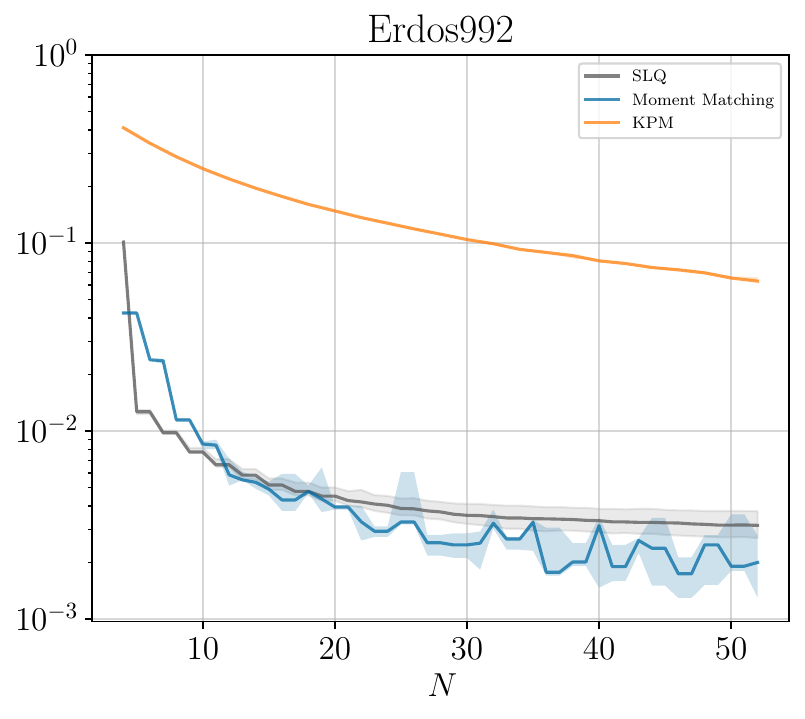}
  % \caption{Plot of relative error for 1000 queries, sorted by the magnitude of error for k-means++ and ScaNN respectively.}
  \label{fig:sdeRep-Erdos992}
\end{subfigure}%
\begin{subfigure}
  \centering
  \includegraphics[width=0.312\linewidth]{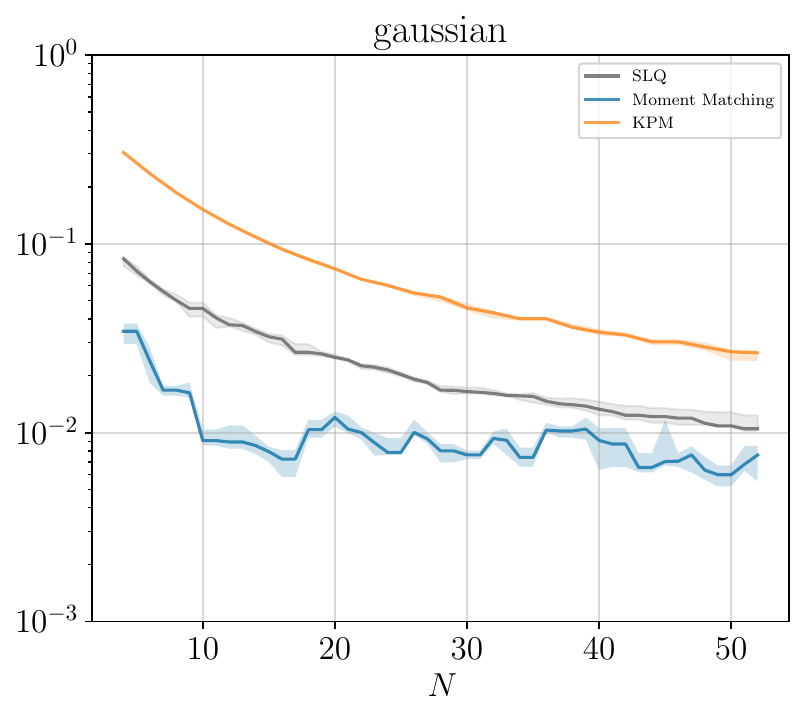}
  % \caption{Plot of relative error for 1000 queries, sorted by the magnitude of error for k-means++ and ScaNN respectively.}
  \label{fig:sdeRep-gaussian}
\end{subfigure}%
\newline
\begin{subfigure}
  \centering
  \includegraphics[width=0.325\linewidth]{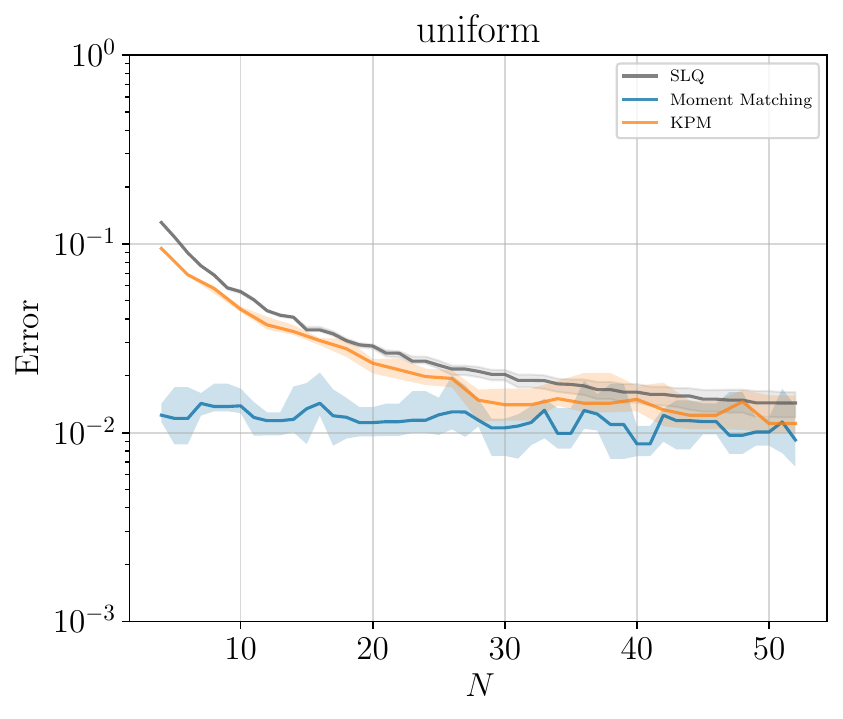}
  % \caption{Plot of relative error for 1000 queries, sorted by the magnitude of error for k-means++ and ScaNN respectively.}
  \label{fig:sdeRep-uniform}
\end{subfigure}%
\begin{subfigure}
  \centering
  \includegraphics[width=0.312\linewidth]{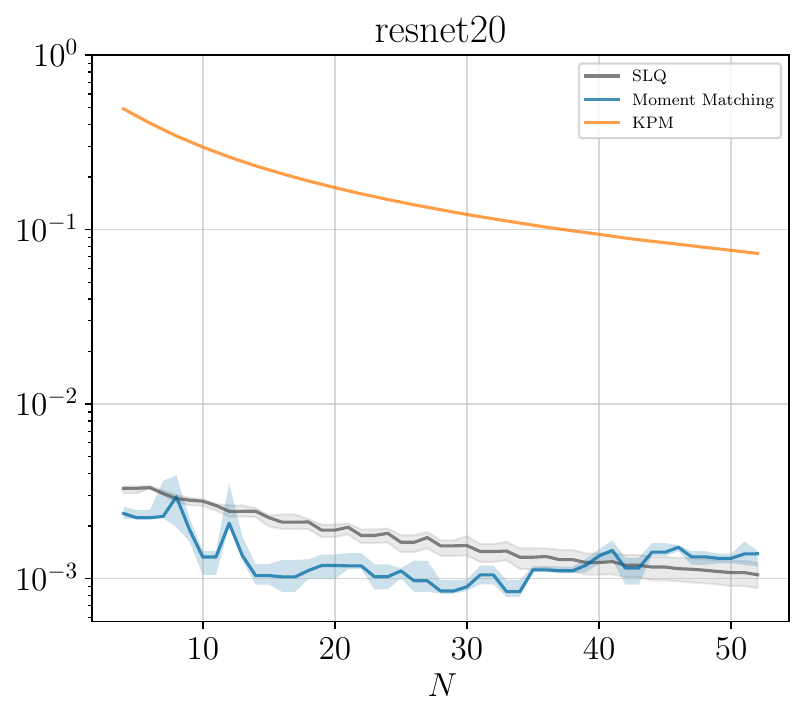}
  % \caption{Plot of relative error for 1000 queries, sorted by the magnitude of error for k-means++ and ScaNN respectively.}
  \label{fig:sdeRep-resnet20}
\end{subfigure}%
\begin{subfigure}
  \centering
  \includegraphics[width=0.312\linewidth]{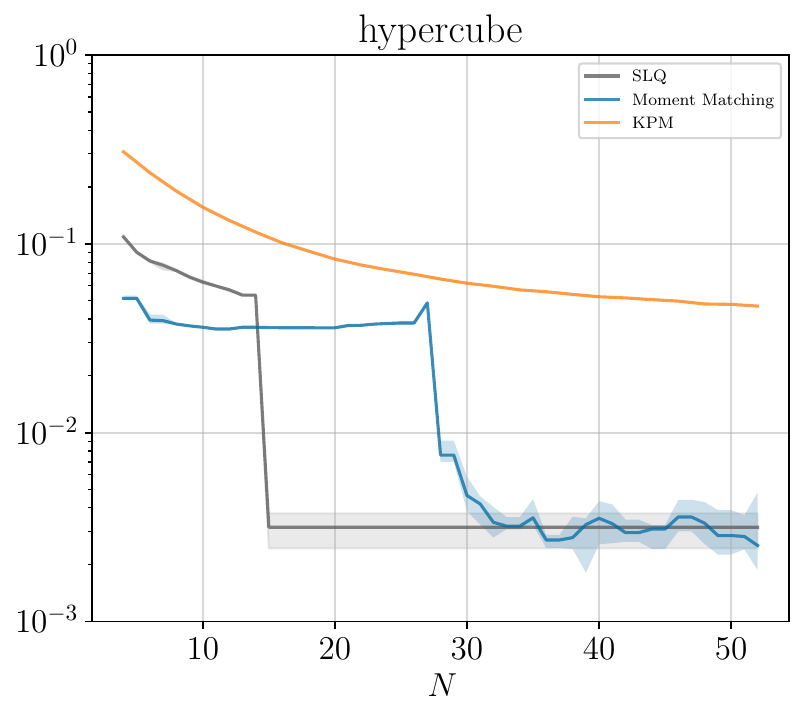}
  \label{fig:sdeRep-hypercube}
\end{subfigure}%
\caption[Exact Matrix-Vector Multplication]{Wasserstein error of density estimate resulting from approximate Chebyshev moment matching method (MM), the Jackson damped kernel polynomial method (KPM) and Stochastic Lanczos Quadrature (SLQ) method. For MM and KPM, Hutchinson's estimator is used to estimate the Chebyshev moments. The x-axis corresponds to the number of moments computed for MM and KPM, and the number of Lanczos iterations used for SLQ. All methods use 5 (random) starting vectors except for resnet20 and hypercube that use 1 starting vector, so the $x$-axis is directly proportional to the number of matrix-vector multiplications used by each method. Each experiment is repeated $10$ times; the solid line represents the median error of the $10$ trials and the shaded regions represent the first and third quartiles.}
\label{fig:error_plots}
\vspace{0.5em}
\hrule 
\vspace{-0em}
\end{figure}

We support our theoretical results by implementing our Chebyshev moment matching method (Algorithm \ref{alg:MomentMatching_full}). When using exact matrix-vector multiplications, the kernel polynomial method (KPM) of Algorithm \ref{alg:SDEpolynomial_full} and the stochastic Lanczos quadrature method (SLQ) studied in  \cite{ChenTrogdonUbaru:2021} have both been confirmed to work well empirically. So, one set of experiments is aimed at comparing these methods to the moment matching method (MM) implemented with exact matrix-vector multiplications. A second set of experiments evaluates the performance of the MM and KPM methods when implemented with approximate matrix-vector multiplies. Specifically, we use our sublinear time randomized method for multiplication by graph adjacency matrices from Section \ref{sec:graphSDE}.

\begin{figure}[!t]
\centering
\begin{subfigure}
  \centering
  \includegraphics[width=0.32\linewidth]{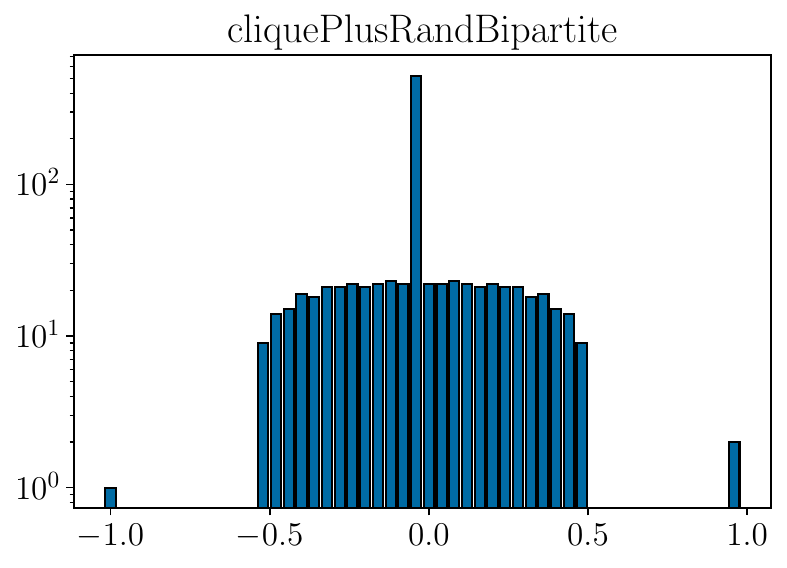}
  % \caption{Plot of relative error for 1000 queries, sorted by the magnitude of error for k-means++ and ScaNN respectively.}
  \label{fig:eigs-hist-cliquePlusRandBipartite}
\end{subfigure}%
\begin{subfigure}
  \centering
  \includegraphics[width=0.32\linewidth]{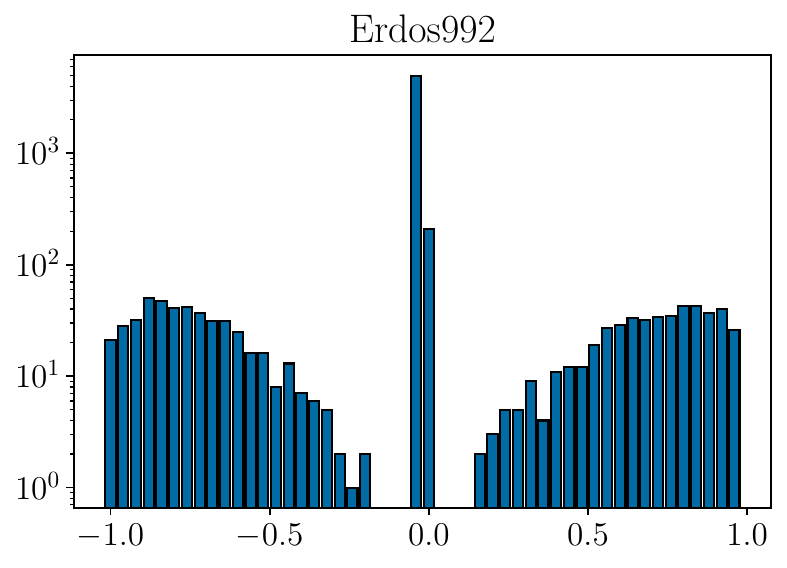}
  % \caption{Plot of relative error for 1000 queries, sorted by the magnitude of error for k-means++ and ScaNN respectively.}
  \label{fig:eigs-hist-Erdos992}
\end{subfigure}%
\begin{subfigure}
  \centering
  \includegraphics[width=0.32\linewidth]{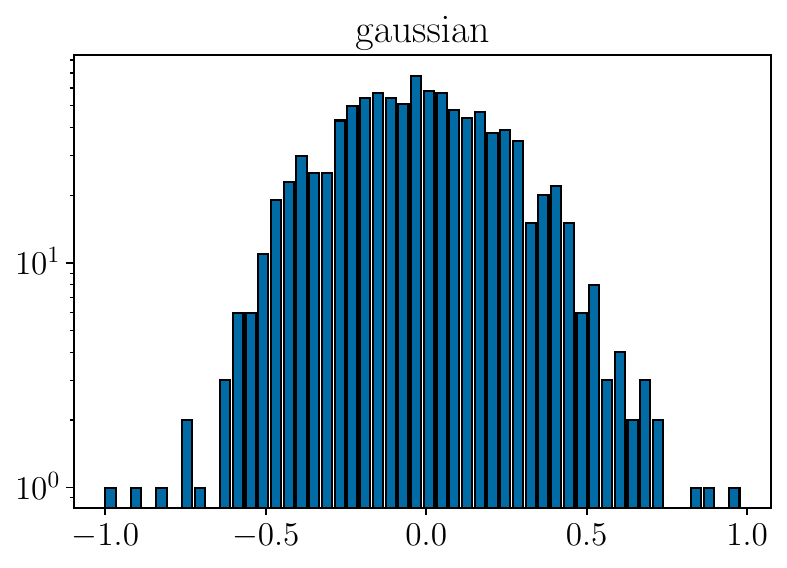}
  % \caption{Plot of relative error for 1000 queries, sorted by the magnitude of error for k-means++ and ScaNN respectively.}
  \label{fig:eigs-hist-gaussian}
\end{subfigure}%
\newline
\begin{subfigure}
  \centering
  \includegraphics[width=0.32\linewidth]{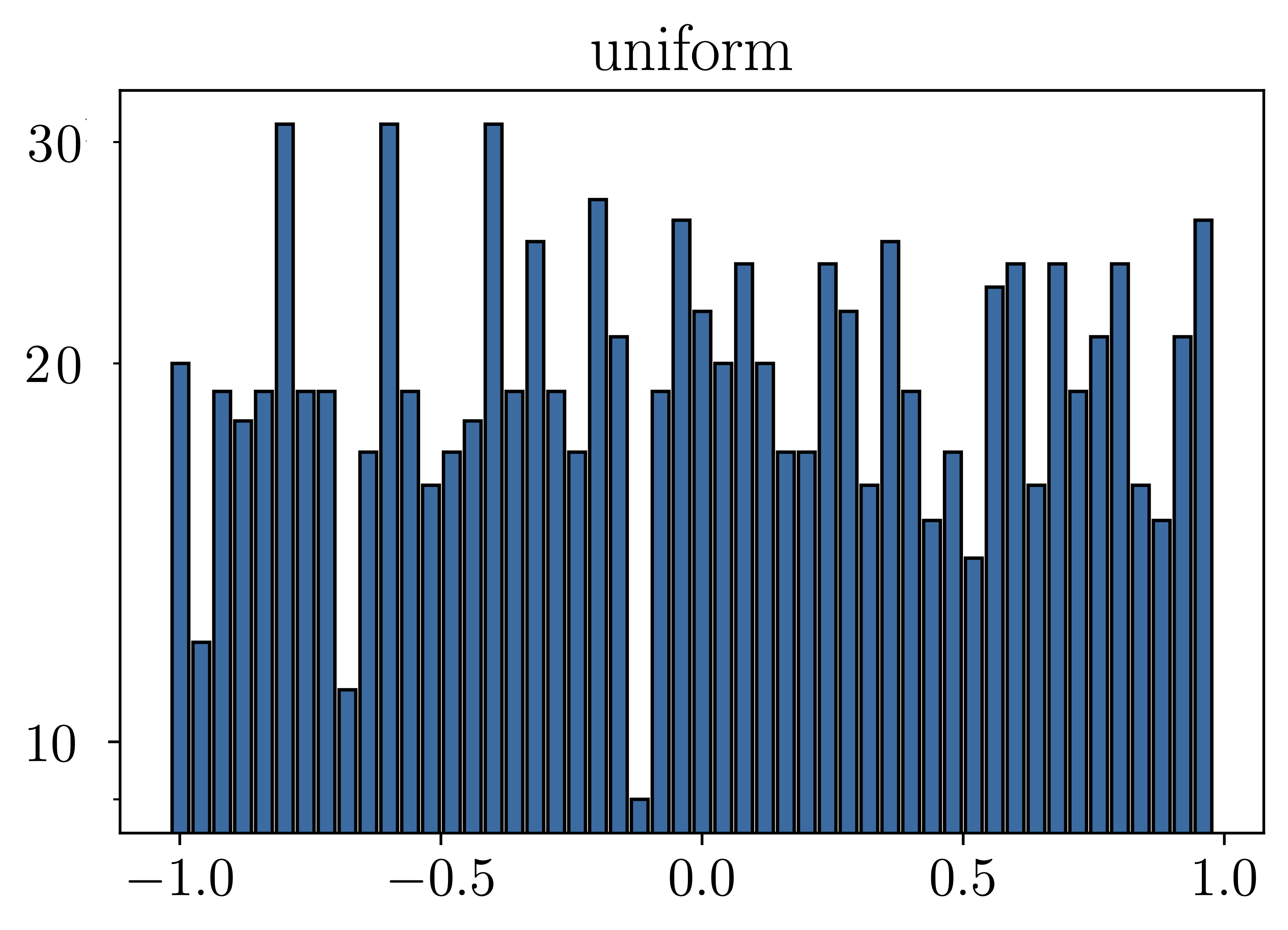}
  % \caption{Plot of relative error for 1000 queries, sorted by the magnitude of error for k-means++ and ScaNN respectively.}
  \label{fig:eigs-hist-uniform}
\end{subfigure}%
\begin{subfigure}
  \centering
  \includegraphics[width=0.32\linewidth]{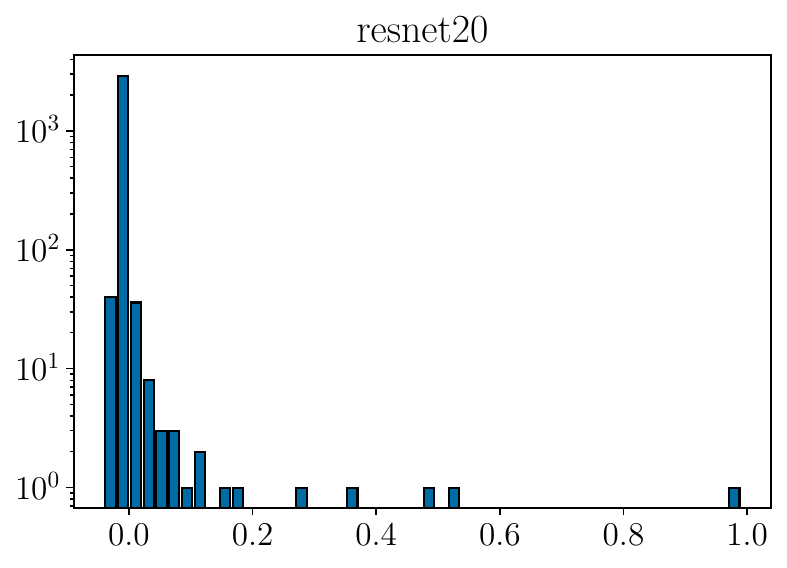}
  % \caption{Plot of relative error for 1000 queries, sorted by the magnitude of error for k-means++ and ScaNN respectively.}
  \label{fig:eigs-hist-resnet20}
\end{subfigure}%
\begin{subfigure}
  \centering
  \includegraphics[width=0.32\linewidth]{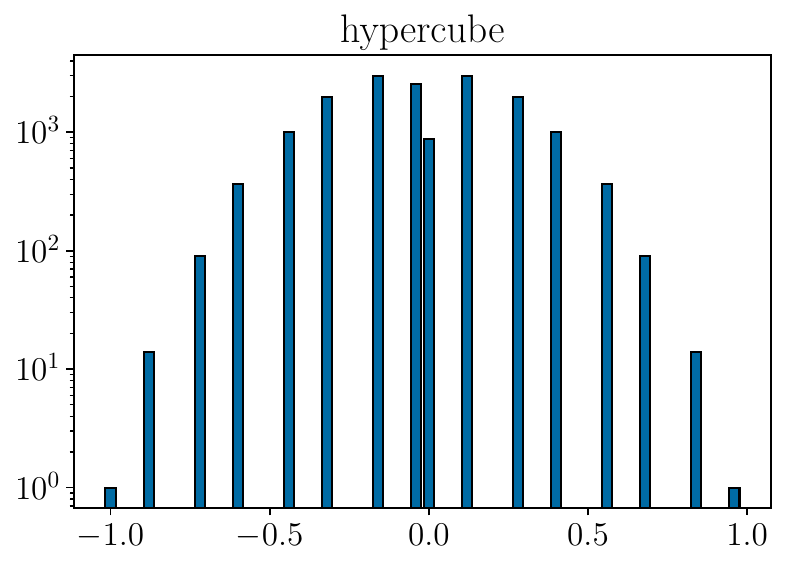}
  \label{fig:eigs-hist-hypercube}
\end{subfigure}%
\caption{Histograms of the eigenvalues of \texttt{cliquePlusRandBipartite}, \texttt{Erdos992}, \texttt{gaussian}, \texttt{uniform}, \texttt{resnet20} and \texttt{hypercube} using 50 equally spaced buckets.}
\label{fig:eigs-hist_plots}
\vspace{0.5em}
\hrule 
\vspace{-0em}
\end{figure}

We consider the normalized adjacency matrix of three graphs, two of which we construct and one which we obtain from a publicly available dataset for sparse matrices: 
\begin{itemize}
	\item \texttt{cliquePlusRandBipartite} is a graph with $10000$ vertices, partitioned into two disconnected components. The first component is a clique with 5000 nodes and the second is a bipartite graph with 2500 vertices in each partition, constructed by sampling each of the $2500^2$ possible edges independently with probability $0.05$. This graph has a normalized adjacency matrix with $\sim5000$ eigenvalues at $0$, two eigenvalues at $1$, one at $-1$ and the rest of its eigenvalues are roughly evenly spread out between $-0.5$ and $0.5$.
	\item \texttt{hypercube} is a 16384 vertex boolean hypercube graph on 14 bit strings.\footnote{A boolean hypercube contains a vertex for each distinct $b$ bit string, and an edge between two vertices if the corresponding strings differ on exactly 1 bit.} Its normalized adjacency matrix has eigenvalues at $-1, -\frac{6}{7}, \frac{-5}{7}, \ldots, 0, \ldots, \frac{6}{7}, 1$. The multiplicity of the $0$ eigenvalue is largest, with eigenvalues closer to $-1$ and $1$ having lower multiplicity.
	\item \texttt{Erdos992} is an undirected graph with 6100 vertices, containing 15030 edges from the sparse matrix suite of \cite{davis2011university}. Its normalized adjacency matrix has $\sim5000$ eigenvalues at $0$, one at $1$ and the rest evenly spread out between $-0.5$ and $0.5$.
\end{itemize}
We consider three additional matrices to evaluate the performance of MM against KPM and SLQ when exact matrix-vector multiplies are used to estimate the Chebyshev moments: 
\begin{itemize}
	\item \texttt{gaussian} is a $1000 \times 1000$ matrix constructed by drawing $n=1000$ Gaussian random variables $\lambda_1, \dotsm, \lambda_n \sim \mathcal{N}(0,1)$ and a random orthogonal matrix $U \in \R^{n \times n}$, and outputting $U\Lambda U^\top$ where $\Lambda$ is a $n \times n$ diagonal matrix with entries $\frac{\lambda_1}{\max_i \lambda_i}, \dots, \frac{\lambda_n}{\max_i \lambda_i}$.
	\item \texttt{uniform} is a $1000 \times 1000$ matrix constructed identically to \texttt{gaussian} except with $\lambda_1, \dots, \lambda_n$ drawn independently and uniformly from the interval $[-1, 1]$.
	\item \texttt{resnet20} is a Hessian for the ResNet20 network \cite{he2016deep} trained on the Cifar-10 dataset. The matrix is $3000 \times 3000$ and its eigenvalues have been normalized to lie between $[-1, 1]$ for our experiments. 
\end{itemize}
For reference, the histogram of the eigenvalues for each matrix are shown in Figure \ref{fig:eigs-hist_plots} by breaking the range of the eigenvalues into 50 equally spaced intervals for each matrix.

In the first set of experiments, we compute the normalized Chebyshev moments ${\tau}_1, \ldots, {\tau}_N$ of each of the six aforementioned matrices using Hutchinson's moment estimator as in Algorithm \ref{alg:hutch_moments}, and, compute a spectral density estimate by passing these moments into Algorithm \ref{alg:MomentMatching_full} for approximate Chebyshev moment matching method (MM)\footnote{We solve the optimization problem from Line \ref{line:regression_mm} by formulating it as a linear program and using an off-the-shelf solver from \texttt{scipy}.} and into Algorithm \ref{alg:SDEpolynomial_full} for the Jackson damped kernel polynomial method (KPM). For KPM we compute the density with $N = 4, 6, 8, 10, \dots, 52$ and for MM we compute it with $N = 4, 5, 6, 7, \dots, 52$. We also compute the density estimate resulting from the stochastic Lanczos quadrature (SLQ) method of \cite{ChenTrogdonUbaru:2021} with $N = 4, 5, 6, 7, \dots, 52$ Lanczos iterations. We use $\ell = 5$ starting vectors (i.e., random vectors in Hutchinson's method, or random restarts of the SLQ method) for each method, except for the large \texttt{resnet20} and \texttt{hypercube} matrices, for which $\ell=1$ random vector is used. Each experiment is repeated 10 times and the Wasserstein-error between the true density and the density estimate are shown in Figure \ref{fig:error_plots}. The results show that MM is more than $10$x more accurate than KPM in almost all cases. The error of MM and SLQ are more comparable, except for \texttt{hypercube}, on which the errors are comparable for larger values of $N$. Both methods show an unusual convergence curve for this matrix, which we believe is related to the sparsify of its spectrum (a small number of distinct eigenvalues).

\begin{figure}[!t]
\centering
\begin{subfigure}
  \centering
  \includegraphics[width=0.325\linewidth]{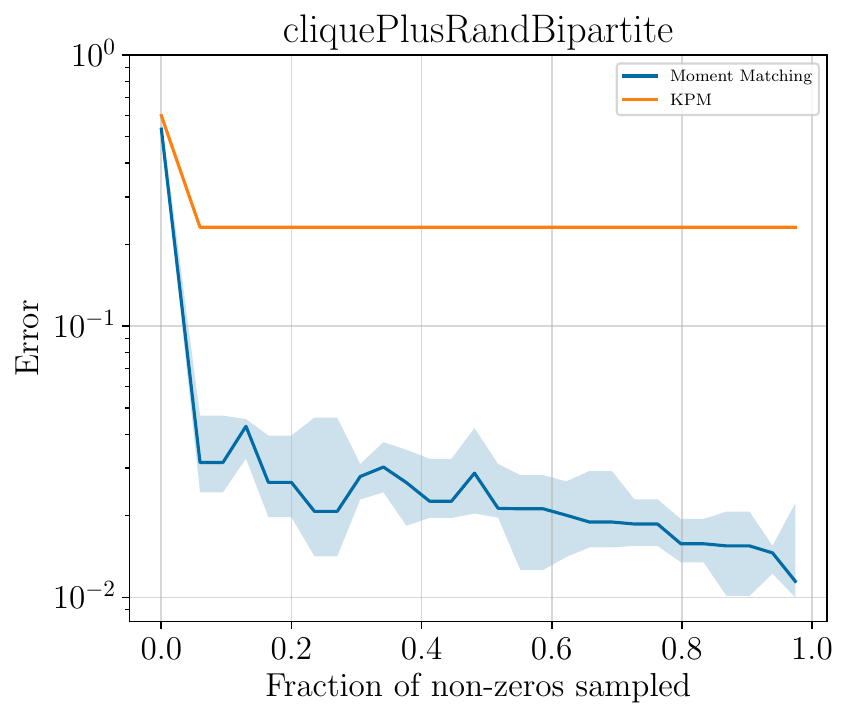}
  % \caption{Plot of relative error for 1000 queries, sorted by the magnitude of error for k-means++ and ScaNN respectively.}
  \label{fig:colSampling-cliquePlusRandBipartite}
\end{subfigure}%
\begin{subfigure}
  \centering
  \includegraphics[width=0.312\linewidth]{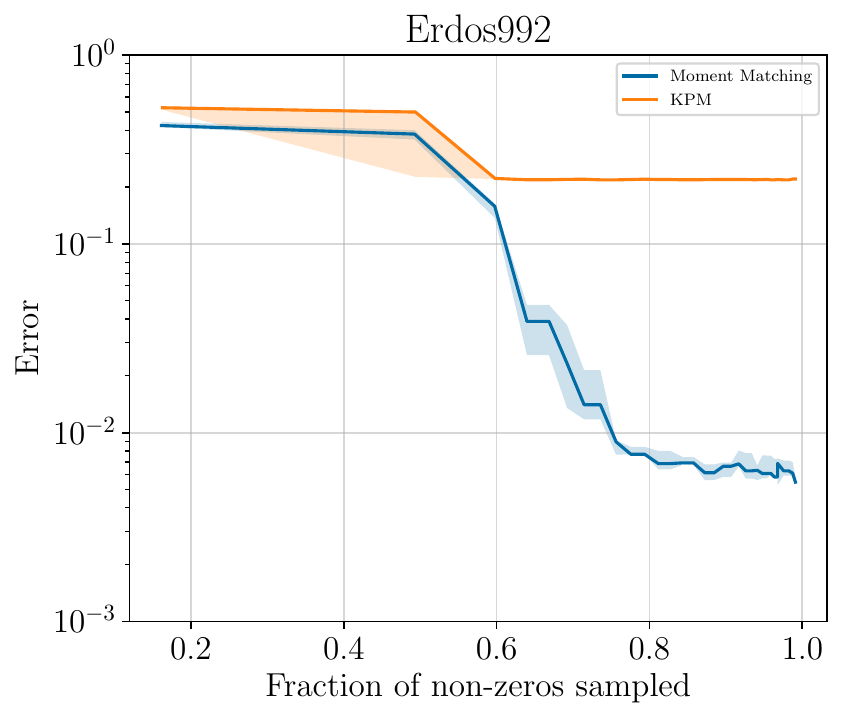}
  % \caption{Plot of relative error for 1000 queries, sorted by the magnitude of error for k-means++ and ScaNN respectively.}
  \label{fig:colSampling-Erdos992}
\end{subfigure}%
\begin{subfigure}
  \centering
  \includegraphics[width=0.312\linewidth]{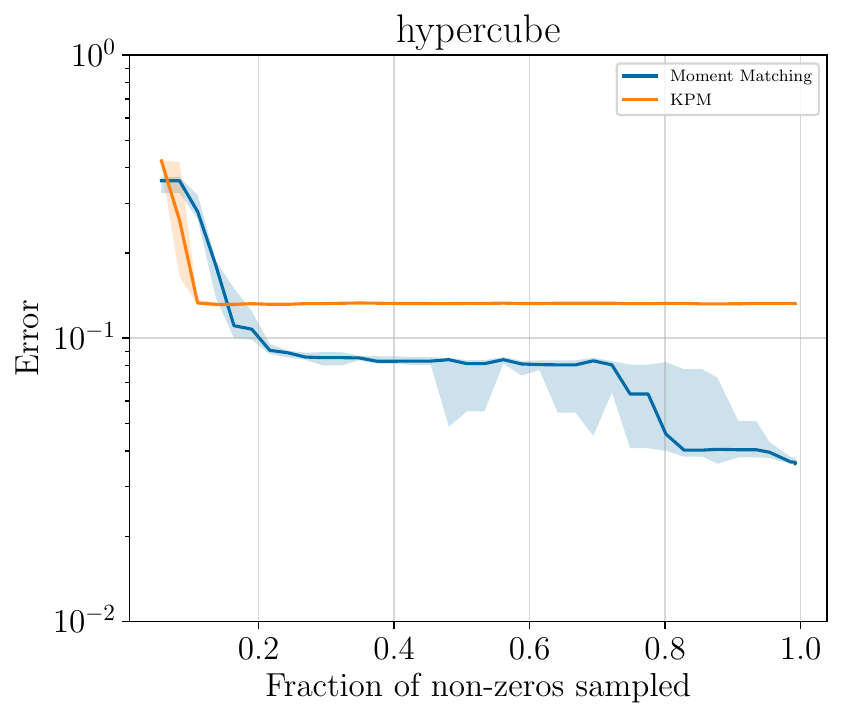}
  \label{fig:colSampling-hypercube}
\end{subfigure}%
\caption{Wasserstein error of density estimate returned by MM and KPM on the \texttt{hypercube}, \texttt{cliquePlusRandBipartite} and \texttt{Erdos992} graphs using approximate matrix-vector multiplications (Algorithm \ref{alg:amvGraphs}) to estimate the Chebyshev moments. For both methods, $N=32$ moments are computed using $5$ random starting vectors for \texttt{cliquePlusRandBipartite} and \texttt{Erdos992}  and $1$ for \texttt{hypercube}. The x-axis corresponds to the average fraction of non-zeros sampled from the matrix and the y-axis is the Wasserstein error from the resulting density estimate. Each experiment is repeated for $10$ trials: the solid line correponds to the median error of the 10 trials and the shaded region corresponds to the first and third quartiles.}
\label{fig:approxMM_plots}
\vspace{0.5em}
\hrule 
\vspace{-0em}
\end{figure}

In our second set of experiments, we test the performance of our randomized sublinear time algorithm (Algorithm \ref{alg:amvGraphs}) for approximate matrix-vector multiplies with normalized graph adjacency matrices. This method is used to estimate Chebyshev moments in Algorithm \ref{alg:MomentMatching_full} (MM) and in Algorithm \ref{alg:SDEpolynomial_full} (KPM).  We compute the normalized Chebyshev moments $\tau_1, \dots, \tau_N$ for $N = 12$ using various values of the oversampling parameter $t$ in the approximate matrix-vector multiplication method. We then compute, for each value of $t$, the average number of non-zero elements of $A$ accessed by the method for each matrix-vector product, which reflects the runtime improvement over a full matrix-vector product. 
%We then input the resulting moment estimates into Algorithm \ref{alg:MomentMatching_full} and Algorithm \ref{alg:SDEpolynomial_full}, which output a density estimate. 
Figure \ref{fig:approxMM_plots} plots the Wasserstein error of the density estimate (y-axis) and the average fraction of non-zeros used in each matrix-vector multiplication (x-axis) to estimate the Chebyshev moments used by MM and KPM respectively. 

The results show that the KPM method can achieve error nearly identical to that obtained when using exact matrix-vector multiplications, while only using a small fraction of non-zero entries for each approximate matrix-vector multiplication. Specifically, on the dense \texttt{cliquePlusRandBipartite} graph and even the relatively sparse  \texttt{hypercube} graph, KPM uses less than $15\%$ of the non-zero entries on average to achieve nearly the same error as when using exact multiplies. On \texttt{cliquePlusRandBipartite}, the MM method achieves error close to that of the exact method while using $\sim20\%$ of the non-zero entries on average. On the sparse \texttt{Erdos992} and  \texttt{hypercube} graphs, the MM method requires $\sim80\%$ of the non-zero entries on average to achieve error comparable to exact matrix-vector multiplications. However, it still obtains a good approximation (consistently better than the KPM method) when coarse matrix-vector multiplications are used (i.e., fewer non-zeros are sampled).

\section{Acknowledgements}
We thank Cameron Musco, Raphael Meyer, and Tyler Chen for helpful discussions and suggestions. This research was supported in part by NSF CAREER grants 2045590 and 1652257, ONR Award N00014-18-1-2364, and the Lifelong Learning Machines program from DARPA/MTO.

\bibliography{Refs}

\appendix

\section{The Kernel Polynomial Method}\label{sec:KPM}
In this section we show how to obtain a spectral density estimate based on a version of the kernel polynomial method that also approximates Chebyshev polynomial moments: $\tr(T_0(A)),\dots, \tr(T_N (A))$. We again rely on Jackson's classic work on universal polynomial approximation bounds for Lipschitz functions: we take advantage of the fact that Jackson’s construction of such polynomials is both linear and preserves positivity \cite{Jackson:1912}. 
%We then prove that it is still possible to extract a spectral density estimation even when the Chebyshev moments are computed to relatively low precision – in contrast to methods based on the standard moments, this allows use to obtain a polynomial dependence on $\epsilon$.

\subsection{Idealized Kernel Polynomial Method}\label{sec:JacksonKernelWasserstein}
As an alternative to the moment matching method presented in Section \ref{sec:approxCheby_mm}, a natural approach to using computed Chebyshev moments is to construct a truncated Chebyshev series approximation to $s$ (see Definition \ref{def:cheby_series}). To do so, note that the scaled moments $\frac{1}{n}\tr(\bar{T}_0(A)), \ldots, \frac{1}{n}\tr(\bar{T}_{N}(A))$ are exactly equal to the first $N$ Chebyshev series coefficients of $s/w$, where $w(x) = \frac{1}{\sqrt{1 - x^2}}$ is as defined in Section \ref{sec:Preliminaries}. Specifically, the eigenvalues of $\bar{T}_k(A)$ are equal to $\bar{T}_k(\lambda_1),\ldots,\bar{T}_k(\lambda_n)$, where $\lambda_1, \ldots, \lambda_n$ are the eigenvalues of $A$.  Since the trace of a diagonalizable matrix is the sum of its eigenvalues, we have $\frac{1}{n}\tr(\bar{T}_k(A)) = \frac{1}{n}\sum_{i=1}^n \bar{T}_k(\lambda_i) = \langle s,\bar{T}_k\rangle = \langle s/w,w\cdot\bar{T}_k\rangle$.

After using the scaled Chebyshev moments to construct a truncated Chebyshev series for $s/w$, i.e. a degree $N$ polynomial approximation, we can then multiply the final result by $w$ to obtain an approximation to $s$. Unfortunately, there are two issues with this approach: 1) it is difficult to analyze the quality of the Chebyshev series approximation, since $s$ is not a smooth function, and 2) this approximation will not in general be a non-negative function, which is a challenge because our goal is to find \emph{probability density} that well approximates $s$.

\begin{wrapfigure}{r}{0.41\linewidth}
	\vspace{-1em}
	\centering
	\includegraphics[width=.41\textwidth]{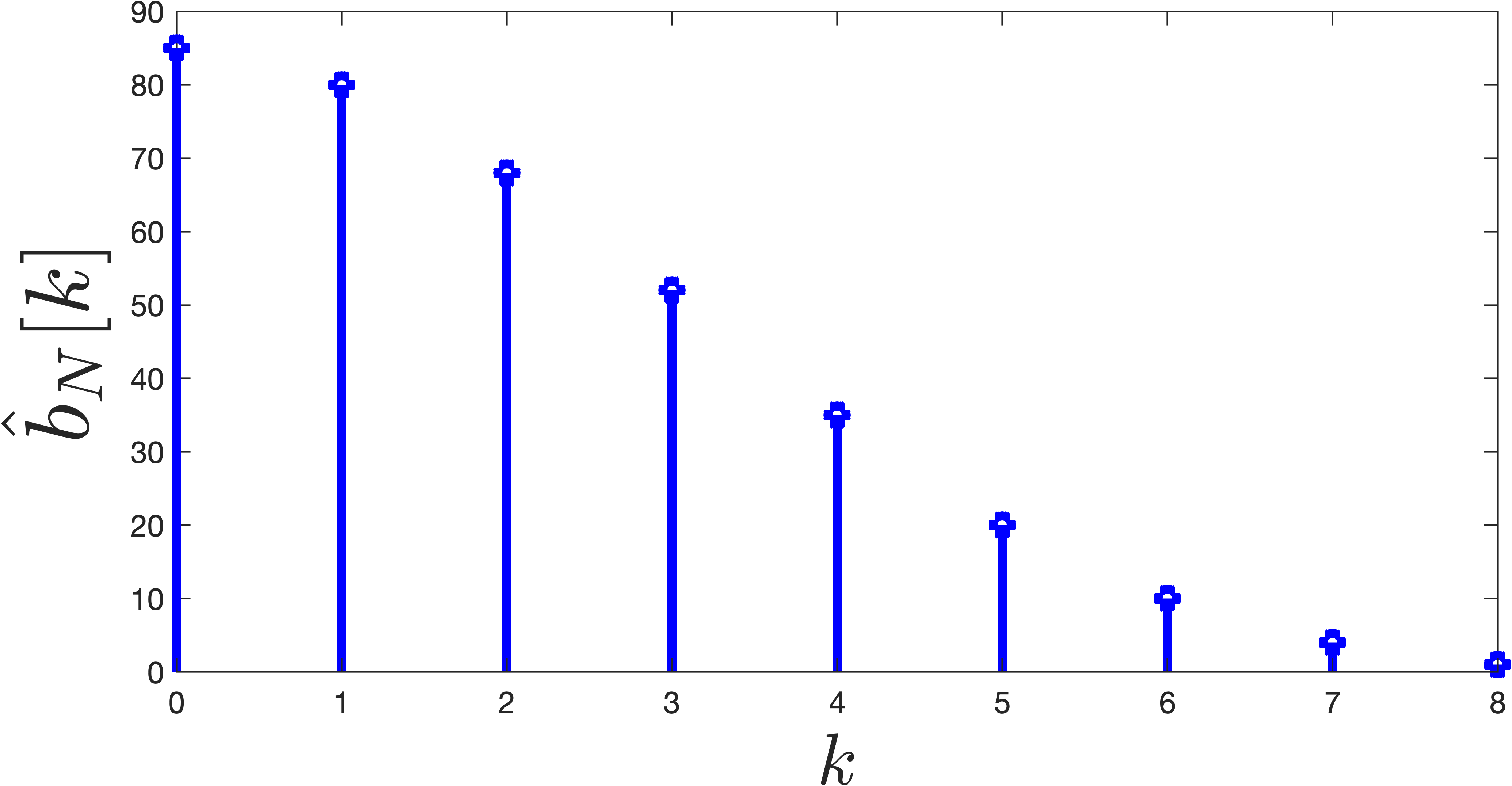}
	
	\vspace{-1em}
	\caption{Jackson coefficients for $N = 8$.}
	\label{fig:bhat_intro}
	\vspace{-.5em}
\end{wrapfigure}
A common approach for dealing with the second issue is to instead use a \emph{damped Chebyshev expansion} \cite{weisse2006kernel}, where the Chebyshev coefficients are slightly reweighted to ensure that the resulting polynomial is always non-negative. Such non-negativity preserving damping schemes follow from the connection between Chebyshev and Fourier series: we refer the reader to Appendix \ref{appendix:jacksonKernel} for details. In short, by the convolution theorem, Fourier series truncation corresponds to convolution with a function whose Fourier support is bounded. If this function is also non-negative, convolution preserves non-negativity of the function being approximated, leading to truncated series that is guaranteed to be positive.
One such damping schemes was introduced in classic work of Jackson \cite{Jackson:1912}. For any positive integer $z$, let $N = 4z$. Then, for $k = 0, \ldots, N$, define the coefficient
\begin{align}
	\label{bump_coeffs_intro}
	\hat{b}_{N}[k] &= \sum_{j=-\frac{N}{2}-1}^{\frac{N}{2}+1-k} \left(\frac{N}{2}+1-|j|\right)\cdot\left(\frac{N}{2}+1-|j+k|\right). 
\end{align}
While \eqref{bump_coeffs_intro} may look opaque, $\hat{b}_{N}[0], \ldots, \hat{b}_{N}[N]$ are actually equal to the result of a simple discrete convolution operation. Let $g\in \Fc(\Z, \R)$ have $g[j] = 1$ for $j = -z, \ldots, z$, and $g[j]=0$ otherwise. Then let $\hat{b}_{N} = (g*g)*(g*g)$ and $\hat{b}_{N}[0], \ldots, \hat{b}_{N}[N]$ be the values corresponding to non-negative indices.\footnote{This formulation allows the coefficients to be easily computed in most high-level programming languages. E.g., in \texttt{MATLAB} we can compute \texttt{g = ones(2*z+1,1); c = conv(conv(g,g),conv(g,g)); b = c(N+1:2*N+1);}.} See Fig. \ref{fig:bhat_intro} for an illustration of these coefficients. They are all positive and $\hat{b}_{N}[0] > \hat{b}_{N}[1] > \ldots > \hat{b}_{N}[N]$.
Jackson suggests approximating a function using the following truncation based on these coefficients:
\begin{definition}[Jackson damped Chebyshev series]
	Let $f\in \FcOne$ have Chebyshev series $\sum_{k = 0}^\infty \iprod{f}{w\cdot \bar{T}_k} \cdot \bar{T}_k$. The Jackson approximation to $f$ is a degree $N$ polynomial $\bar{f}_N$ obtained via the following truncation with modified coefficients:
	\begin{align}\label{eqn:jacksno_trun}
		\bar{f}_N(x) \eqdef \sum_{k = 0}^N \frac{\hat{b}_{N}[k]}{\hat{b}_{N}[0]} \iprod{f}{w\cdot\bar{T}_k}\bar{T}_k(x).
	\end{align}
\end{definition}

Note that $\hat{b}_{N}[0]/\hat{b}_{N}[0] = 1$, and all other terms are strictly less than one. It is not hard to show this damped series preserves positivity. We prove the following fact as Lemma \ref{lem:non-negativity} in the appendix:
\begin{fact}\label{fact:positivity}
	If $f: [-1,1]\rightarrow \R^{\geq 0}$ is a non-negative function, then the polynomial $\bar{f}_N(x)$ defined in \eqref{eqn:jacksno_trun} is non-negative for all $x \in [-1, 1]$.
%	Moreover, if $s$ is a probability density, then $\bar{s}_N$ is a probability density: $\int_{-1}^1\bar{s}_N(x)dx = 1$.
\end{fact}
Beyond preserving non-negativity, as claimed in Fact \ref{fact:Jacksons}, the Jackson damped Chebyshev approximation is more well-known for the fact that it provably provides a good \emph{uniform} polynomial approximation to any Lipschitz function.
For completeness, we give a proof of this fact as Theorem \ref{thm:jacksons_theorem_algebraic} in the appendix.
With Facts \ref{fact:positivity} and \ref{fact:Jacksons} in place, we are ready to introduced the basic kernel polynomial method for approximating the spectral density $s$ as Algorithm \ref{alg:SDEpolynomial_ideal}. This algorithm is identical to the ``Jackson Kernel'' KPM from \cite{weisse2006kernel}. Recall that, for now, we assume we have access to \emph{exact} Chebyshev moment of the spectral density $s$ for our matrix $A$. In Section \ref{sec:SDE} we prove that Algorithm \ref{alg:SDEpolynomial_ideal} is robust to using approximate moments.
%In reality, these moments will need to be approximated using a randomized algorithm based on Hutchinson's trace estimator -- we will show how to do so in Section \ref{sec:SDE}, and prove that Algorithm \ref{alg:SDEpolynomial_ideal} is robust to this approximation.

\begin{algorithm}[h]\caption{Idealized Jackson Damped Kernel Polynomial Method}\label{alg:SDEpolynomial_ideal}
	\begin{algorithmic}[1]
		\Require Symmetric $A\in \R^{n\times n}$ with spectral density $s : [-1, 1] \rightarrow \R^{\geq 0}$,  degree $N \in 4\N^+$.
		\Ensure Density function $q : [-1, 1] \rightarrow \R^{\geq 0}$.
		\State For $k=0, \ldots, N$ compute $\tau_k = \frac{1}{n}\tr(\bar{T}_k(A)) = \iprod{s}{\bar{T}_k}$. 
		\State For $k=0, \ldots, N$ compute $\hat{b}_{N}[k]$ as is \eqref{bump_coeffs_intro}.
		\State Return $q = w\cdot\sum_{k=0}^N \frac{\hat{b}_{N}[k]}{\hat{b}_{N}[0]} \cdot \tau_k \cdot \bar{T}_k$.
	\end{algorithmic}
\end{algorithm}

\begin{lemma}\label{lem:idealized_wass_bound}
	If $N \geq \frac{18}{\epsilon}$, then the function $q : [-1, 1] \rightarrow \R^{\geq 0}$ returned by Algorithm \ref{alg:SDEpolynomial_ideal} is a probability density and satisfies:\vspace{-1em}
	\begin{align*}
		W_1(s,q) \leq \epsilon.
	\end{align*}
\end{lemma} 
\begin{proof}
	We first prove that $q$ is a probability density. To see that it is positive, note that $h = \sum_{k=0}^N \frac{\hat{b}_{N}[k]}{\hat{b}_{N}[0]} \cdot \tau_k \cdot \bar{T}_k$ is a Jackson approximation to the positive function $s/w$, so is must be non-negative by Fact \ref{fact:positivity}. Since $w$ is also non-negative, we conclude that $q = w\cdot h$ is as well. Then we consider $q$'s integral. We need to show that $\int_{-1}^1q(x)dx = 1 = \int_{-1}^1s(x)dx$. Since $\bar{T}_0$ is a scaling of the constant function, it suffices to show that $\langle \bar{T}_0,q\rangle = \langle \bar{T}_0,s\rangle$. We have:
		\begin{align*}
		\langle \bar{T}_0,q\rangle = \tau_0\cdot \langle \bar{T}_0,w\cdot \bar{T}_0\rangle = \langle \bar{T}_0,s\rangle\cdot 1.
	\end{align*}
The first step follows directly from the orthogonality of the Chebyshev polynomials under the weight function $w$, which implies that $\langle \bar{T}_0,w\cdot \bar{T}_k\rangle = 0$ for all $k > 0$. We also use that $\langle \bar{T}_0, w \cdot \bar{T}_0\rangle = 1$.

Next, we prove the approximation guarantee. Referring to the formulation of Wasserstein-1 distance from equation \eqref{eq:wass_def}, we have that $W_1(s,q) = \sup \langle s-q,f\rangle$ where $f$ is a $1$-Lipschitz function. So, we want to show that any $1$-Lipschitz $f$ has small inner product with the difference between $s$ and its degree-$N$ Jackson approximation, $q$. To do so, we show that this inner product is actually \emph{exactly equal to} the inner product between $s$ and a degree-$N$ Jackson approximation to $f$. Since $f$ is 1-Lipschitz, this approximation is guaranteed to be have small error. This key equivalency follows because, like a standard Chebyshev series approximation, the Jackson approximation can be viewed as the output of a symmetric linear operator applied to $s$. 

Formally, we introduce notation for several linear operators needed to analyze \eqref{eqn:jacksno_trun}. Let $\bar{\Tc} : \FcOne \rightarrow \FcN$ be the operator mapping a function $f \in \FcOne$ to its inner-product with the normalized Chebyshev polynomials. Define the transpose operator $\bar{\Tc}^{*} : \FcN \rightarrow \FcOne$ to satisfy $\iprod{\bar{\Tc}f}{g} = \iprod{f}{\bar{\Tc}^*g}$ for any $g \in \FcN$. Concretely, for $i \in \N$ and $x \in [-1, 1]$,
\begin{align}
	[\bar{\Tc} f][i] \eqdef \int_{-1}^1 \bar{T}_i(x)f(x)dx &&\text{and} &&&[\bar{\Tc}^{*} g](x) \eqdef \sum_{i = 0}^\infty \bar{T}_i(x)g[i].
\end{align} 
We also define operators $\Wc:\FcOne \rightarrow \FcOne$ and $\Ic : \FcN \rightarrow \FcN$ as follows:
\begin{align*}
	[\Wc f](x) \eqdef w(x)f(x) = \frac{1}{\sqrt{1 - x^2}}f(x)  &&\text{and} &&& [\Ic g][i] \eqdef g[i].
\end{align*} 
Note that $\Ic$ is an identity operator. For any $N\in 4\N$, we define $\Bc_N : \FcN \rightarrow \FcN$  as:
\begin{align*}
	[\Bc_Ng](i) \eqdef \begin{cases}
		\frac{\hat{\beta}_N[i]}{\hat{\beta}_N[0]}g(i) & \text{for } 0\leq i\leq N \\
		0 & i > N.
	\end{cases}
\end{align*} 
The operators $\Wc$, $\Ic$, and $\Bc_N$ are all commutative with respect to the inner-products in their respective spaces. Specifically, for $f_1, f_2 \in \FcOne$ and $g_1, g_2 \in \FcN$, $\iprod{f_1}{\Wc f_2} = \iprod{\Wc f_1}{f_2}$, $\iprod{g_1}{\Ic g_2} = \iprod{\Ic g_1}{g_2}$, and $\iprod{g_1}{\Bc_N g_2} = \iprod{\Bc_N g_1}{g_2}$. Also note that by orthogonality of the Chebyshev polynomials under $w$, $\bar{\Tc}^* \bar{\Tc}\Wc$ is the identity operator on $\FcOne$ and so is $\Wc\bar{\Tc}^* \bar{\Tc}$.

With these operators defined, the remainder of the proof is short. We have via direct calculation:
\begin{align*}
	\iprod{f}{s-q} &= \iprod{f}{s - \Wc \bar{\Tc}^* \Bc_N \bar{\Tc} s} \\
	&= \iprod{f}{\Wc \bar{\Tc}^* (\Ic - \Bc_N) \bar{\Tc} s} = \iprod{\bar{\Tc}^* (\Ic - \Bc_N) \bar{\Tc} \Wc f}{s} = \iprod{f - \bar{\Tc}^* \Bc_N \bar{\Tc} \Wc f}{s}.
\end{align*}
Note that $\bar{\Tc}^* \Bc_N \bar{\Tc} \Wc f$ is exactly the degree-$N$ Jackson approximation to $f$. So by Fact \ref{fact:jacksons}, if $f$ is a $1$-Lipschitz function, $\|f - \bar{\Tc}^* \Bc_N \bar{\Tc} \Wc f \|_\infty \leq 18/N$. Since $s$ is a non-negative function that integrates to $1$, it follows that $\iprod{f}{s-q} = \iprod{f - \bar{\Tc}^* \Bc_N \bar{\Tc} \Wc f}{s} \leq 18/N$. Since $W(s,q) = \sup_{\text{1-Lipschitz }f} \iprod{f}{s-q}$, we conclude that $W(s,q) \leq \epsilon$ as long as as long as $N \geq 18/\epsilon$.
\end{proof}
\textbf{Remark.} Given access to the Chebyshev polynomial moments, $\tr(\bar{T}_0(A)), \ldots, \tr(\bar{T}_N(A))$, Algorithm \ref{alg:SDEpolynomial_ideal} can be implemented in $O(1/\epsilon)$ additional time. The function it returns is an $O(1/\epsilon$) degree polynomial times the closed form function $w$. The polynomial can be represented as a sum of Chebyshev polynomials, or converted to standard monomial form in $O(1/\epsilon^2)$ time. The function is easily plotted or integrated over a range -- see discussion around Fact \ref{eqn:chebyshev_integral} for more details.
 
\subsection{Full Kernel Polynomial Method}
\label{sec:SDE}
Since it is not possible to efficiently compute the exact Chebyshev polynomial moments, we need to show that the kernel polynomial method can work with  approximations to these moments, computed e.g. using a stochastic trace estimator as described in Section \ref{sec:SDE_approx}. Here, we first prove a general result on the accuracy of approximation needed to ensure we obtain a good spectral density estimation. Specifically, we analyze the following ``robust'' version of Algorithm \ref{alg:SDEpolynomial_ideal}.

\begin{algorithm}[h]\caption{Jackson Damped Kernel Polynomial Method}\label{alg:SDEpolynomial_full}
	\begin{algorithmic}[1]
		\Require Symmetric $A\in \R^{n\times n}$ with spectral density $s : [-1, 1] \rightarrow \R^{\geq 0}$, degree parameter $N \in 4\N^+$, algorithm $\mathcal{M}(A)$ that computes moment approximations $\tilde{\tau}_1, \ldots, \tilde{\tau}_N$ with the guarantee that $|\tilde{\tau}_k - \frac{1}{n}\tr(\bar{T}_k(A))| \leq 1/N^2$ for all $k$. 
		\Ensure Density function $q : [-1, 1] \rightarrow \R^{\geq 0}$.
		\State For $k=1, \ldots, N$ use $\mathcal{M}$ to compute $\tilde{\tau}_1, \ldots, \tilde{\tau}_N$ as above. Set $\tilde{\tau}_0 = 1/\sqrt{\pi}$. 
		\State For $k=0, \ldots, N$ compute $\hat{b}_{N}[k]$ as is \eqref{bump_coeffs_intro}.
		\State Compute polynomial $\tilde{s}_N = w\cdot\sum_{k=0}^N \frac{\hat{b}_{N}[k]}{\hat{b}_{N}[0]} \cdot \tilde{\tau}_k \cdot \bar{T}_k$.
		\State {Return} the probability density
			$q = \left(\tilde{s}_N + \frac{w\sqrt{2}}{N\sqrt{\pi}}\right)/\left(1 + \frac{\sqrt{2\pi}}{N}\right)$.
	\end{algorithmic}
\end{algorithm}
The final transformation of $\tilde{s}_N$ in Line 4 of Algorithm \ref{alg:SDEpolynomial_full} ensures that we return a proper density, since error incurred by approximating $\frac{1}{n}\tr(\bar{T}_k(A)) = \iprod{s}{\bar{T}_k}$ could leave the function with negative values. So, we shift by a small positive function, and rescale to maintain unit integral. Our main result on the error of Algorithm \ref{alg:SDEpolynomial_full}, which parallels Lemma \ref{lem:idealized_wass_bound} for Algorithm \ref{alg:SDEpolynomial_ideal}, is as follows:
\begin{restatable}{lemma}{degreeAndErrorBoundsforWass}\label{lem:degreeAndErrorBoundsforWass}
If $N \geq \frac{18}{\epsilon}$, then the function $q : [-1, 1] \rightarrow \R^{\geq 0}$ returned by Algorithm \ref{alg:SDEpolynomial_full} is a probability density and satisfies:\vspace{-1em}
\begin{align*}
	W_1(s,q) \leq 2\epsilon.
\end{align*}
\end{restatable} 
%\textbf{Remark.} When $N = O(1/\epsilon)$, as required by Lemma \ref{lem:degreeAndErrorBoundsforWass}, Algorithm \ref{alg:SDEpolynomial_full} only requires a $1/N^2 = O(\epsilon^2)$ approximation to the Chebyshev polynomial moments, which is rather coarse in that it depends polynomially, rather than exponentially, on $\epsilon$. As we will see in the next section, such an approximation is easily obtained with $\poly(1/\epsilon)$ matrix-vector multiplications with $A$, and for some matrices, can even be obtained with a sublinear time algorithm.

\begin{proof}
We first prove that $q$ is a probability distribution. Let $s_N$ denote the ideal distribution returned by Algorithm \ref{alg:SDEpolynomial_ideal} if exact Chebyshev moments were used. I.e., 
\begin{align*}
	s_N = w\cdot\sum_{k=0}^N \frac{\hat{b}_{N}[k]}{\hat{b}_{N}[0]} \cdot {\tau}_k \cdot \bar{T}_k
\end{align*}
where $\tau_k = \frac{1}{n}\tr(\bar{T}_k(A)) = \iprod{s}{\bar{T}_k}$. Note that for any density $s$, $\tau_0 = \iprod{s}{\bar{T}_0} = 1/\sqrt{\pi}$.
Let $\Delta_k = \tilde{\tau}_k - \tau_k$. We have $\tilde{s}_N(x) = {s}_N + \sum_{k = 1}^N \Delta_k \frac{\hat{b}_{N}[k]}{\hat{b}_{N}[0]} w(x)\bar{T}_k(x)$. Define functions $\eta = s_N/w$ and $\tilde{\eta} = \tilde{s}/w$.
% be the functions $\eta = \sum_{k = 0}^N \frac{\hat{b}_{N}[k]}{\hat{b}_{N}[0]} \tau_k \bar{T}_k$ and $\tilde{\eta} = \sum_{k = 0}^N \frac{\hat{b}_{N}[k]}{\hat{b}_{N}[0]} \tilde{\tau}_k \bar{T}_k$ and note that ${s}_N(x) = w(x)\eta(x)$ and $\tilde{s}_N(x) = w(x)\tilde{\eta}(x)$. 
It follows that for any $x \in [-1, 1]$, 
\begin{align}\label{eqn:delta_ub}
\abs{\tilde{\eta}(x) - \eta(x)}  = \abs{\sum_{k = 1}^N \frac{\hat{b}_{N}[k]}{\hat{b}_{N}[0]} \Delta_k \bar{T}_k(x)}
\leq \frac{\sqrt{2}}{N\sqrt{\pi}}.
\end{align}
The last inequality uses that $0 \leq {\hat{b}_{N}[k]}/{\hat{b}_{N}[0]} \leq 1$ and for $x \in [-1, 1]$, $\bar{T}_k(x) \leq \sqrt{2/\pi}$ for $k\geq 1$. Since $\eta$ is a non-negative function, from \eqref{eqn:delta_ub} we can conclude that the function $\tilde{\eta} + \frac{\sqrt{2}}{N\sqrt{\pi}}$ is non-negitive, and thus $w\cdot (\tilde{\eta} + \frac{\sqrt{2}}{N\sqrt{\pi}})= \tilde{s}_N + w\frac{\sqrt{2}}{N\sqrt{\pi}}$ is also non-negative. The density of this function is $\int_{-1}^1\tilde{s}_N(x) dx + \frac{\sqrt{2}}{N\sqrt{\pi}}\int_{-1}^1 w(x) dx = 1 + \frac{\sqrt{2\pi}}{N}$, so dividing by $1 + \frac{\sqrt{2\pi}}{N}$ gives a probability density. 

Next we prove the approximation guarantee. By Lemma \ref{lem:idealized_wass_bound} we know that $W_1(s,s_N)\leq \epsilon$, so if we can show that $W_1(s_N,q) \leq \epsilon$, then by triangle inequality we will have shown that $W_1(s,q) \leq W_1(s,s_N) + W_1(s_N,q) \leq 2\epsilon$. 

To bound $W_1(s_N,q)$, we need to show that $\iprod{f}{s_N - q} \leq \epsilon$ for any 1-Lipschitz function $f\in \FcOne$. Without loss of generality, we can assume that 
$\int_{-1}^1 f(x)dx = 0$, as the 1-Lipschitz function $f' = f - \int_{-1}^1 f(x)dx$ satisfies $\iprod{f}{s_N - q} = \iprod{f'}{s_N - q}$ (since $s_N$ and $q$ are both probability densities). If $\int_{-1}^1 f(x)dx = 0$, $f(x)$ must be zero for some $x\in [-1,1]$, and since it is also $1$-Lipschitz we can in turn bound $\|f\|_\infty \leq 1$.\footnote{Let z maximize $f(x)$. Since $f$ is 1-Lipschitz we have $f(z) \leq |x-z| - f(x)$ for all $x$. Integrating both sides from $-1$ to $1$, we have $2f(z) \leq (z^2+1) - 0 \leq 2$. So, $f(z) \leq 1$.} We can then bound the inner product:  
\begin{align*}
\iprod{f}{s_N - q} \leq \|{f(\bar{s}_N - q)}\|_1 &\leq \|f\|_\infty \|\bar{s}_N(x) - q(x)\|_1 
\leq \left\|w \cdot \brackets{\eta - \frac{\tilde{\eta} +\frac{\sqrt{2}}{N\sqrt{\pi}}}{1 + \frac{\sqrt{2\pi}}{N} }} \right\|_1 \\ 
&\leq  \underbrace{\left\|w \cdot \brackets{\eta - \tilde{\eta} -\frac{\sqrt{2}}{N\sqrt{\pi}}} \right\|_1}_{z_1} + \  \underbrace{\left\|\frac{\sqrt{2\pi}}{N} \cdot w\cdot \brackets{\tilde{\eta} + \frac{\sqrt{2}}{N\sqrt{\pi}}} \right\|_1}_{z_2}
\end{align*}
The last inequality uses the fact that $1 - \frac{1}{1+\gamma} \leq \gamma$ for $0 \leq \gamma \leq 1$, which we apply with $\gamma = \frac{\sqrt{2\pi}}{N}$. 
Using the fact that $\|w\|_1 = \int_{-1}^1 \frac{1}{\sqrt{1 - x^2}}dx = \pi$ and the bound on $\|\eta - \tilde{\eta}\|_\infty$ from \eqref{eqn:delta_ub}, we have 
\begin{align*}
z_1 &\leq \|w\|_1 \cdot \left\|\eta - \tilde{\eta} -\frac{\sqrt{2}}{N\sqrt{\pi}} \right\|_\infty \leq \frac{2\pi\sqrt{2}}{N\sqrt{\pi}}.
\end{align*}
Examining $z_2$, recall that we showed earlier that $w(\tilde{\eta} + \frac{\sqrt{2}}{N\sqrt{\pi}}) = \tilde{s}_N + \frac{\sqrt{2}}{N\sqrt{\pi}}w$ has $\ell_1$ norm $1 + \frac{\sqrt{2\pi}}{N}$. So we have $z_2 \leq \frac{\sqrt{2\pi}}{N} ({1 + \frac{\sqrt{2}}{N\sqrt{\pi}}}) \leq \frac{2\sqrt{2\pi}}{N}$ for all $N \geq 1$. 

Compiling the bounds on $z_1$ and $z_2$, we have that for all $1$-Lipschitz $f$, $\iprod{f}{s_N - q} \leq \frac{4\sqrt{2\pi}}{N} \leq \frac{11}{N}$, and thus $W_1(\bar{s}_N, q)\leq \frac{11}{N}$. For $N\geq \frac{18}{\epsilon}$ we conclude that $W_1(\bar{s}_N, q) \leq \epsilon$. Applying triangle quality as discussed above completes the proof. 
\end{proof}

Lemma \ref{lem:degreeAndErrorBoundsforWass} is exactly analogous to Lemma \ref{lemma:mm_guarantee}. We can take advantage of the result by using the Hutchinson's based method from Section \ref{sec:SDE_approx} or the sublinear time method from Section \ref{sec:graphSDE} to obtain the approximations for the Chebyshev moments required by Algorithm \ref{alg:SDEpolynomial_full}. The end result is that we can obtain the same bounds as Theorem \ref{thm:exactMatMultSDE} and Theorem \ref{thm:approxMatMultSDE} with $\ell = \max(1, \ \frac{C'}{n}\epsilon^{-4}\log^2(\frac{1}{\epsilon\delta}))$ and $\epsMV = C''\epsilon^{-4}$, respectively. The slightly worse $\epsilon$ dependence follows from the fact that Algorithm \ref{alg:SDEpolynomial_full} has a more stringent requirement on the approximate Chebyshev moments used than Algorithm \ref{alg:MomentMatching_full}.

\section{Approximate Eigenvalues from Spectral Density Estimate}\label{sec:appxEigs}
Algorithm \ref{alg:SDEpolynomial_ideal} and Algorithm \ref{alg:SDEpolynomial_full} in the previous sections output a closed form representation of a distribution $q$ which is close in Wasserstein-1 distance to $s$. In particular, the distribution output is continuous. Alternatively, we describe a simple greedy algorithm (Algorithm \ref{alg:appxEigs}) that recovers a list of $n$ eigenvalues $\tilde{\Lambda} = [\tilde{\lambda}_1, \dots, \tilde{\lambda}_n]$ such that $\|\Lambda - \tilde{\Lambda}\|_1 \leq 3n\epsilon$, which implies that the discrete distribution associated with $\tilde{\Lambda}$ is $3\epsilon$ close to $s$ in Wassersetin-1 distance. Formally:

% Algorithm \ref{alg:SDEpolynomial_ideal} and Algorithm \ref{alg:SDEpolynomial_full} in the previous sections output a closed form representation of a distribution $q$ which is close in Wasserstein-1 distance to $s$. In particular, the distribution output is continuous. Alternatively, we describe an algorithm from \cite{Cohen-SteinerKongSohler:2018} (Algorithm \ref{alg:appxEigs}) that recovers a list of $n$ eigenvalues $\tilde{\Lambda} = [\tilde{\lambda}_1, \dots, \tilde{\lambda}_n]$ such that $\|\Lambda - \tilde{\Lambda}\|_1 \leq 2n\epsilon$, which implies that the discrete distribution associated with $\tilde{\Lambda}$ is $2\epsilon$ close to $s$ in Wassersetin-1 distance. Formally:
\begin{restatable}{theorem}{approxEigs}\label{thm:approxEigs}
	Let $s$ be a spectral density and let $q$ be a density on $[-1,1]$ such $W_1(s, q) \leq \epsilon$ for $\epsilon \in (0, 1)$. As long as $q$ can be integrated over any subinterval of $[-1,1]$ (e.g., has a closed form antiderivative), there is an algorithm (Algorithm \ref{alg:appxEigs}) that computes $1/\epsilon$ such integrals and in $O\left(n + {1}/{\epsilon}\right)$ additional time outputs a list of $n$ values $\tilde{\Lambda} = [\tilde{\lambda}_1, \dots, \tilde{\lambda}_n]$ such that $\|\Lambda - \tilde{\Lambda}\|_1 \leq 3n\epsilon$. 
\end{restatable}

At a high-level, Algorithm \ref{alg:appxEigs} computes a grid with spacing $\epsilon$ for the interval $[-1, 1]$, ``snaps'' the mass of the continuous density onto the nearest point in the grid, and then readjusts the resulting point masses to a distribution where every point mass is divisible by $1/n$ (and can therefore be represented by a certain number of eigenvalues). It does so by iteratively shifting fractional masses to the next point in the grid so that the mass at the current point is divisible by $1/n$. 

The method requires computing the mass $\int_{a}^b q(x)dx$ where $-1 \leq a < b \leq 1$. For Algorithms \ref{alg:SDEpolynomial_ideal} and  \ref{alg:SDEpolynomial_full}, $q$ is written as $q = w \cdot p$ where $p$ is a degree $N$ polynomial written as a sum of the first $N+1$ Chebyshev polynomials. So to compute the integral $\int_{a}^b q(x)dx$, we just need to compute the integral $\int_{a}^b T_k(x)w(x)dx$ for any $k \in 0,\ldots, N$. We can do so using the following closed form expression (see Appendix \ref{appendix:chebyshev_integral} for a short derivation):
% and prove it in Appendix \ref{appendix:chebyshev_integral}.
\begin{fact}\label{eqn:chebyshev_integral}
	For $k \in \N^{> 0}$ and $-1 \leq a < b \leq 1$ we have that 
	\begin{align*}
%		\int_{a}^b \frac{T_k(x)}{\sqrt{1 - x^2}}dx = \frac{-\cos(ku)}{k} \bigg\rvert_{\sin^{-1} a}^{\sin^{-1} b}  
		\int_{a}^b \frac{T_k(x)}{\sqrt{1 - x^2}}dx = \frac{-\cos(k\sin^{-1} b)}{k} - \frac{-\cos(k\sin^{-1} a)}{k}
	\end{align*}
	For $k=0$, $T_k(x)=1$ for all $x$ and we have that $\int_{a}^b T_k(x)w(x)dx = \sin^{-1}(b) - \sin^{-1}(a)$.
\end{fact}
Using the above fact, when $q = w\cdot p$ for a degree $N$ polynomial $p$, we can compute $\int_{a}^b q(x)dx$ in $O(N)$ time. In our main results $N = O(1/\epsilon)$, so this cost is small.

% \begin{algorithm}\caption{Approximate Eigenvalues from Spectral Density}\label{alg:appxEigs}
% 	\begin{algorithmic}[1]
% 		\Require Spectral density $q : [-1, 1] \rightarrow \R^+$, integer $n$.
% 		\Ensure Vector $\tilde{\Lambda} = [\tilde{\lambda}_1, \dots, \tilde{\lambda}_n]$.
% 		\State compute $\vec{v} = (v_{-1 + \epsilon}, v_{-1 + 2\epsilon} \dots, v_0, v_{\epsilon}, \dots, v_1)$ such that $v_t = \int_{t - \epsilon}^{t} q(x)dx$
		
% 		\State define $f_v : [0, 1] \rightarrow [-1, 1]$ to be a non-decreasing step-function with the property that for $X \sim \text{Uniform}([0, 1])$, $f_v(X) = -1 + t\epsilon$ with probability $v_{-1 + t\epsilon}$ for $t \in [1, \dots, \frac{2}{\epsilon}]$.
% 		\State Set $$\tilde{\lambda}_i = \ExpBig{f_v(X) \ | \ X \in \left[\frac{i-1}{n}, \frac{i}{n} \right ]}$$ \label{line:lambdas}
% 	\end{algorithmic}
	
% \end{algorithm}

\begin{algorithm}\caption{Approximate Eigenvalues from Spectral Density}\label{alg:appxEigs}
	\begin{algorithmic}[1]
		\Require Spectral density $q : [-1, 1] \rightarrow \R^+$, integer $n$.
		\Ensure Vector $\tilde{\Lambda} = [\tilde{\lambda}_1, \dots, \tilde{\lambda}_n]$.
		\State compute $\vec{v} = (v_{-1 + \epsilon}, v_{-1 + 2\epsilon} \dots, v_0, v_{\epsilon}, \dots, v_1)$ such that $v_t = \int_{t - \epsilon}^{t} q(x)dx$
		% \For{$v_t$ in $(v_{-1 + \epsilon}, v_{-1 + 2\epsilon} \dots, v_0, v_{\epsilon}, \dots, v_1)$}
		\For{$t$ in $({-1 + \epsilon}, {-1 + 2\epsilon} \dots, 0, {\epsilon}, \dots, 1)$}
			\State $r \gets v_{t}  - \lfloor v_{t}\rfloor_{{1}/{n}}$ \Comment{$\lfloor v_{t}\rfloor_{{1}/{n}}$ is the largest value $\leq v_t$ that is divisible by $\frac{1}{n}$}
			\State $v_{t + \epsilon} \gets r + v_{t + \epsilon}$ \label{line:massmove}
			\State Set $n \cdot {\lfloor v_{t}\rfloor_{{1}/{n}}}$ values in $\hat{\Lambda}$ to be $t$  \label{line:lambdas}
		\EndFor
		\State return $\hat{\Lambda}$
	\end{algorithmic}	
\end{algorithm}

% To prove Theorem \ref{thm:approxEigs}, we first state a proposition from \cite{Cohen-SteinerKongSohler:2018}, which we prove for completeness below. It shows that, for any equally weighted point-mass distribution, the Wasserstein-1 distance to the point-mass distribution $\tilde{\Lambda}$ is at most the distance to the distribution $\vec{v}$. 

% \begin{proposition}[Lemma 8, \cite{Cohen-SteinerKongSohler:2018}]\label{prop:discreteDistWass}
% 	Let $p$ be a distribution of $n$ equally weighted point-masses on points $p_1, \dots, p_n$ with $p_i \leq p_{i+1}$. Let $v$ be the distribution defined by Algorithm \ref{alg:appxEigs} on the grid $(-1+\epsilon, \dots, 0, \epsilon, \dots, 1)$ with probabilities $\vec{v}$. Let $\tilde{\Lambda}$ be the output of Algorithm \ref{alg:appxEigs} and let $l$ denote the distribution corresponding to having $n$ equally weighted point masses on the points in $\tilde{\Lambda}$. We then have that $W_1(p, l) \leq W_1(p, v)$.
% \end{proposition}
% The proof of Theorem \ref{thm:approxEigs} is relatively simple.
\begin{proof}[Proof of Theorem \ref{thm:approxEigs}]
	Consider the output $\tilde{\Lambda}$ of Algorithm \ref{alg:appxEigs} with input $q$ and $n$. Notice that $W_1(v, q) \leq \epsilon$ by the definition of $v$ and the earthmover's definition of the Wasserstein distance. Hence, by triangle inequality, we have that $W_1(v, s) \leq 2\epsilon$.
	Let $\tilde{v}$ be the vector of masses after the shifting procedure (Line \ref{line:massmove}) in the for-loop of the algorithm. Notice that $\tilde{v}$ is the distribution corresponding to having $n$ equally weighted point-masses on the points in $\tilde{\Lambda}$. Since the procedure in Line \ref{line:massmove} moves at most $1/n$ mass at most $\epsilon$ distance in each iteration, we have $W_1(v, \tilde{v}) \leq \epsilon$ by the earthmover's distance definition of the Wasserstein-1 distance. It follows then that $W_1(\tilde{v}, s) \leq 3 \epsilon$. 

	% By Proposition \ref{prop:discreteDistWass} we know that $\|\Lambda - \tilde{\Lambda}\|_1 = W_1(s, l) \leq W_1(v, s) \leq 2\epsilon$.
	
%	By Fact \ref{eqn:chebyshev_integral} we can compute the distribution $v$ exactly by performing the integration in closed form. Each integral takes $O(1/\epsilon)$ time for a total of $O(1/\epsilon^2)$ time. The final operation, Line \ref{line:lambdas}, takes $O(n)$ time in total. 
\end{proof}

% \begin{proof}[Proof of Proposition \ref{prop:discreteDistWass}]
% By the earthmover's distance definition of the Wasserstein-1 distance, the minimal cost mass-moving scheme that yields distribution $p$ from $v$ corresponds to moving the $\frac{1}{n}$ probability mass from the conditional distribution $f_v(X) | X \in [\frac{i-1}{n}, \frac{i}{n}]$ to location $p_i$. Notice that $W_1(p, \tilde{\Lambda}) = \frac{1}{n}\sum_{i = 1}^n |p_i - \tilde{\lambda}_i|$. It is sufficient then to analyze $|p_i - \tilde{\lambda}_i|$ independently for each $i$. 
% \begin{align*}
% W_1(p, v) &= \frac{1}{n} \sum_{i = 1}^n \sum_{t = 1}^{2/\epsilon} |(-1 + t\epsilon) - p_i| \cdot \ProbBig{f_v(X) = (-1 + t\epsilon) \ | \ X \in \left[\frac{i-1}{n}, \frac{i}{n}\right] }\\
% &\geq \frac{1}{n} \sum_{i = 1}^n \abs{p_i - \sum_{t = 1}^{2/\epsilon} (-1 + t\epsilon) \cdot \ProbBig{f_v(X) = (-1 + t\epsilon) \ | \ X \in \left[\frac{i-1}{n}, \frac{i}{n}\right] }} \\ 
% &= \frac{1}{n} \sum_{i = 1}^n \abs{p_i - \tilde{\lambda}_i} = W_1(p, l).
% \end{align*}
% \end{proof}
We note that there are other options beyond Algorithm \ref{alg:appxEigs} for discretizing a continuous density return by the Jackson damped kernel polynomial method -- i) the optimal discretization of a continuous density on the interval $[-1, 1]$ into $n$ equally-weighted point-masses, and ii) an algorithm by \cite{Cohen-SteinerKongSohler:2018} that can be seen as a combination of Algorithm \ref{alg:appxEigs} and the optimal method. 

\paragraph{Optimal Discretization.} Given the continuous density $q$, consider the discrete density that results from the following procedure:
\begin{enumerate}
	\item Initialize $t = -1$, then repeat the following steps until $t = 1$.
	\item Let $t' \geq t$ be the smallest value such that $\int_{t}^{t'} q(x)dx = \frac{1}{n}$.
	\item Place a point-mass at $\ExpSub{x \sim q}{x \ \vert \ x \in [t, t']}$. I.e. a point-mass is placed in the interval $[t, t']$ at the point given by the conditional distribution of $q$ on the interval. 
	\item Update $t \gets t'$. 
\end{enumerate}  
The values $\tilde{\Lambda} = \tilde{\lambda}_1, \dots, \tilde{\lambda}_n$ given by the point-masses computed by the aforementioned procedure is a optimal discretization of $q$ into $n$ equally-weighted point-masses on $[-1, 1]$ in terms of Wasserstein-1 distance. 
To see why this is the case, consider the first $1/n$ fraction of the mass of the density $q$, i.e. the smallest $t > -1$ such that $\int_{-1}^t q(x)dx = 1/n$. The policy minimizing the earthmover's distance to any $n$ equally-weighted point-wise masses must ``move'' the mass of $q$ on the interval $[-1, t]$ to the point-mass closest to $-1$. Hence, it is sufficient to restrict our attention to the interval $[-1, t]$ when computing the smallest point-mass, i.e. the mass closest to $-1$. Now that we are constrained to looking at the interval $[-1, t]$ one can check that the point-mass minimizing the earthmover's distance to $q$, restricted to $[-1, t]$, is the point-mass at $\ExpSub{x \sim q}{x \ | \ x \in[-1, t]}$. The optimality of the procedure follows from making this argument inductively for all $n$ point-masses.

We note that all steps of the procedure takes roughly $O(n)$ time, although a numerical integration technique or binary search would need to be used to find each $t'$ to sufficiently high accuracy.

A result combining the greedy discretization in Algorithm \ref{alg:appxEigs} and the optimal discretization is given in \cite{Cohen-SteinerKongSohler:2018}. They compute a fractional discretization on an $\epsilon$-spaced grid of $[-1, 1]$, as in Algorithm \ref{alg:appxEigs}, but then compute the eigenvalues using the conditional expectation of every $1/n$ fraction of mass based on the discrete density on the grid.

\section{Positive Polynomial Approximation}\label{appendix:jacksonKernel}
In this section, we introduce Jackson's powerful result from 1912 on the uniform approximation of Lipschitz continuous periodic functions by low-degree trigonometric polynomials \cite{Jackson:1912,Jackson:1930}. This result will directly translate to the result for algebraic polynomials needed to analyze the kernel polynomial method. We start with basic definitions and preliminaries below. 

\subsection{Fourier Series Preliminaries}
\begin{definition}[Fourier Series]
	\label{def:fourier_series}
	A function $f$ with period $2\pi$ that is integrable on the length of that period can be written via the Fourier series:
	\begin{align*}
		f(x) = \frac{\alpha_0}{2} + \sum_{k=1}^\infty \alpha_k\cos(kx) + \beta_k\sin(kx)
	\end{align*}
	where
	\begin{align*}
	\alpha_k &= \frac{1}{\pi}\int_{-\pi}^\pi f(x)\cos(kx)dx & 
	\beta_k &= \frac{1}{\pi}\int_{-\pi}^\pi f(x)\sin(kx)dx.
	\end{align*}
Equivalently we can write $f$ in \emph{exponential form} as:
\begin{align*}
	f(x) = \sum_{k=-\infty}^\infty \hat{f}_k e^{ikx}
\end{align*}
where $i = \sqrt{-1}$, $\hat{f}_0 = \alpha_0/2$, $\hat{f}_k = \hat{f}_{|k|}^*$ for $k < 0$, and for $k > 0$, 
\begin{align*}
	\hat{f}_k &= \frac{1}{2}(\alpha_k - i\beta_k). 
\end{align*}
\end{definition}
If the Fourier series of a periodic function $f$ has $\hat{f}_k = 0$ for $k > N$ (equivalently, $\alpha_k = \beta_k = 0$ for $k > N$), we say that $f$ is a degree $N$ trigonometric polynomial.

In working with Fourier series, we require the two standard convolution theorems:
\begin{claim}[First Convolution Theorem]\label{clm:first_conv_thm}
	Let $f,g$ be integrable $2\pi$-periodic functions with exponential form Fourier series coefficients $[\hat{f}_k]_{k=-\infty}^{\infty}$ and $[\hat{g}_k]_{k=-\infty}^{\infty}$, respectively. Let $h$ be their convolution:
	\begin{align*}
	h(x) = [f*g](x) = \int_{-\pi}^\pi f(u)g(x-u)du.
	\end{align*}
	The exponential form Fourier series coefficients of $h$, $[\hat{h}_k]_{k=-\infty}^{\infty}$, satisfy: 
	\begin{align*}
		\hat{h}_k = 2\pi\cdot \hat{f}_k \hat{g}_k 
	\end{align*}
\end{claim}

\begin{claim}[Second Convolution Theorem]\label{clm:second_conv_thm}
	Let $f,g$ be integrable $2\pi$-periodic functions with exponential form Fourier series coefficients $[\hat{f}_k]_{k=-\infty}^{\infty}$ and $[\hat{g}_k]_{k=-\infty}^{\infty}$, respectively. Let $h$ be their product:
		\begin{align*}
		h(x) = f(x)\cdot g(x).
	\end{align*}
	The exponential form Fourier series coefficients of $h$, $[\hat{h}_k]_{k=-\infty}^{\infty}$, satisfy: 
	\begin{align*}
	\hat{h}_k = \sum_{j=-\infty}^{\infty} \hat{f}_j\cdot \hat{g}_{k-j}
	\end{align*}
	In other words, the Fourier coefficients of $h$ are the discrete convolution of those of $f$ and $g$. 
\end{claim}

\subsection{Jackson's Theorem for Trigonometric Polynomials}
We seek a low-degree trigonometric polynomial $\tilde{f}$ that is a good \emph{uniform} approximation to any sufficiently smooth periodic function $f$. I.e., we want $\|f - \tilde{f}\|_{\infty} < \epsilon$ where $\|z\|_{\infty}$ denotes $\|z\|_{\infty} = \max_{x} z(x)$. 
A natural choice for $\tilde{f}$ is the truncated Fourier series $\sum_{k=-N}^N c_k e^{ikx}$, but this does not lead to good uniform approximation in general. Instead, Jackson showed that better accuracy can be obtained with a Fourier series with \emph{damped coefficients}, which is equivalent to the \emph{convolution} of $f$ with an appropriately chosen ``bump'' function (aka kernel), defined below: 
\begin{definition}[Jackson Kernel]
	For any positive integer $m$, let $b$ be the $2m-2$ degree trigonometric polynomial:
	\begin{align*}
		b = \left(\frac{\sin(mx/2)}{\sin(x/2)}\right)^4 = \sum_{k=-2m+2}^{2m-2} \hat{b}_k e^{ikx},
	\end{align*}
	which has exponential form coefficients $\hat{b}_{-2m+2}, \ldots, \hat{b}_0, \ldots, \hat{b}_{2m-2}$ equal to:
	\begin{align}
		\label{bump_coeffs}
		\hat{b}_{-k} &= \hat{b}_{k} = \sum_{j=-m}^{m-k} (m-|j|)\cdot(m-|j+k|) & \text{for } k&=0, \ldots, 2m-2. 
	\end{align}
\end{definition}
When $m$ is odd it is easy to see that $b$ is a degree $2m-2$ trigonometric polynomial. Specifically, for odd $m$ we have the well known Fourier series of the periodic sinc function $s(x) = \frac{\sin(mx/2)}{\sin(x/2)} = \sum_{k=-(m-1)/2}^{(m-1)/2}e^{ikx}$. We then apply the convolution theorem (Claim \ref{clm:second_conv_thm}) to $s(x)\cdot s(x)$. to see that $s^2(x) = \left(\frac{\sin(mx/2)}{\sin(x/2)}\right)^2$ is an $m-1$ degree trigonometric polynomial with coefficients ${c}_{-k} = {c}_{k} = m - k$. Applying it again to $s^2(x)\cdot s^2(x)$ yields \eqref{bump_coeffs}. For a derivation of \eqref{bump_coeffs} when $m$ is even, we refer the reader to \cite{Jackson:1930} or \cite{Lorentz:1966}. 

 \begin{figure}[h]
 	\centering
 	\includegraphics[width=.45\linewidth]{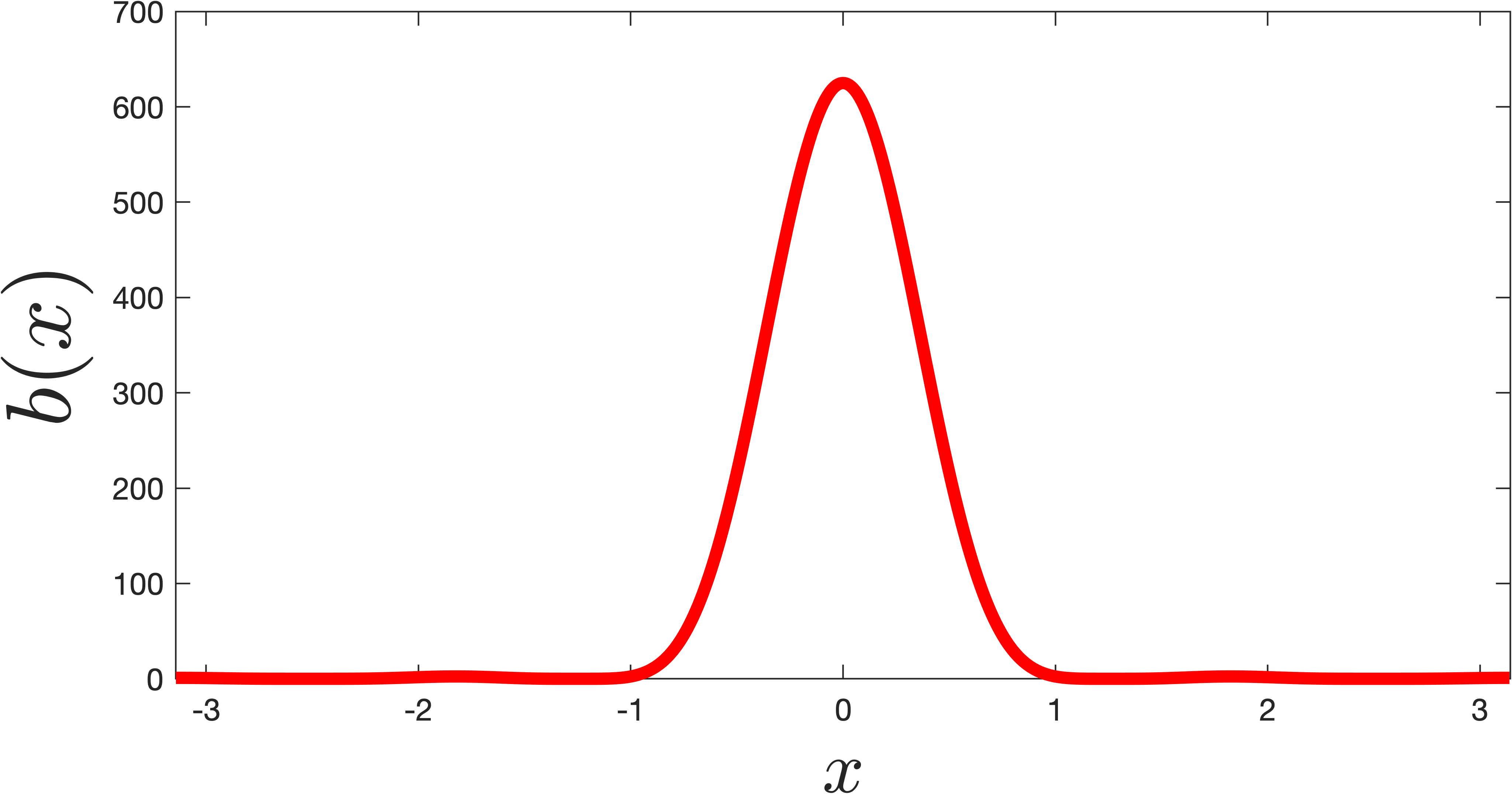}  \hfill
 	\includegraphics[width=.45\linewidth]{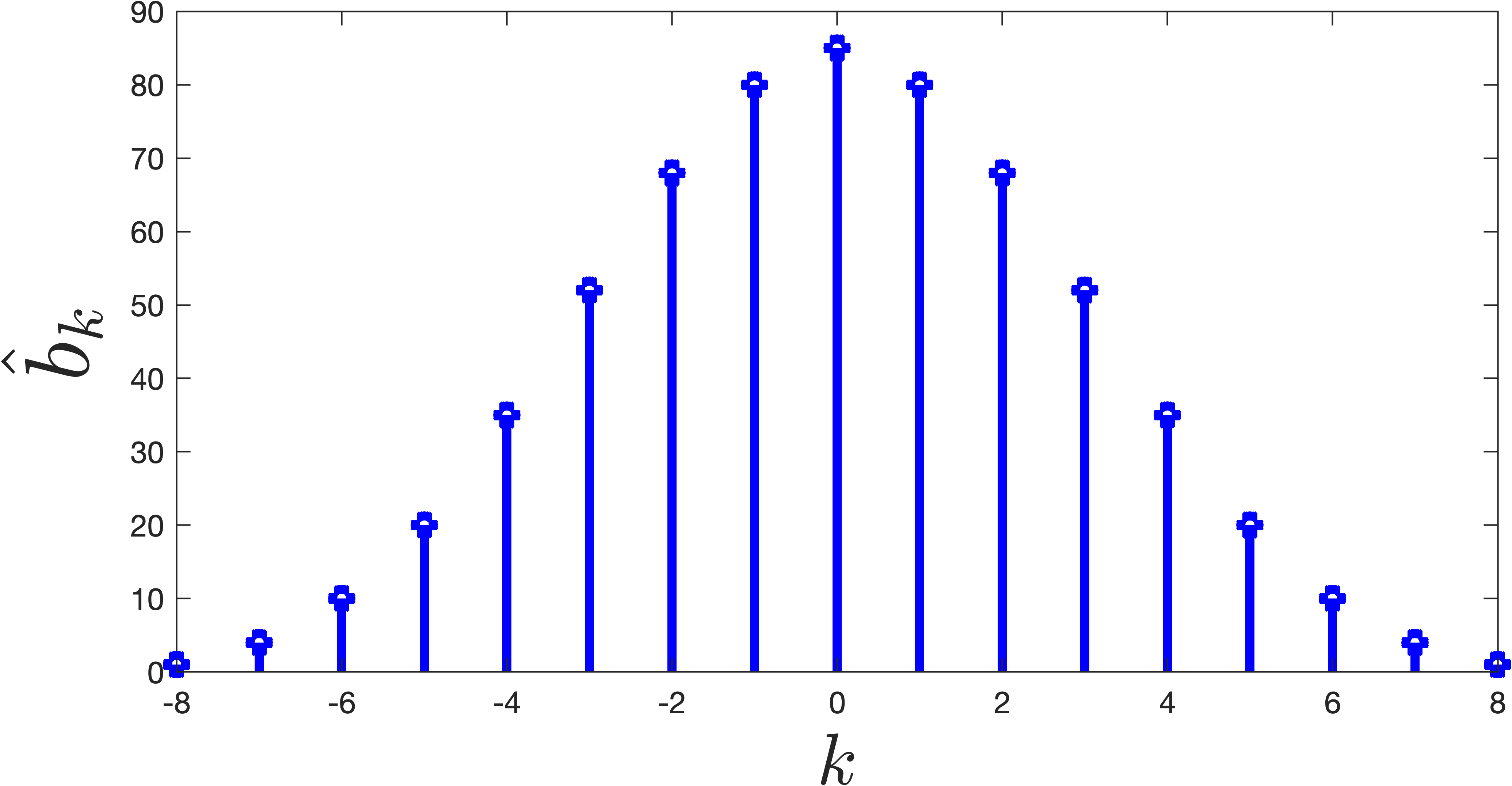} 
 	\caption{Jackson's bump function $b(x)$ for $m =5$, alongside its Fourier series coefficients.}
 	\label{fig:jacksons}
 \end{figure} 

Jackson's main result is as follows. We include a short proof for completeness. 
\begin{theorem}[Jackson \cite{Jackson:1912}, see also \cite{Jackson:1930}]
	\label{thm:jacksons_theorem}
	Let $f$ be a $2\pi$-periodic, Lipschitz continuous function with Lipschitz constant $\lambda$. I.e., $|f(x) - f(y)|\leq \lambda |x-y|$ for all $x,y$.  
	For integer $m$, let $b$ be the bump function from Definition \ref{bump_coeffs}, with $k^\text{th}$ Fourier ceofficients $\hat{b}_{k}$. The function  
	$\tilde{f}(x) = \frac{1}{2\pi \hat{b}_0}\int_{-\pi}^\pi b(u)f(x-u)du$ satisfies:
	\begin{align*}
		\|\tilde{f} - f\|_{\infty} \leq 9\frac{\lambda}{m}.
	\end{align*}
	$\tilde{f}$ is a $2m-2$ degree trigonometric polynomial, and by the convolution theorem, its exponential form Fourier series coefficients are given by $\hat{\tilde{f}}_k = \frac{\hat{b}_k}{\hat{b}_0}\cdot \hat{f}_k$ for $k = -2m+2, \ldots, 2m-2$. 
\end{theorem}
\textbf{Remark.} The function $\tilde{f}$ takes the form of a \emph{damped} truncation of $f$'s Fourier series: $\frac{\hat{b}_0}{\hat{b}_0} = 1$ and $\frac{\hat{b}_k}{\hat{b}_0}$ falls off towards zero as $k \rightarrow 2m-2$. After $2m-2$, the Fourier series coefficients from $f$ are fully truncated to 0. 
%Also note a few other useful properties of $\tilde{f}$: first, when $f$ is a non-negative function, since $b$ is non-negative, $\tilde{f}$ is also non-negative. Second, $\hat{\tilde{f}}_0 = \hat{f}_0$ $\tilde{f}$ and $f$ have the same integral, so when $f$ is a probability density, so if $\tilde{f}$. 
\begin{proof}
	Recalling that $\hat{b}_0 = \frac{1}{2\pi}\int_{-\pi}^\pi b(x)dx$, we have that $\int_{-\pi}^\pi \frac{1}{2\pi\hat{b}_0}b(u)du = 1$, and thus
	\begin{align*}
		|\tilde{f}(x) - f(x)| \leq \int_{-\pi}^\pi \frac{1}{2\pi\hat{b}_0}b(u)\cdot|f(x) - f(x-u)|du.
	\end{align*}
	By our Lipschitz assumption of $f$, we can bound $|f(x) - f(x-u)|\leq \lambda|u|$ and thus have:
%	\aknote{I'm not sure I follow why the chance of limits from $[-\pi, \pi]$ to $[0, \pi]$}
	\begin{align}
		\label{eq:ratio_to_bound}
		\max_{x}|\tilde{f}(x) - f(x)| = \|\tilde{f} - f\|_{\infty} \leq \lambda\cdot\frac{\int_{-\pi}^\pi |u|b(u)du}{2\pi\hat{b}_0} = \lambda\cdot\frac{\int_{0}^\pi u b(u)du}{\int_{0}^\pi b(u) du}.
	\end{align}
In the last equality, we use that $b$ is symmetric about zero.
We have that $2\cdot\sin\left(\frac{u}{2}\right) \leq u \leq \pi\cdot\sin\left(\frac{u}{2}\right)$ for $x\in [0,\pi]$ and thus:
	\begin{align*}
	\int_{0}^\pi u b(u) du \leq \pi^4\int_{0}^\pi u\frac{\sin(mu/2)^4}{u^4} du =   \pi^4m^2\int_{0}^{\pi m} \frac{\sin(v/2)^4}{v^3} dv \leq \pi^4m^2\int_{0}^{\infty} \frac{\sin(v/2)^4}{v^3} dv.
\end{align*}
The last integral evaluates of $\frac{\ln 2}{4}$, so overall we have $\int_{0}^\pi b(u)\cdot u du \leq \frac{\pi^4\ln 2}{4} \cdot m^2$. Moreover we can check that:
\begin{align*}
	\int_{0}^\pi b(u) du = \pi\cdot(\frac{2}{3}m^3 + \frac{1}{3}m) \geq \frac{2\pi}{3}m^3. 
\end{align*}
Plugging into \eqref{eq:ratio_to_bound} we have that:
\begin{align*}
	\|\tilde{f} - f\|_{\infty} \leq 8.06\frac{\lambda}{m}.
\end{align*}
The result follows. We note that the constant above is loose: numerical results suggest the bound can be improved to $\leq \frac{\pi}{2}\frac{\lambda}{m}$. 
\end{proof}

Theorem \ref{thm:jacksons_theorem} translates to a result for \emph{algebraic polynomials} via a standard transformation between Fourier series and Chebyshev series, which we detail below.

\subsection{Jackson's Theorem for Algebraic Polynomials}

\begin{theorem}
	\label{thm:jacksons_theorem_algebraic}
	Let $f\in \FcOne$ be a Lipschitz continuous function on $[-1,1]$ with Lipschitz constant $\lambda$. I.e., $|f(x) - f(y)|\leq \lambda |x-y|$ for all $x,y$. For integer $m$, let $\hat{b}_0,\ldots, \hat{b}_{2m-2}$ be the coefficients from \eqref{bump_coeffs}. Let $c_k =  \iprod{f}{w\cdot \bar{T}_k}$ be the $k^\text{th}$ coefficient in $f$'s Chebyshev polynomial expansion, where $w$ and $\bar{T}_k$ are as defined in Section \ref{sec:Preliminaries}.
	The degree $(2m-2)$ algebraic polynomial 
	\begin{align*}
		\tilde{f}(x) = \sum_{n=0}^{2m-2} \frac{\hat{b}_k}{\hat{b}_0} c_k \cdot \bar{T}_k(x)
	\end{align*}
	satisfies $\|\tilde{f} - f\|_{\infty} \leq 9\frac{\lambda}{m}$. 
%	Furthermore, if $f$ is non-negative on $[-1,1]$, so is $\tilde{f}$. 
\end{theorem}
\begin{proof}
	To translate from the trigonometric case to the algebraic setting, we will use the identity that for all $k$,
	\begin{align}
		\label{eq:cosine_relation}
		T_k(\cos\theta) = \cos(k\theta).
	\end{align}
Consider any function $r\in \FcOne$ with Chebyshev expansion coefficients $c_0, c_1, \ldots$, where $c_k =  \iprod{r}{w\cdot \bar{T}_k}$. Transform $r$ into a periodic function as follows: let $g(\theta) = r(\cos\theta)$ for $\theta \in [-\pi,0]$ and let $h(\theta) = g(-|\theta|)$ for $\theta \in [-\pi,\pi]$. The function $h(\theta)$ is periodic, and also even, so its Fourier series has all coefficients $\beta_1, \beta_2, \ldots$ equal to 0. We thus have that
\begin{align*}
h(\theta) = \sum_{n=0}^\infty \alpha_k\cos(n\theta),
\end{align*}
where
\begin{align*}
\alpha_0 =  \frac{1}{2\pi}\int_{-\pi}^\pi h(\theta)\cos(k\theta)d\theta = \frac{1}{\pi}\int_{-\pi}^0 g(\theta)\cos(k\theta)d\theta
\end{align*}
and, for $n > 0$,
\begin{align*}
\alpha_k =  \frac{1}{\pi}\int_{-\pi}^\pi h(\theta)\cos(k\theta)d\theta = \frac{2}{\pi}\int_{-\pi}^0 g(\theta)\cos(k\theta)d\theta.
\end{align*} 
Using \eqref{eq:cosine_relation} and the fact that $\frac{d}{dx}\cos^{-1}(x) = \frac{1}{\sqrt{1-x^2}}$, we have: 
\begin{align*}
\int_{-\pi}^0 g(\theta)\cos(k\theta)d\theta = \int_{-1}^1 r(x)T_k(x)\frac{1}{\sqrt{1-x^2}}dx.
\end{align*}
We conclude that the Chebyshev coefficients of $r$ are precisely a scaling of the Fourier coefficients of $h$. Specifically, since $\bar{T}_0 = \sqrt{\frac{2}{\pi}}T_0$ and $\bar{T}_k = \sqrt{\frac{1}{\pi}}T_k$, we have:
\begin{align}
	\label{eq:translate}
	\sqrt{\frac{2}{\pi}} c_0 &=  \alpha_0,  & \sqrt{\frac{1}{\pi}} c_k &=  \alpha_k \text{ for } k > 0.
\end{align}
With this fact in hand, Theorem \ref{thm:jacksons_theorem_algebraic} follows almost immediately from  Theorem \ref{thm:jacksons_theorem}. Specifically, given $f \in \FcOne$ with Chebyshev series coefficients $c_0, c_1, \ldots$, we let $g(\theta) = f(\cos\theta)$ and $h(\theta) = g(-|\theta|)$. Let $\alpha_0, \alpha_1, \ldots$ denote $h$'s non-zero Fourier coefficients. Then, let $\tilde{h}$ be the approximation to $h$ given by Theorem \ref{thm:jacksons_theorem}. $\tilde{h}$ is a $2m-2$ degree trigonometric polynomial and is even since $h$ is even and the bump function $b$ is symmetric. Denote $\tilde{h}$'s non-zero Fourier coefficients by $\tilde{\alpha}_0, \ldots, \tilde{\alpha}_{2m-2}$. We have that $\tilde{\alpha}_k = \frac{\hat{b}_k}{\hat{b}_0}\alpha_k$ for $0\leq k \leq 2m-2$. Finally, let $\tilde{f}\in \FcOne$ be defined by $\tilde{f}(\cos(\theta)) = h(-\theta)$. By \eqref{eq:translate}, we have that $\tilde{f}$ is a degree $2m-2$ polynomial and its Chebyshev series coefficients $\tilde{c}_0, \ldots, \tilde{c}_{2m-2}$ are exactly equal to $\frac{\hat{b}_k}{\hat{b}_0}c_k$. 

Moreover, we have $\|f - \tilde{f}\|_\infty = \max_{x\in[-1,1]}|f(x) - \tilde{f}(x)| = \max_{x} |h(x) - \tilde{h}(x)|$. By Theorem \ref{thm:jacksons_theorem} we have $\max_{x} |h(x) - \tilde{h}(x)| < 9\frac{\lambda}{m}$, so we conclude that $\|f - \tilde{f}\|_\infty < 9\frac{\lambda}{m}$.  
%Additionally, if $f$ is non-negative on $[-1,1]$, then $h$ is non-negative, and as discussed earlier, $\tilde{h}$ is also non-negative. So $\tilde{f}$ must be non-negative as well.
\end{proof}

In addition to the main result of Theorem \ref{thm:jacksons_theorem_algebraic}, our SDE algorithm alsos require an additional property of the damped Chebyshev approximation $\tilde{f}$:

\begin{lemma}
	\label{lem:non-negativity}
	For any non-negative function $f\in \FcOne$ (not necessarily Lipschitz), let $\tilde{f}$ be as in Theorem \ref{thm:jacksons_theorem_algebraic}. We have that $\tilde{f}$ is also non-negative on $[-1,1]$.
%	 and $\int_{-1}^1 \tilde{f}(x)dx = \int_{-1}^1 {f}(x)dx$. In particular, if $f$ is a probability density on $[-1,1]$, so is $\tilde{f}$. 
\end{lemma}
\begin{proof}
	Let $h(\theta)$ and $\tilde{h}(\theta)$ be the $2\pi$ perioduc functions as in the proof of Theorem \ref{thm:jacksons_theorem_algebraic}. I.e., $h(\theta) = g(-|\theta|)$ where $g(\theta) = f(\cos\theta)$ and $\tilde{h}$ is the truncated, Jackson-damped approximation to $h$ from  Theorem \ref{thm:jacksons_theorem}. If $f$ is non-negative, then so is $h$, and since $\tilde{h}$ is the convolution of $h$ with a non-negative function, it is non-negative as well. Finally, since $\tilde{f}(\cos(\theta)) = h(-\theta)$, we conclude that $\tilde{f}(x) \geq 0$ for $x\in [-1,1]$.  
\end{proof}

\section{Derivation of Fact \ref{eqn:chebyshev_integral}}\label{appendix:chebyshev_integral}
Let $x = \sin(u)$ then we have that $dx = \cos(u)du$. Substituting the change of variable in the integral and noting the fact that $T_k(\cos \theta ) = \cos (k \theta)$ for $\theta \in [-\pi, \pi]$ gives us that
\begin{align*}
\int_{a}^b \frac{T_k(x)}{\sqrt{1 - x^2}}dx &= \int_{\sin^{-1} a}^{\sin^{-1} b} \frac{\cos(k\cos^{-1}\sin(u))}{\sqrt{1 - \sin^2(u)}} \cos(u)du 
= \int_{\sin^{-1} a}^{\sin^{-1} b} \cos (k(\pi/2 - u))du \\ 
&= \frac{-\sin(k(\pi/2 - u))}{k} \bigg\vert_{\sin^{-1} a}^{\sin^{-1} b} = \frac{-\cos(ku)}{k} \bigg\rvert_{\sin^{-1} a}^{\sin^{-1} b}
\end{align*}
where we used the fact that $\cos^2(u) + \sin^2(u) = 1$ and $\int \cos(u)du = \sin(u) + c$.

\section{Proof of Fact \ref{fact:moment_magnitude_bound}}\label{appx:moment_magnitude_proof}
\begin{proof}
We start by doing a change of variables; set $x = \cos\theta$ and note that $dx = -\sin \theta d\theta$. Substituting this into the expression for $\iprod{f}{w \cdot \Tbar_k}$ and noting that $T_k(\cos \theta) = \cos k\theta$ gives us that 
\begin{align*}
\sqrt{\frac{2}{\pi}}\int_{-1}^1 f(x)\frac{T_k(x)}{\sqrt{1 - x^2}} dx &= \sqrt{\frac{2}{\pi}}\int_{-\pi}^0 -f(\cos \theta)(\cos k\theta) d \theta
\end{align*}
since $\sqrt{1 - \cos^2 \theta} = \sin \theta$ and $dx = -\sin \theta d\theta$. Integrating by parts and noting that $(f(\cos \theta)\int -\cos k\theta d\theta)\vert_{-\pi}^0 = -f(\cos \theta)\frac{\sin k \theta}{k} |_{-\pi}^0 = 0$ gives us that 
\begin{align*}
    \iprod{f}{w \cdot \Tbar_k} &= \sqrt{\frac{2}{\pi}} \int_{-\pi}^0 \frac{\sin k \theta}{k} \ df(\cos \theta).
\end{align*}
We use the definition of the Riemann-Stieltjes integral and let $M  \in \N^+$ be a parameter and $\mathcal{P}_M = \{-\pi = x_0 \leq \dots \leq x_M = 0\}$ be the set of all $M$ intervals partitioning the interval $[-\pi, 0]$. Then for a partition $P \in \mathcal{P}_M$ we denote $\text{norm}(P)$ to be the length of its longest sub-interval. The Riemann-Stieltjes integral $\int_{-\pi}^0 \sin(k\theta)\ df(\cos \theta)$ can be written as
\begin{align*}
    \int_{-\pi}^0 {\sin k \theta} \ df(\cos \theta) &= \lim_{\epsilon \to 0} \sup_{\substack{M, \ P \in \mathcal{P}_M \\ \text{s.t.} \text{norm}(P) \leq \epsilon}} \sum_{i =0}^{m-1} (f(\cos x_{i+1}) - f(\cos x_i) )  \sin kx_{i}.
\end{align*}
Since $f(x) \in \text{lip}_1$ and $|\sin k \theta | \leq 1$ we can bound the magnitude of the above summation as
\begin{align*}
    \left \vert\sum_{i=0}^{m-1} (f(\cos x_{i+1}) - f(\cos x_i) )  \sin kx_{i} \right \vert \leq \sum_{i=0}^{m-1} \lambda|\cos x_{i+1} - \cos x_i| \leq 2. 
\end{align*}
The last inequality follows from the fact that $\cos(\theta)$ is $1$-Lipschitz. Putting these bounds together gives us that $|\iprod{f}{w \cdot \Tbar_k}| \leq 2\lambda/k$.
\end{proof}

\end{document}